\documentclass[11pt]{article}
\usepackage[utf8]{inputenc}

\usepackage[utf8]{inputenc}
\usepackage[T1]{fontenc}

\usepackage{geometry}
\usepackage[bottom]{footmisc}
\usepackage{todonotes}
\usepackage{xspace}
\usepackage{array}
\usepackage{verbatim}
\usepackage{float}
\usepackage{multirow}
\usepackage{comment}

\usepackage{nicefrac}
\usepackage{algorithm}
\usepackage{algorithmicx}
\usepackage[noend]{algpseudocode}

\usepackage{amssymb}
\usepackage{amsmath}
\usepackage{amsthm}
\usepackage{dsfont}
\usepackage{footnote}

\usepackage{xcolor}
\usepackage{nameref}
\definecolor{ForestGreen}{rgb}{0.1333,0.5451,0.1333}
\definecolor{DarkRed}{rgb}{0.65,0,0}
\definecolor{Red}{rgb}{1,0,0}
\usepackage[linktocpage=true,
pagebackref=true,colorlinks,
linkcolor=DarkRed,citecolor=ForestGreen,
bookmarks,bookmarksopen,bookmarksnumbered]
{hyperref}

\usepackage{cleveref}

\usepackage{thm-restate}
\usepackage[tableposition=bottom]{caption}
\usepackage[all]{hypcap}
\usepackage{subcaption}
\usepackage{framed}
\usepackage[framemethod=tikz]{mdframed}
\usepackage[skins]{tcolorbox}
\usepackage{enumerate}
\makeatletter
\g@addto@macro{\maketitle}{\@thanks}
\makeatother

\newtheorem{thm}{Theorem}[section]
\newtheorem{cor}[thm]{Corollary}
\newtheorem{prop}[thm]{Proposition}

\newtheorem{lem}[thm]{Lemma}
\newtheorem{definition}[thm]{Definition}
\newtheorem{obs}[thm]{Observation}
\newtheorem{claim}[thm]{Claim}

\renewcommand{\algorithmiccomment}[1]{\bgroup\hfill$\rhd$~#1\egroup}

\newcounter{note}[section]

\newenvironment{wrapper}[1]
{
	\begin{center}
		\begin{minipage}{\linewidth}
			\begin{mdframed}[hidealllines=true, backgroundcolor=gray!20, leftmargin=0cm,innerleftmargin=0.4cm,innerrightmargin=0.4cm,innertopmargin=0.4cm,innerbottommargin=0.4cm,roundcorner=10pt]
				#1}
			{\end{mdframed}
		\end{minipage}
	\end{center}
} 

\renewcommand{\paragraph}[1]{\medskip\noindent\textbf{#1}}

\ifdefined\ShowComment
    \def\sayan#1{\marginpar{$\leftarrow$\fbox{Sa}}\footnote{$\Rightarrow$~{\sf\textcolor{blue}{#1 --Sayan}}}}   \def\shay#1{\marginpar{$\leftarrow$\fbox{Sh}}\footnote{$\Rightarrow$~{\sf\textcolor{purple}{#1 --Shay}}}}  \def\martin#1{\marginpar{$\leftarrow$\fbox{M}}\footnote{$\Rightarrow$~{\sf\textcolor{ForestGreen}{#1 --Martin}}}}
    
\else
    \def\sayan#1{}    
    \def\shay#1{}
    \def\martin#1{}
\fi 

\usepackage{etoolbox}
\usepackage{tikz}
\usetikzlibrary{tikzmark}
\usetikzlibrary{calc}

\errorcontextlines\maxdimen

\newcommand{\ALGtikzmarkcolor}{black}
\newcommand{\ALGtikzmarkextraindent}{4pt}
\newcommand{\ALGtikzmarkverticaloffsetstart}{-.5ex}
\newcommand{\ALGtikzmarkverticaloffsetend}{-.5ex}
\makeatletter
\newcounter{ALG@tikzmark@tempcnta}

\newcommand\ALG@tikzmark@start{%
	\global\let\ALG@tikzmark@last\ALG@tikzmark@starttext%
	\expandafter\edef\csname ALG@tikzmark@\theALG@nested\endcsname{\theALG@tikzmark@tempcnta}%
	\tikzmark{ALG@tikzmark@start@\csname ALG@tikzmark@\theALG@nested\endcsname}%
	\addtocounter{ALG@tikzmark@tempcnta}{1}%
}

\def\ALG@tikzmark@starttext{start}
\newcommand\ALG@tikzmark@end{%
	\ifx\ALG@tikzmark@last\ALG@tikzmark@starttext
	\else
	\tikzmark{ALG@tikzmark@end@\csname ALG@tikzmark@\theALG@nested\endcsname}%
	\tikz[overlay,remember picture] \draw[\ALGtikzmarkcolor] let \p{S}=($(pic cs:ALG@tikzmark@start@\csname ALG@tikzmark@\theALG@nested\endcsname)+(\ALGtikzmarkextraindent,\ALGtikzmarkverticaloffsetstart)$), \p{E}=($(pic cs:ALG@tikzmark@end@\csname ALG@tikzmark@\theALG@nested\endcsname)+(\ALGtikzmarkextraindent,\ALGtikzmarkverticaloffsetend)$) in (\x{S},\y{S})--(\x{S},\y{E});%
	\fi
	\gdef\ALG@tikzmark@last{end}%
}

\apptocmd{\ALG@beginblock}{\ALG@tikzmark@start}{}{\errmessage{failed to patch}}
\pretocmd{\ALG@endblock}{\ALG@tikzmark@end}{}{\errmessage{failed to patch}}
\makeatother
\global\long\def\eps{\epsilon}

\global\long\def\polylog{\mathrm{polylog}}

\renewcommand{\eps}{\epsilon}%
\newcommand{\p}{\textsc{P}}%

\renewcommand{\polylog}{\mathrm{polylog}}%

\title{Nibbling at Long Cycles: \\ Dynamic (and Static) Edge Coloring in Optimal Time}

\author{Sayan Bhattacharya$^1$ \and Mart\'{i}n Costa$^1$ \and Nadav Panski$^2$ \and Shay Solomon$^2$}

\date{
    $^1$University of Warwick\\
    $^2$Tel Aviv University
}

\begin{document}

\maketitle

\pagenumbering{gobble}

\begin{abstract}
We consider the problem of maintaining a  $(1+\epsilon)\Delta$-edge coloring in a dynamic graph $G$ with $n$ nodes and maximum degree at most $\Delta$. 
The state-of-the-art update time is $O_\eps(\polylog(n))$, by Duan, He and Zhang [SODA'19] and by Christiansen [STOC'23], and more precisely $O(\log^7 n/\epsilon^2)$, where $\Delta = \Omega(\log^2 n / \epsilon^2)$.

The following natural question arises: What is the best possible update time of an algorithm for this task? More specifically, {\bf can we bring it all the way down to some constant} (for constant $\epsilon$)? This question coincides with the {\em static} time barrier for the problem: Even for $(2\Delta-1)$-coloring, 
there is only a naive
$O(m \log \Delta)$-time algorithm.

We answer this fundamental question in the affirmative, by presenting a dynamic $(1+\epsilon)\Delta$-edge coloring algorithm with $O(\log^4 (1/\epsilon)/\epsilon^9)$ update time, provided $\Delta = \Omega_\epsilon(\polylog(n))$. As a corollary, we also get the first linear time (for constant $\epsilon$) {\em static} algorithm for $(1+\epsilon)\Delta$-edge coloring; in particular, we achieve a running time of $O(m \log (1/\epsilon)/\epsilon^2)$.

We obtain our results by carefully combining a variant of the {\sc Nibble} algorithm from Bhattacharya, Grandoni and Wajc [SODA'21] with the subsampling technique of Kulkarni, Liu, Sah, Sawhney and Tarnawski [STOC'22]. 

\end{abstract}

\pagebreak



\pagebreak

\tableofcontents

\pagebreak

\pagenumbering{arabic}

\section{Introduction}
\label{sec:intro}

Given an $n$-node graph $G = (V, E)$ and a {\em palette} of {\em colors} $\mathcal{C}$, a {\em (proper) edge coloring} $\chi : E \rightarrow \mathcal{C}$ assigns a color to each edge of $G$ while ensuring that no two neighboring edges receive the same color. If $|\mathcal{C}| = \lambda$, then we say that $\chi$ is a $\lambda$-edge coloring of $G$.  It is easy to observe that we need at least $\Delta$ colors for any proper edge coloring, where $\Delta$ is the maximum degree in the input graph, and a textbook theorem by Vizing~\cite{Vizing} guarantees that any graph admits a $(\Delta+1)$-edge coloring. In contrast, a simple greedy algorithm, which scans the edges of $G$ in any arbitrary order and assigns a ``free'' color to each edge during the scan, returns a $(2\Delta-1)$-edge coloring of $G$.

In this paper, we focus on the problem of maintaining an edge coloring in a dynamic setting. Here, the input graph $G$ undergoes a sequence of {\em updates} (edge insertions/deletions), but its maximum degree always remains at most a known parameter $\Delta$. We wish to maintain a $\lambda$-edge coloring in this dynamic graph $G$, for as small a value of $\lambda$ as possible. We also wish to ensure that the {\em update time} of our algorithm, which is the time it takes to handle an update in $G$, remains small.
Thus, the key challenge is to understand the trade off between the number of colors needed by an algorithm and its update time.  We now summarize the state-of-the-art for this problem.

We can maintain a $(2\Delta-1)$-edge coloring in $O(\log \Delta)$ update time~\cite{DBLP:conf/iccS/BarenboimM17,BhattacharyaCHN18}, which essentially requires dynamizing the greedy algorithm using a variant of binary search. If we wish to move beyond the greedy threshold of $(2\Delta-1)$, then there are two further results. We know how to maintain a $(1+\epsilon)\Delta$-edge coloring in $O(\log^7 n/\epsilon^2)$ update time when $\Delta = \Omega(\log^2 n/\epsilon^2)$~\cite{DuanHZ19},  and a $(1+\epsilon)\Delta$-edge coloring in $O(\log^9 n \log^6 \Delta/\epsilon^6)$ update time with no restrictions on $\Delta$~\cite{Christiansen22}. Given that both these two results incur a {\em large} polylogarithmic factor in their update times, it is very natural to ask the following question, which we address in this paper. 

\begin{wrapper}
Consider any arbitrarily small constant $\epsilon \in (0, 1)$, and suppose that we wish to maintain a $(1+\epsilon)\Delta$-edge coloring in a dynamic graph. Then what is the best possible update time of any dynamic algorithm for this task? Can we bring this update time down all the way to $O(1)$?
\end{wrapper}

\noindent
As we will shortly see,  this paper answers the above question in the affirmative.  Before stating our formal result, we outline the major obstacles that we need to overcome to achieve this goal.

\subsection{Perspective: The Quest for Constant Update Time}
\label{sec:perspective}
Achieving constant update times for fundamental problems is an important research agenda within dynamic algorithms~\cite{AssadiS21,BhattacharyaCH17,BhattacharyaGKL22,BhattacharyaGM17,BhattacharyaHNW21,BhattacharyaK19,Henzinger020,PelegS16,Solomon16,SolomonW18}. There are two major considerations that underpin this research agenda. (i) A constant update time algorithm rules out the possibility of obtaining a  (cell-probe) lower bound for the concerned problem~\cite{Larsen12,PatrascuD06}.  (ii) It immediately implies a {\em linear time} algorithm for the concerned problem in the static setting,\footnote{We take the (static) input graph, and feed it to the dynamic algorithm by inserting its edges one at a time.} and thus aligns with what is essentially the best possible static guarantee. With this backdrop, we encounter a significant hurdle at the very beginning of our quest, since currently there does not even exist a $O_{\epsilon}(m)$ time {\em static} algorithm for $(1+\epsilon)\Delta$-edge coloring.\footnote{Throughout this paper, the notation $O_{\epsilon}(.)$ hides $\text{poly}(1/\epsilon)$ factors.}  To elaborate on this further, we now review the state-of-the-art for static edge coloring algorithms.

The greedy algorithm can easily be implemented by means of a binary search, which gives us $(2\Delta-1)$-edge coloring in $O(m \log \Delta)$ time, where $m$ is the number of edges in the input graph. Beyond the greedy threshold, it is known how to compute a $(\Delta+1)$-edge coloring in $O(m \sqrt{n})$ time~\cite{Gabow85,DBLP:journals/corr/abs-1907-03201}, and a $\left(\Delta+ \Delta^{0.5+\epsilon}\right)$-edge coloring in $\tilde{O}_{\epsilon}(m)$ time.\footnote{The $\tilde{O}(.)$ notation hides $\polylog(n)$ factors.} Finally, the two dynamic algorithms~\cite{Christiansen22,DuanHZ19} immediately imply static algorithms for $(1+\epsilon)\Delta$-edge coloring, with running times $O(m \log^7 n/\epsilon^2)$ and $O(m \log^9 n \log^6 \Delta/\epsilon^6)$ respectively. In addition,~\cite{DuanHZ19} also obtain a $(1+\epsilon)\Delta$-edge coloring algorithm with $O(m \log^6 n/\epsilon^2)$ running time, when $\Delta = \Omega(\log n/\epsilon)$.

In fact, if we insist upon getting an exact linear (i.e., $O(m)$)  running time, and subject to this constraint try to minimize the number of colors being used, then the only game in town happens to be a very simple, folklore algorithm that gives us $(2+\epsilon)\Delta$-coloring. This algorithm can also be dynamized to get $O(1/\epsilon)$ update time, as explained below.

 \medskip
 \noindent 
{\bf A folklore (randomized) dynamic algorithm.} Suppose that we have a palette $\mathcal{C}$ of $(2+\epsilon)\Delta$ colors, and we are currently maintaining a proper coloring $\chi : E \rightarrow \mathcal{C}$ of the input graph $G = (V, E)$. For each node $v \in V$, we maintain the set $\overline{P(v)} := \{ c \in \mathcal{C} : \exists (u, v) \in E \text{ s.t. } \chi(u, v) = c\}$ of colors that are currently assigned to the edges incident on $v$, as a hash table.  If an edge $e$ gets deleted from $G$, then we don't do anything else as the coloring $\chi$ continues to remain proper.  In contrast, if  an edge $(u, v)$ gets inserted into $G$, then we keep sampling colors u.a.r.~from $\mathcal{C}$ until we find a {\em free} color $c \in \mathcal{C} \setminus \left( \overline{P(u)} \cup \overline{P(v)} \right)$ for this edge, and then we set $\chi(u, v) := c$ and update the hash tables $\overline{P(u)},  
 \overline{P(v)}$, which takes $O(1)$ expected time. Note that $\left| \overline{P(x)} \right| \leq \Delta$ for each endpoint $x \in \{u, v\}$, and hence there are  at least $(2+\epsilon)\Delta - 2 \Delta = \epsilon \Delta$ free colors for $(u, v)$ when the edge gets inserted. Accordingly, in expectation we need to sample at most $|\mathcal{C}|/(\epsilon \Delta) = O(1/\epsilon)$ colors from $\mathcal{C}$ until we find a free color for $(u, v)$. Furthermore, for each sampled color $c'$, using the hash tables $\overline{P(u)}, \overline{P(v)}$ we can determine in $O(1)$ expected time whether or not $c'$ is free.  Putting everything together, this leads to a dynamic $(2+\epsilon)\Delta$-edge coloring algorithm with $O(1/\epsilon)$ expected update time, which can easily be converted into a static $(2+\epsilon)\Delta$-edge coloring algorithm with $O(m/\epsilon)$ expected run time. 

\medskip
\noindent {\bf Existing barriers.}
At this point, we revisit the  state-of-the-art on dynamic $(1+\epsilon)\Delta$-edge coloring, and explain the challenges behind extending the known techniques to obtain constant update time. 

\medskip
\noindent (I) The two known dynamic algorithms for $(1+\epsilon)\Delta$-edge coloring~\cite{DuanHZ19,Christiansen22} are both analyzed in a {\em memory-less} manner. Specifically, they assume that we start with  any arbitrary, adversarially  chosen $(1+\epsilon)\Delta$-edge coloring in the current graph $G$, and then show how to modify that coloring (to ensure that it remains proper) in polylogarithmic  time after the insertion/deletion of an edge. A lower bound construction from~\cite{ChangHLPU18}, however, implies that any such memory-less analysis must necessarily imply an update time of $\Omega(\log(\epsilon n)/\epsilon)$. 

\medskip
\noindent (II) There is a weaker version of the dynamic edge coloring problem, where we care about the {\em recourse} (as opposed to update time) of the maintained solution, which basically equals the number of changes the algorithm makes to the coloring after an update.   There exists a $(1+\epsilon)\Delta$-edge coloring algorithm with $O_{\epsilon}(1)$ recourse~\cite{BhattacharyaGW21}, based on the {\sc Nibble} method (see~\Cref{sec:result} for more details) that was first used in the context of edge coloring in the distributed setting~\cite{DubhashiGP98}. It seems very difficult to implement the algorithm of~\cite{BhattacharyaGW21} using $O_{\epsilon}(1)$ update time data structures, for two reasons.  First, the~\cite{BhattacharyaGW21}
 dynamic algorithm needs to resample the color $c$ of an edge $e = (u, v)$ when its palette $P(e) = \mathcal{C} \setminus \left(\overline{P(u)} \cup \overline{P(v)} \right)$  changes by a small amount, even if $c$ continues to be part of $P(e)$.  It is not at all clear how to implement this resampling efficiently, i.e., in constant update time. 
  Second, and more fundamentally,  all existing Nibble method-based algorithms require some form of {\em regularization gadget}, since the inductive approach that is used to analyze these algorithms does not work on graphs that are not near-regular. Implementing this gadget requires $\Omega(n \Delta)$ running time in the static setting, and $\Omega(n \Delta)$ preprocessing time in the dynamic setting.

\subsection{Our Results}
\label{sec:result}

We are now ready to present our results. Towards this end, we define the following parameter $$\Delta^{\star} := \left(\log n/\epsilon^4\right)^{\Theta((1/\epsilon) \log (1/\epsilon))}.$$ Recall that $n$ and $m$ respectively denote the number of nodes and edges in the input graph $G = (V, E)$, and let $\Delta$ be an upper bound on the maximum degree of $G$.

\begin{thm}
\label{thm:static:main}
In the static setting, we can compute a $(1+\epsilon)\Delta$-edge coloring   in the input graph $G$ in  $O(m \log (1/\epsilon)/\epsilon^2)$ time w.h.p., provided $\Delta \geq \Delta^{\star}$.
\end{thm}

We then extend our algorithm to the dynamic setting, and derive the theorem below. 

\begin{thm}
\label{thm:dynamic:main}
We can maintain a $(1+\epsilon)\Delta$-edge coloring in a dynamic graph $G$ in  $O(\log^4(1/\epsilon)/\epsilon^9)$ expected worst-case update time (against an oblivious adversary), provided $\Delta \geq \Delta^{\star}$. 
\end{thm}

We can convert the expected worst-case update time bound of~\Cref{thm:dynamic:main} into a high probability amortized update time guarantee, for polynomially long update sequences (see~\Cref{thm:main 3}).

\subsection{Our Technique} 
\label{sec:technique}
A one-sentence summary of our approach is that we combine the {\sc Nibble} algorithm with the subsampling technique used in a paper by~\cite{KulkarniLSST22} in the context of online edge coloring, and then always maintain the precise output of the resulting static algorithm as the input graph undergoes edge insertions/deletions. We now explain this idea in more detail.

\medskip
\noindent {\bf Our static algorithm.}
We start by summarizing the {\sc Nibble} algorithm (see~\Cref{sec:main:static:nibble}). It runs in $T := \lfloor (1/\epsilon) \log (1/\epsilon) \rfloor$ rounds, on the input graph $G = (V, E)$ with a palette $\mathcal{C}$ of $(1+\epsilon)\Delta$ colors. At the start of round $i \in [T]$, let $P_i(v)$ denote the palette of a node $v \in V$, which consists of all the colors that have {\em not} yet been (tentatively) assigned to any edge incident on $v$. In round $i$, each uncolored edge $e$ {\em selects} itself independently with probability $\epsilon$. Next, every selected edge $e = (u, v)$ picks a tentative color $\tilde{\chi}(e)$ independently and u.a.r.~from its palette $P_i(u) \cap P_i(v)$. At the end of $T$ rounds, we collect all the {\em failed} edges $F \subseteq E$; these are the edges that were either not selected during any of the $T$ rounds and hence did not receive any tentative color, or received a tentative color which conflicts with one of its neighbors. We now color the subgraph $G_F := (V, F)$, using the folklore algorithm and an extra palette of $O(\Delta(G_F))$ colors that is mutually disjoint with $\mathcal{C}$.  This, combined with the tentative colors assigned to the edges in $E \setminus F$, gives us a proper $(1+\epsilon)\Delta + O(\Delta(G_F))$-coloring of $G$. The main challenge now is to show that $\Delta(G_F) = O(\epsilon \Delta)$ w.h.p., for that would give us a $(1+O(\epsilon))\Delta$-coloring of $G$. 

Next, we observe that we do not need any regularizing gadget to analyze the {\sc Nibble} algorithm, {\em if the input graph $G$ is a forest} (see~\Cref{sec:main:static:trees}).  This is primarily because under such a scenario, just before we pick a tentative color for an edge $(u, v) \in E$ in some round $i \in [T]$, the palettes $P_i(u)$ and $P_i(v)$  are mutually independent.  This observation makes it easy to obtain a $O_{\epsilon}(m)$ time implementation of the {\sc Nibble} algorithm for $(1+\epsilon)\Delta$-edge coloring.  We essentially use the same hash table data structures which allow us to efficiently implement the folklore algorithm (see~\Cref{sec:perspective}), along with the fact that w.h.p., at the start of each round $i \in [T]$, the palette $P_i(u, v) := P_i(u) \cap P_i(v)$ of each uncolored edge is of size 
$\Omega(\epsilon^2\Delta)$ (see \Cref{cor:edge-palette}).

At this point, we move on to the general case where the input graph $G$ might contain cycles (see~\Cref{sec:main:static:general}). Here, we first observe that the palette $P_i(v)$ of a node $v$ for a round $i \in [T]$ depends only on the $i$-hop neighborhood of $v$.  Say that a node $v$ is {\em good} in $G$ if its $(T+1)$-hop neighborhood in $G$ does not contain any cycle, and {\em bad} otherwise. Since the {\sc Nibble} algorithm runs for only $T$ rounds,  we can simply pretend that the input graph is a forest while analyzing what the algorithm does to a good node and all its incident edges (see~\Cref{lem:trees:main:good}). In particular, we can show that w.h.p.~every good node will have degree at most $O(\epsilon \Delta)$ in $G_F$.  

We now combine the previous observations with a subsampling technique~\cite{KulkarniLSST22} (see~\Cref{sec:main:static:final}). The basic idea is simple. Fix two parameters $\gamma := 1/(30 T)$ and $\Delta' := \Delta^{\gamma}$, and set $\eta := \Delta/\Delta'$. Next, partition the input graph $G$ into $\eta$ subgraphs $\mathcal{G}_1, \ldots, \mathcal{G}_{\eta}$, by placing each edge $e \in E$ independently and u.a.r.~in one of the subgraphs $\mathcal{G}_1, \ldots, \mathcal{G}_{\eta}$. This partition happens to satisfy the following two properties. (I) W.h.p.~$\Delta(\mathcal{G}_j) \leq (1+\epsilon)\Delta'$ for all $j \in [\eta]$. (II) Say that an edge $e = (u, v) \in E$ is {\em problematic} iff either $u$ or $v$ is a bad node in $\mathcal{G}_j$, where $j \in [\eta]$ is the unique index such that $e \in \mathcal{G}_j$. Let $E^{\star} \subseteq E$ be the set of all problematic edges, and let $G^{\star} := (V,E^{\star})$.  Then $\Delta(G^{\star}) = O(\epsilon \Delta)$ w.h.p. Armed with these two observations,  our final algorithm on general graphs works as follows. We compute the partition of the input graph $G$ into the subgraphs $\mathcal{G}_1, \ldots, \mathcal{G}_{\eta}$. For each $j \in [\eta]$, we run the {\sc Nibble} algorithm on $\mathcal{G}_j$ in an attempt to color it with a (distinct) palette $\mathcal{C}_j$ of $(1+O(\epsilon))\Delta'$ colors. Overall, this requires $\eta \cdot (1+O(\epsilon))\Delta' = (1+O(\epsilon))\Delta$ colors. At the end of this step, we are left with two types of failed edges that could not be properly colored: the ones that are problematic, and the ones that are not. Because of the locality of the {\sc Nibble} method, for each $j \in [\eta]$, w.h.p.~each node $v \in V$ is incident on at most $O(\epsilon \Delta')$ non-problematic failed edges in $\mathcal{G}_j$. Thus, w.h.p.~the maximum degree of the subgraph consisting of all non-problematic edges over all $j \in [\eta]$ is at most $\eta \cdot O(\epsilon \Delta') = O(\epsilon \Delta)$. Next, by Property II above, the maximum degree of the subgraph consisting of all problematic edges, over all $j \in [\eta]$, is given by $\Delta(G^{\star}) = O(\epsilon \Delta)$. Putting everything together, we infer that at the end of the first step, the subgraph consisting of all failed edges has maximum degree $O(\epsilon \Delta)$. Hence, we can easily color these remaining failed edges, using the folklore algorithm from~\Cref{sec:perspective} and an extra set of $O(\epsilon \Delta)$ colors. This leads to a $(1+O(\epsilon))\Delta$-coloring of the input graph $G$, without any additional overhead in the running time compared to the scenario where $G$ was a forest, because the subsampling step can easily be implemented very efficiently. 

\medskip
\noindent {\bf Our dynamic algorithm.} We dynamize our static algorithm using a very natural approach (see~\Cref{sec:overview:dynamic}). It is not surprising that the subsampling step is relatively easy to dynamize, so here we only focus on highlighting what our dynamic algorithm does when the input graph remains a forest. Essentially, our dynamic algorithm hinges upon two main observations. 

\medskip
\noindent (I) When an edge $e$ gets inserted, we might as well fix an index $i_e \in [T+1]$ by sampling $i_e$ from a capped geometric distribution with probability $\epsilon$,\footnote{Thus, we have $\Pr[i_e = i] = (1-\epsilon)^{i-1}\epsilon$ for all $i \in [T]$ and $\Pr[i_e = T+1] = (1-\epsilon)^T$.} and we can also fix an infinite length color-sequence $c_e$, such that for each $\ell \in \mathbb{Z}^{+}$ the $\ell^{th}$ entry $c_e(\ell)$ in this sequence is a color sampled independently and u.a.r.~from the input palette $\mathcal{C}$. These rounds $\{i_e\}_e$ and color-sequences $\{c_e\}_e$ uniquely determine the output of the {\sc Nibble} algorithm on a given input graph $G = (V, E)$. Specifically, each edge $e \in E$  selects itself in round $i_e$, and then identifies the smallest integer $\ell \in \mathbb{Z}^+$ such that $c_e(\ell) \in P_i(e)$, and sets $\tilde{\chi}(e) := c_e(\ell)$.\footnote{Note that this is equivalent to sampling a color $\tilde{\chi}(e)$ from $\mathcal{C}$ u.a.r.} Throughout the sequence of updates, we simply maintain the output of this static algorithm w.r.t.~the indices $\{i_e\}_e$ and color-sequences $\{c_e\}_e$. We show that this natural approach itself suffices to guarantee an expected worst-case recourse of $O_{\epsilon}(1)$ (see~\Cref{sec:main:static:recourse}), and is in sharp contrast with the algorithm of~\cite{BhattacharyaGW21} which required repeated resampling of colors. 

\medskip
\noindent 
(II) We need one additional insight to implement our low-recourse algorithm in $O_{\epsilon}(1)$ update time. Specifically, we {\em truncate} the  color-sequences $\{c_e\}$ at length $K := \Theta((1/\epsilon^2) \log (1/\epsilon))$, i.e., for each edge $e$, we stop constructing the sequence $c_e$ after sampling the first $K$ colors $c_e(1), \ldots, c_e(K)$. The intuition behind why everything still works is as follows. W.h.p., we know that $|P_{i_e}(e)| = \Omega(\epsilon^2 \Delta)$. Thus, conditioned on this event, at least one of the first $K$ colors in $c_e$ should appear in $P_i(e)$ with probability at least $1-\epsilon$. In other words, the resulting algorithm is equivalent to the following process: We throw away each edge $e \in E$ with probability $\epsilon$, and we run the {\sc Nibble} algorithm on the surviving edges. The subgraph consisting of the edges that get thrown away has maximum degree $O(\epsilon \Delta)$ w.h.p., and hence we can separately color this subgraph using the folklore algorithm from~\Cref{sec:perspective}. It turns out that we can develop data structures that support the implementation of this modified algorithm in $O_{\epsilon}(1)$ expected worst-case update time (see~\Cref{sec:main:static:updatetime}). The main reason is that the truncation step acts in a way which is reminiscent of {\em palette-sparsification}~\cite{DuanHZ19}. Indeed, now each edge gets assigned a tentative color that comes from a small set of size $K = O_{\epsilon}(1)$. This helps us design a supporting data structure whose expected update time is proportional to the recourse of the algorithm from step (I) above, which, we already know to be $O_{\epsilon}(1)$.

\subsection{Remark on the Lower Bound on $\Delta$} 
In the edge coloring literature, it is common to assume that $\Delta = \Omega(\polylog (n))$~\cite{BhattacharyaGW21,DuanHZ19,KulkarniLSST22}. In this paper, the reason we need the lower bound on $\Delta$ is as follows (see Section~\ref{sec:main:static:final} for details). We partition the input graph $G$ into $\eta = \Delta/\Delta'$ subgraphs $G_1, \ldots, G_{\eta}$, by throwing each edge of $G$ u.a.r.~into one of these subgraphs. We now need to enforce the following two properties. (i) $\Delta'$ needs to be large enough to ensure that we have enough concentration to be able to run the Nibble algorithm on each $G_i$. (ii) $\Delta'$ needs to be small enough to ensure that each $G_i$ is sufficiently ``locally treelike'' for our analysis to go through. In particular, so that the subgraph $G^{\star}$, which consists of all the bad edges, has maximum degree at most $\epsilon \Delta$ (see Claim~\ref{cl:subsampling:bad}), because these bad edges will get separately colored using a greedy algorithm. Now, it so happens that for Property (i) to hold, we need $\Delta' \geq \Omega(\log n/ \epsilon^{4})$, and for Property (ii) to hold, we need $\Delta' \leq \Delta^{\Theta(1/T)}$, where $T = \lfloor (1/\epsilon) \log (1/\epsilon) \rfloor$ is the number of rounds required by the Nibble method. Combining these two inequalities together, we get that $\Delta \geq \left(\log n/\epsilon^4\right)^{\Theta((1/\epsilon) \log (1/\epsilon))}$. We leave it as a challenging open question to improve this lower bound on $\Delta$.

\newcommand{\C}{\mathcal C}

\section{Overview of our Static Algorithm}\label{sec:overview:static}

In this section, we describe how our algorithm (see~\Cref{thm:static:main}) works in the static setting, and present an overview of the key ideas that underpin its analysis. To convey the main intuition behind our framework, here we intentionally explain some of the arguments in an informal/semi-rigorous manner. The complete, formal proofs from this section are deferred to~\Cref{app:static} and~\Cref{sec:app:datastructstatic}.

\medskip
\noindent {\bf Organization.} In~\Cref{sec:main:static:nibble}, we present the {\sc Nibble} algorithm.~\Cref{sec:main:static:trees} explains how to analyze and implement this algorithm when the input graph is a forest.  In~\Cref{sec:main:static:general}, we show that on general graphs, the analysis from~\Cref{sec:main:static:trees} still holds for all those nodes that are {\em not} part of any short cycle. Finally,~\Cref{sec:main:static:final} contains an overview of our final algorithm, which involves combining the {\sc Nibble} method along with a subsampling technique of~\cite{KulkarniLSST22}.\shay{in prel', should define proper coloring, partial coloring, etc.
Also, should define 
$\tilde{\chi}(W)$ over vertex set $W$}

\subsection{The {\sc Nibble} Algorithm}
\label{sec:main:static:nibble}

Fix any input graph  $G = (V, E)$  with $n$ nodes and maximum degree at most $\Delta$, any constant $\epsilon \in (0, 1/10)$, and any {\em palette} $\C$ of $\lceil (1+\epsilon) \Delta \rceil$ colors.\shay{later on we fix $C$
as $[(1+\epsilon)\Delta]$} The {\sc Nibble} algorithm  runs for $T :=  \lfloor(1 / \epsilon) \log(1/\epsilon) \rfloor$ {\em rounds}.  At the start of round $i \in [T]$, we have a subset of edges $E_i \subseteq E$ such that the algorithm has already assigned tentative colors to the remaining edges $E \setminus E_i$. We denote this {\em tentative} partial coloring by $\tilde{\chi} : E \setminus E_i \rightarrow \C \cup \{ \bot \}$, which need not necessarily be proper. For each node $v \in V$, we refer to the set of colors $P_i(v) :=  \C \setminus \tilde \chi(N(v) \setminus E_i)$\shay{Perhaps write $P_i(v) :=  \C \setminus \tilde \chi(N(v)) = \C \setminus \tilde \chi(N(v) \setminus E_i)$. Also, in the pseudo-code, we use $\cap S_{< i}$ instead, perhaps should draw attention to the equivalence} as the {\em palette} of $v$ at the start of round $i$, where $N(v) \subseteq E$ is the set of edges incident on $v$ in $G$. In words, the palette $P_i(v)$ consists of the set of colors that were {\em not} tentatively assigned to any incident edge of $v$ in previous rounds. We  define $P_i(u, v) := P_i(u) \cap P_i(v)$ to be the {\em palette} of any edge $(u, v) \in E_i$ at the start of round $i$.

We start by initializing $E_1 \leftarrow E$, and $P_1(v) \leftarrow \C$  for all $v \in V$. Subsequently, for $i = 1, \ldots, T$, we implement round $i$ as follows. Each edge $e \in E_i$ {\em selects} itself independently with probability $\epsilon$. Let $S_i \subseteq E_i$ be the set of selected edges. Next, in parallel, each edge $e \in S_i$ samples a color $\tilde \chi(e)$ independently and uniformly at random from  $P_i(e)$. For any edge $e \in S_i$, if $P_i(e) = \varnothing$ then we set $\tilde{\chi}(e) \leftarrow \bot$. At this point, we define the collection $F_i \subseteq S_i$ of {\em failed} edges in round $i$. We say that an edge $e = (u, v) \in S_i$ {\em fails} in round $i$ iff either (i) $\tilde{\chi}(e) = \bot$, or (ii) there is a neighboring edge $f \in (N(u) \cup N(v)) \cap S_i$  which was also selected in round $i$ and received the same tentative color as the edge $e$ (i.e., $\tilde{\chi}(e) = \tilde{\chi}(f)$). Let $F_i \subseteq S_i$ denote this collection of failed edges (in round $i$). We now set $E_{i+1} \leftarrow E_i \setminus S_i$  and proceed to the next round $i+1$.

To ease notations, at the end of the last round $T$ we define $F_{T+1} \leftarrow E_{T+1}$, and $\tilde{\chi}(e) \leftarrow \bot$ for all $e \in F_{T+1}$. We let  $F := \bigcup_{i=1}^{T+1} F_i$ denote the set of failed edges across all the rounds. It is easy to check that the tentative coloring $\tilde{\chi}$, when restricted to the edge-set $E \setminus F$, is already proper.

\begin{obs}
\label{obs:main:nibble}
$\tilde{\chi}$ is a proper $(1+\epsilon)\Delta$-edge coloring in the subgraph $G_{E \setminus F} := (V, E \setminus F)$.
\end{obs}

 The pseudocode of this procedure appears in \Cref{alg:nibble}. We now describe some further notations that will be used throughout the rest of this paper.

 \medskip
 \noindent {\bf Notations.} For all  $v \in V$ and  $i \in T$, we define $N_i(v) := N(v) \cap S_i$ to be the set of edges incident on the node $v$ that get selected in round $i$. For each edge $e = (u, v) \in E$, let $N(e) := N(u) \cup N(v)$ denote the set of its neighboring edges.  Furthermore, for any graph $H = (V, E_H)$ and any node $v \in V$, let $\text{deg}_H(v)$ denote the {\em degree} of $v$ in $H$, and let $\Delta(H)$ denote the {\em maximum degree} of any node in $H$. We will sometimes abuse these notations and write $\text{deg}_{E_H}(v)$ and $\Delta(E_H)$ when the node-set $V$ is clear from the context. Finally, given any sequence of sets $A_1, A_2, \ldots $, we will use the shorthands  $A_{<i} := \bigcup_{j < i} A_j$, $A_{\leq i} := \bigcup_{j \leq i} A_j$, $A_{> i} := \bigcup_{j > i} A_j$ and $A_{\geq i} := \bigcup_{j \geq i} A_j$. For instance, this means that $S_{< i} := \bigcup_{j = 1}^{i-1} S_j$ (see~\Cref{shorthand} in~\Cref{alg:nibble}).

\begin{algorithm}
\caption{\textsc{Nibble}$(G = (V, E), \Delta, \epsilon)$}\label{alg:nibble}
\begin{algorithmic}[1]
    \State $\mathcal{C} \leftarrow [(1+\epsilon)\Delta]$, $E_1 \leftarrow E$, and $\tilde \chi(e) \leftarrow \perp$ for all $e \in E$  \label{alg:palette}
    \For{$i = 1,...,T$}
        \State $S_i \leftarrow \varnothing$
        \For{$e \in E_i$}
            \State Add $e$ to $S_i$ independently with probability $\epsilon$
        \EndFor
        \For{$e = (u, v) \in S_i$}
            \State $P_i(e) \leftarrow \mathcal{C} \setminus \tilde \chi(N(e) \cap S_{< i} )$ \label{shorthand}
            \If{$P_i(e) \neq \varnothing$}
                \State Sample $\tilde \chi(e) \sim P_i(e)$ independently and u.a.r
            \EndIf
        \EndFor

        \State $F_i \leftarrow \{ e \in S_i \, | \, \exists f \in N(e) \cap S_i \textrm{ such that } \tilde \chi(f) = \tilde \chi(e)\} \cup  \{ e \in S_i \, | \, \tilde \chi(e) = \perp \} $
        \State$E_{i+1} \leftarrow E_i \setminus S_i$
    \EndFor
    \State $F_{T+1} \leftarrow E_{T+1}$ \label{alg:failed:1}
    \State$ F \leftarrow \bigcup_{i=1}^{T+1} F_i$ \label{alg:failed:2}
    \State \Return {$\tilde{\chi}, F$}
\end{algorithmic}
\end{algorithm}

\medskip
Intuitively, it is easy to see that the {\sc Nibble} algorithm is {\em symmetric} w.r.t.~the palette $\mathcal{C}$, i.e., it does not give preference to one color over another. The lemma below formalizes this intuition, and will be repeatedly invoked during our analysis (we defer its proof to~\Cref{app:symmetry}).

\begin{lem}\label{lem:main:strong sym}
    Fix the random bits used by the {\sc Nibble} algorithm that determine which edges get selected in which rounds. Then for all $u \in V$, $i \in [T]$, $C \subseteq \mathcal C$ and permutations $\pi : \mathcal C \longrightarrow \mathcal C$, we have
    \[\Pr \left[ P_i(u) = C \right] = \Pr \left[ \pi \left( P_i(u) \right) = C\right].\]
\end{lem}

\subsection{Analysis of the {\sc Nibble} Algorithm on Forests}
\label{sec:main:static:trees}

The {\sc Nibble} algorithm returns a proper $(1+\epsilon)\Delta$-edge coloring $\tilde{\chi}$ on the subgraph $G_{E \setminus F} = (V, E \setminus F)$ (see~\Cref{obs:main:nibble}). We now wish to argue that the remaining subgraph $G_F := (V, F)$  has small maximum degree. Specifically, if it so happens that $\Delta(G_F) = O(\epsilon \Delta)$, then we can just color the edges of $G_F$, using the folklore algorithm from~\Cref{sec:perspective} and an extra set of $O(\epsilon \Delta)$ colors.\shay{at this stage it's still unclear why a greedy coloring is easy to do in linear time} This, combined with the coloring $\tilde{\chi}$ of $G_{E\setminus F}$, would give us a $(1+O(\epsilon))\Delta$-coloring of $G$. We now prove this desired upper bound on $\Delta(G_F)$, under the assumption that the input graph $G$ is a forest with sufficiently large maximum degree. The main result in this section is summarized below.

\begin{lem}
\label{lem:trees:main}
Let $G$ be a forest with maximum degree $\Delta \geq \frac{(100 \log n)}{\epsilon^4}$. Then w.h.p.~we have:
$$\text{deg}_{F}(v) = O(\epsilon \Delta) \text{ for all nodes } v \in V.$$
\end{lem}

We start with the following key observation.

\begin{obs}
\label{obs:independence}
Let  $G = (V, E)$ be a forest. Fix the random bits used by the {\sc Nibble} algorithm that determine which edges get selected in which rounds. Consider any  $i \in [T]$,   $u \in V$, and edges $(u, v_1), \ldots, (u, v_k) \in S_i$. Then the palettes $\{P_i(u), P_i(v_1), \ldots, P_i(v_k)\}$ are mutually independent. 
\end{obs}

\begin{proof}(Sketch)
The palettes $P_i(u), P_i(v_1), \ldots, P_i(v_k)$ only depend on what happens during rounds $< i$ in  $G_{<i} := (V, S_{<i})$, and the nodes $u, v_1, \ldots, v_k$ lie in different connected components of $G_{<i}$. 
\end{proof}

In \Cref{lem:trees:bound:1},~\Cref{cor:lem:trees:bound:1} and \Cref{lem:trees:bound:2}, we derive concentration bounds on the sizes of the sets $N_i(v)$ and the palettes of nodes/edges at the start of each round. Later on, we  use these concentration bounds in \Cref{sec:lem:trees:main}, which gives an overview of the proof of~\Cref{lem:trees:main}. Finally, we explain how to efficiently implement the {\sc Nibble} algorithm in linear time in~\Cref{sec:forest:runtime}.

\begin{lem}
\label{lem:trees:bound:1}
Let $\mathcal{Z}$ denote the event which occurs iff we have $|N_i(u)| <  (1+ \epsilon) \cdot \epsilon (1-\epsilon)^{i-1} \Delta$ for all nodes $u \in V$ and all rounds $i \in [T]$. Then the event $\mathcal{Z}$ occurs w.h.p.
\end{lem}

\begin{proof}(Sketch)
Fix any node $u \in V$ and any round $i \in [T]$. Any given edge $(u, v) \in E$ incident on $u$ appears in $S_i$ with probability $\epsilon (1-\epsilon)^{i-1}$. Thus, by linearity of expectation, we have $\mathbb{E} \left[ |N(u) \cap S_i | \right] \leq \epsilon (1-\epsilon)^{i-1}\Delta$. Furthermore, note that each edge decides independently (of all other edges) whether to get selected in round $i$. Since $\Delta \geq \frac{(100 \log n)}{\epsilon^4}$ and $T = \lfloor (1/\epsilon) \log (1/\epsilon) \rfloor$, the lemma now follows from an application of Chernoff bound, and a union bound over all $u \in V, i \in [T]$.
\end{proof}

\begin{cor}
\label{cor:lem:trees:bound:1}
Conditioned on  $\mathcal{Z}$, we always have $|P_i(u)| > (1+\epsilon) (1-\epsilon)^{i-1}\Delta$ for all $u \in V, i \in [T]$.
\end{cor}

\begin{proof}
Fix any node $u \in V$ and any round $i \in [T]$. Conditioned on the event $\mathcal{Z}$, we have:
\begin{equation}
\label{eq:cor:lem:trees:bound:1}
\left| \bigcup_{j=1}^{i-1} N_j(u) \right| <  \sum_{j=1}^{i-1} (1+\epsilon)  \cdot \epsilon (1-\epsilon)^{j-1} \Delta = (1+\epsilon)\Delta \cdot \left(1 - (1-\epsilon)^{i-1} \right).
\end{equation}
The corollary now follows from~(\ref{eq:cor:lem:trees:bound:1}) and the observation that $|P_i(u)| \geq (1+\epsilon)\Delta - \left| \bigcup_{j=1}^{i-1} N_j(u) \right|$.
\end{proof}

\begin{lem}
\label{lem:trees:bound:2}
Fix the random bits used by the {\sc Nibble} algorithm that determine which edges are selected in which rounds, in {\em any} way which ensures that the event $\mathcal{Z}$ occurs.\shay{perhaps want to mention that it's always possible to do so; also, ``in such a way'' - is any way that guarantees that $\mathcal{Z}$ occurs good for us or are there any restrictions? this is important for the probability distribution} Then w.h.p., we have $|P_i(e)| > (1-\epsilon^2) (1-\epsilon)^{2(i-1)}\Delta$ for all rounds $i \in [T]$ and all edges $e \in E_i$.
\end{lem}

\begin{proof}(Sketch)
Consider any  $i \in [T]$ and any edge $e = (u, v) \in E_i$. For each color $c \in \mathcal{C}$ and each node $w \in V$, let $X^w_c \in \{0, 1\}$  be an indicator random variable that is set to $1$ iff $c \in P_i(w)$. Clearly, we have $|P_i(e)| = \sum_{c \in \mathcal{C}} X^u_c \cdot X^v_c$. By~\Cref{obs:independence},  $X^u_c$ and $X^v_c$ are independent,\shay{(*) this independence
is the key; it needs to be emphasized via some observation before this lemma (that there are no dependencies between different ``components'', by the local nature of the alg)} and hence:
\begin{equation}
\label{eq:lem:trees:bound:2:1}
\mathbb{E}\left[ |P_i(e)| \right] = \sum_{c \in \mathcal{C}} \mathbb{E}\left[ X^u_c \cdot X^v_c \right] = \sum_{c \in \mathcal{C}} \mathbb{E}\left[ X^u_c \right] \cdot \mathbb{E}\left[X^v_c \right].
\end{equation}
We now consider any fixed color $c \in \mathcal{C}$, and derive a lower bound on $\mathbb{E}[X^u_c]$.\footnote{A lower bound on $\mathbb{E}[X^v_c]$ can be derived in the same way.} Using the symmetry of the {\sc Nibble} algorithm w.r.t.~the colors (see~\Cref{lem:main:strong sym}), in conjunction with the fact that $|P_i(u)| > (1+\epsilon)(1-\epsilon)^{i-1}\Delta$ (see~\Cref{lem:trees:bound:1}), we get: 
\begin{equation}
\label{eq:symmetry:key}
\mathbb{E}\left[ X^u_c \right] = \Pr\big[X^u_c = 1 \big] \geq \frac{(1+\epsilon)(1-\epsilon)^{i-1} \Delta}{|\mathcal{C}|}.
\end{equation}

Since $|\C| = (1+\epsilon)\Delta$, from~\Cref{eq:lem:trees:bound:2:1} and~\Cref{eq:symmetry:key}, we derive that:
\begin{equation}
\label{eq:lem:trees:bound:2:last}
\mathbb{E}\left[ \left|P_i(e) \right| \right] \geq |\mathcal{C}| \cdot \left(\frac{(1+\epsilon)(1-\epsilon)^{i-1} \Delta}{|\mathcal{C}|}\right)^2 = (1+\epsilon)(1-\epsilon)^{2(i-1)}\Delta.
\end{equation}
With some extra effort, we can modify the above argument by working with a slightly different (but analogous) set of random variables $Y^w_c$ (instead of $X^w_c$) such that the collection of random variables $\{ Y^u_c \cdot Y^v_c\}_{c \in \mathcal{C}}$ is {\em negatively associated}
(see~\Cref{def:NA}). This allows us to derive a concentration bound out of~\Cref{eq:lem:trees:bound:2:last}, which leads to  the proof of the lemma (see Appendix~\ref{app:static} for details).
\end{proof}

\begin{cor}
\label{cor:edge-palette}
Fix the random bits used by the {\sc Nibble} algorithm that determine which edges are selected in which rounds, in such a way that the event $\mathcal{Z}$ occurs. Then w.h.p., we have $|P_i(e)| \geq \epsilon^2 (1+\epsilon)\Delta/8$ for all rounds $i \in [T]$ and all edges $e \in E_i$.
\end{cor}

\begin{proof}(Sketch)
Follows from~\Cref{lem:trees:bound:2} and the observation that $i \leq T := \lfloor (1/\epsilon) \log (1/\epsilon) \rfloor$.
\end{proof}

\Cref{cor:edge-palette} will be useful in analyzing the runtime of the {\sc Nibble} algorithm in~\Cref{sec:forest:runtime}.

\subsubsection{Proof (Sketch) of Lemma~\ref{lem:trees:main}}
\label{sec:lem:trees:main}

Consider any node $v \in V$,\martin{minor comment---we could be more consistent with some variables e.g. using $u$ and $v$, expectations, probs} and let $N_F(v) = N(v) \cap F$ be the set of {\em failed edges} incident on $v$. We will show that $\text{deg}_F(v) = |N_F(v)| = O(\epsilon \Delta)$ w.h.p. \Cref{lem:trees:main} will then follow from a union bound over all nodes $v \in V$.  We begin by partitioning the set $N_F(v)$ into  three subsets: 
\begin{itemize}
\item $N^{(0)}_F(v) := N_F(v) \cap F_{T+1}$
(see~\Cref{alg:failed:1} in~\Cref{alg:nibble})
\item $N_F^{(1)}(v) := \left\{e \in N_F(v) \setminus N_F^{(0)}(v) : \tilde{\chi}(e) = \bot\right\}$
\item $N_F^{(2)}(v) := N_F(v) \setminus \left(N_F^{(0)}(v) \cup N_F^{(1)}(v)\right)$
\end{itemize}
We refer to the edges in $N_F^{(0)}(v)$ as {\em residual edges}, since they are not selected in any round $i \in [T]$. The set $N_F^{(1)}(v)$ consists of those edges $e \in N_F(v)$ which got selected in some round $i \in [T]$ but unfortunately had $P_i(e) = \varnothing$.  Finally, the set $N_F^{(2)}(v)$ consists of the remaining edges in $N_F(v)$, the ones who received the same tentative color as (at least one of) their neighbors in some round.

We will separately upper bound the sizes of the sets $N_F^{(0)}(v)$,  $N_F^{(1)}(v)$ and $N_F^{(2)}(v)$.

\begin{claim}
\label{cl:residual:main}
We have $\left|N_F^{(0)}(v)\right| = O(\epsilon \Delta)$ w.h.p.
\end{claim}

\begin{proof}(Sketch)
Any given edge $(u, v) \in E$ appears in the set $F_{T+1} = E \setminus S_{< T+1}$ with probability $= (1-\epsilon)^T \leq e^{-\epsilon T} = \epsilon$. This implies that $\mathbb{E}\left[ \left|N_F^{(0)}(v)\right|\right] = O(\epsilon \Delta)$.  Since $\Delta = \Omega\left(\log n/\epsilon^4\right)$ and since the random bits that determine whether different edges appear in $F_{T+1}$ are independent of each other, the proof now follows from an application of Chernoff bounds.
\end{proof}

For the rest of~\Cref{sec:lem:trees:main}, we fix the random bits used by the {\sc Nibble} algorithm that determine which edges are selected in which rounds, in such a way that the event $\mathcal{Z}$ occurs (the event $\mathcal{Z}$ occurs w.h.p., according to~\Cref{lem:trees:bound:1}).~\Cref{lem:trees:main} will follow from~\Cref{cl:residual:main},~\Cref{cl:deg:bound:99} and~\Cref{cor:last:101}.
\begin{claim}
\label{cl:deg:bound:99} We have $\left|N_F^{(1)}(v)\right| = 0$ w.h.p.
\end{claim}

\begin{proof}(Sketch)
~\Cref{lem:trees:bound:2} implies that w.h.p.~the following event occurs: 
$P_i(e) \neq \varnothing$ for all $i \in [T]$ and all $e \in S_i$. Conditioned on this event, no edge gets added to the set $N_F^{(1)}(v)$.
\end{proof}

We now focus on deriving an upper bound on $\left| N_F^{(2)}(v) \right|$. We start with the following claim.

\begin{claim}
\label{cl:conflict}
Consider any round $i \in [T]$, any edge $(u, v) \in S_i$, and any endpoint $x \in \{ u, v\}$.  Then we have: 
$\Pr\left[~\exists~ e \in N_i(x) \setminus \{(u, v)\}:  \tilde{\chi}(u,v) = \tilde{\chi}(e) \right] \leq \epsilon$.\shay{changed ``for some $e$...''; or we can write instead
$\Pr\left[ \tilde{\chi}(u,v) \in \tilde{\chi}(N_i(x) \setminus \{(u, v)\})\right]$}
\end{claim}

\begin{proof}(Sketch) W.l.o.g.~suppose that $x = u$ (the proof for the case $x = v$ is symmetric). Consider any edge $(u, w) \in S_i$ with $w \neq v$.  We first lower bound the probability that $\tilde{\chi}(u, v) = \tilde{\chi}(u, w)$.

\Cref{obs:independence} implies that the palettes $\{P_i(u), P_i(v), P_i(w)\}$ are mutually independent.\shay{(*) as before, this is the key; refer to an observation on independence in different components}  Furthermore, since we have conditioned on the event $\mathcal{Z}$,~\Cref{cor:lem:trees:bound:1} lower bounds the sizes of each of these  palettes $\{P_i(u), P_i(v), P_i(w)\}$. Now, fix (i.e., condition on) the palettes $P_i(u)$ and $P_i(v)$, and let $c = \tilde{\chi}(u, v) \in P_i(u) \cap P_i(v)$ be the tentative color received by the edge $(u, v)$ in round $i$. Using the symmetry of  the {\sc Nibble} algorithm w.r.t.~the colors (see~\Cref{lem:main:strong sym}),\shay{(*) should explain why symmetry holds even after conditioning} and recalling that $|P_i(u)| \geq (1+\epsilon) (1-\epsilon)^{i-1}\Delta$ as per~\Cref{cor:lem:trees:bound:1},  we get:
$$\Pr[\tilde{\chi}(u, w) = c] ~\leq~ \frac{1}{|P_i(u)|} ~\leq~ \frac{1}{(1+\epsilon)(1-\epsilon)^{i-1}\Delta}.$$
Therefore, we conclude that $\Pr\left[\tilde{\chi}(u, v) \neq \tilde{\chi}(u, w)\right] \geq 
 \left(1-\frac{1}{(1+\epsilon)(1-\epsilon)^{i-1}\Delta}\right)$, and hence:
\begin{equation}
\label{eq:lem:trees:bound:100}
\Pr\left[\tilde{\chi}(u, v) \neq \tilde{\chi}(u, w) \  \forall \ (u, w) \in N_i(u) \setminus \{ (u, v)\} \right]  \geq  \left(1-\frac{1}{(1+\epsilon)(1-\epsilon)^{i-1}\Delta}\right)^{|N_i(u)|} \geq 1-\epsilon.
\end{equation}
The first inequality holds because of~\Cref{obs:independence}, whereas  the last inequality\shay{(*) should explain the first inequality - why can you assume independence of these 
$|N_i(u)| - 1$ events, given all the conditionings that were made}
holds since we have conditioned on the event $\mathcal{Z}$ (see~\Cref{lem:trees:bound:1}). The lemma follows if we consider the probability of the complement of the event captured by~\Cref{eq:lem:trees:bound:100}.  
\end{proof}

\begin{claim}
\label{cl:lm:trees:bound:101}
Consider any round $i \in [T]$. Then we have: $\left|N_F^{(2)}(v) \cap S_i\right| = O(\epsilon^2 (1-\epsilon)^{i-1} \Delta)$ w.h.p.
\end{claim}

\begin{proof}(Sketch)
Consider the following two subsets of edges incident on $v$.
\begin{itemize}

\item $F'_i(v) := \{ (u, v) \in S_i : \exists (v, w) \in S_i \text{ with } \tilde{\chi}(u, v) = \tilde{\chi}(v, w)\}$
\item $F''_i(v) := \{ (u, v) \in S_i : \exists (u, w) \in S_i \text{ with } \tilde{\chi}(u, v) = \tilde{\chi}(u,w) \}$
\end{itemize}

\noindent Observe that $N_F^{(2)}(v) \subseteq F'_i(v) \cup F''_i(v)$. Now, consider any edge $(u, v) \in N_i(v)$. By~\Cref{cl:conflict}, we have $\Pr[(u, v) \in F'_i(v)] \leq \epsilon$ and $\Pr[(u, v) \in F''_i(v)] \leq \epsilon$. Thus, by linearity of expectation, we get:
\begin{equation}
\label{eq:expectation:101}
\mathbb{E}\left[ \left| N_F^{(2)}(v) \right|\right] \leq \mathbb{E}\left[ \left| F'_i(v) \right|\right] + \mathbb{E}\left[ \left| F''_i(v) \right|\right] \leq 2\epsilon \cdot |N_i(v)| = O(\epsilon^2 (1-\epsilon)^{i-1} \Delta).
\end{equation}
In the above derivation, the last step holds because we  conditioned on the event $\mathcal{Z}$ (see~\Cref{lem:trees:bound:1}). With a little bit of extra work, we can infer that: (a) $|F'_i(v)|$ can be expressed as a function of a collection of random variables that is Lipschitz with all constants $2$ (see~\Cref{def:lippy}), which allows us to derive a concentration bound on $|F_i'(v)|$ using the method of bounded differences (see~\Cref{lem:bounded diff}). (b) Using the fact that the input graph $G$ is a forest, $|F''_i(v)|$ can be expressed as the sum of mutually independent\shay{(*) again, explain why we can assert this after the conditioning} $0/1$ random variables, which allows us to derive a concentration bound on $|F''_i(v)|$ by applying a Chernoff bound. This concludes the proof of the claim.
\end{proof}

\begin{cor}
\label{cor:last:101}
We have: $\left|N_F^{(2)}(v)\right| = O(\epsilon \Delta)$ w.h.p.
\end{cor}

\begin{proof}
Since $N_F^{(2)}(v) \subseteq \bigcup_{i=1}^T S_i$, from~\Cref{cl:lm:trees:bound:101} we derive that whp:
$$\left|N_F^{(2)}(v)\right| = O\left(\sum_{i=1}^T \epsilon^2 (1-\epsilon)^{i-1}\Delta \right) ~=~ O\left(\epsilon (1 - (1-\epsilon)^T) \Delta\right) ~=~ O(\epsilon \Delta).$$
\end{proof}

~\Cref{lem:trees:main} now follows from~\Cref{cl:residual:main},~\Cref{cl:deg:bound:99},~\Cref{lem:trees:bound:1} and~\Cref{cor:last:101}.

\subsubsection{Running Time of the {\sc Nibble} Algorithm on Forests}
\label{sec:forest:runtime}

We now briefly describe how we can implement the  \textsc{Nibble} algorithm in linear time when the input graph $G = (V, E)$ is a forest. We begin by scanning through all of the edges  $e \in E$, and assigning a round $i_e \in [T+1]$ to each edge $e$, sampled i.i.d. from a capped geometric distribution with success probability $\epsilon$. We then construct the sets $S_1,\dots,S_{T+1}$, where $S_i$ consists of the edges that will be colored during round $i$. This can be done in $O(m)$ time. For each $i=1,\dots,T$, we then scan through all of the edges $e \in S_i$ and sample a color $\tilde \chi(e) \sim P_i(e)$ independently and u.a.r. In order to implement this sampling, suppose that we can check whether or not a given color $c$ is contained in $P_i(e)$ in $O(1)$ time. Since we know by~\Cref{cor:edge-palette} that $|P_i(e)| = \Omega(\epsilon^2\Delta)$ w.h.p., we can sample a color $c$ from $\mathcal C$ independently and u.a.r.~and this color will be contained in $P_i(e)$ with probability $\Omega(\epsilon^2)$. Hence, in expectation, we need to sample $O(1/\epsilon^2)$ many colors before finding one that is contained in $P_i(e)$. It follows that the total expected time required to sample all of the tentative colors is $O(m/\epsilon^2)$. Finally, we can color all the failed edges using the folklore algorithm from~\Cref{sec:perspective} in $O(m)$ time.

It turns out that using hash tables we can maintain, for each node $v \in V$,   the complement of its palette, given by $\overline{P(v)} := \{ c \in \mathcal{C} : \exists (u, v) \in E \text{ s.t. } \chi(u, v) = c\}$. This allows us to check whether a color is in the palette of an edge $e = (u,v)$ in $O(1)$ time. Specifically, suppose that at the start of round $i$, we maintain $\overline{P_i(v)}$ for each node $v$. We can then use this data structure to implement the sampling efficiently as described above. After sampling all of the tentative colors for the edges in $S_i$, we can then update each of these sets in $O(|S_i|)$ time, by inserting the appropriate colors to each set $\overline{P(v)}$, obtaining $\overline{P_{i+1}(v)}$ for all nodes $v$. We can also implement data structures that keep track of all the failed edges. We defer all the details of these data structures to \Cref{sec:app:datastructstatic}. 

To summarize, we get the following result.

\begin{lem}
    Given a forest $G$ with maximum degree $\Delta \geq (100 \log n)/\epsilon^4$ as input, the {\sc Nibble} algorithm runs in $O(m/\epsilon^2)$ expected time.
\end{lem}

\subsection{Analyzing the {\sc Nibble} Algorithm on General Graphs}
\label{sec:main:static:general}

In this section, we analyze the {\sc Nibble} algorithm when the input graph $G$ may contain cycles. Looking back at the analysis in~\Cref{sec:main:static:trees}, it is easy to check that we crucially needed $G$ to be acyclic because we wanted to apply~\Cref{obs:independence} in our analysis.\shay{(*) I suggest to state this as a claim, perhaps give a short proof sketch, and then use this claim in all the proofs} This observation was used, for example, in the proof of~\Cref{lem:trees:bound:2} and in~\Cref{sec:lem:trees:main}. The key insight that we employ in this section is that even if $G$ contains cycles, the preceding claim still holds for all the nodes/edges that are {\em not} part of any cycle of length $\leq T+1$ (the number of rounds in the {\sc Nibble} algorithm). Formally, by doing induction on $i$, we can prove the following lemma.

\begin{lem}
\label{lem:long:cycle}
Let $\mathcal{N}_G(u, j) := G\left[\{ v \in V : \text{dist}_G(u, v) \leq j\}\right]$ denote the $j$-hop neighborhood of any node $u \in V$.  Then for every $i \in [T]$, the palette $P_i(u)$ depends only on:  $\mathcal{N}_G(u, i-1)$, and  the random bits used by the {\sc Nibble} algorithm to implement rounds $j \in \{1, \ldots, i-1\}$ in $ \mathcal{N}_G(u,i-1)$.\footnote{We use the symbol $\text{dist}_G(u, v)$ to denote the distance between $u$ and $v$ in $G$. Furthermore, for any subset of nodes $V' \subseteq V$, the symbol $G[V']$ denotes the subgraph of $G$ induced by $V'$.}
\end{lem}

We now introduce a key definition below.

\begin{definition}
    \label{def:good:nodes}
    We say that a node $v \in V$ is {\em good} w.r.t.~$G = (V, E)$ iff $\mathcal{N}_G(u, T+1)$, the subgraph induced\shay{it's not the subgraph induced by the neighborhood, it's the neighborhood itself (under the def'n of lemma 2.11)} by its $(T+1)$-hop neighborhood in $G$, does not contain any cycle. Let $U_G \subseteq V$ denote the set of all good nodes in $G$. We refer to the rest of nodes $B_G := V \setminus U_G$ as {\em bad} w.r.t.~$G$.
\end{definition}

Informally, since~\Cref{alg:nibble} only runs for $T$ rounds, and the decisions it makes are \textit{local}, what the algorithm does to some $H \subseteq G$ should only depend on the subgraph of $G$ that is reachable from $H$ by paths of length at most $T$. Hence, if the $(T+1)$-neighborhood of a node $u$ is acyclic, then it is safe to pretend that the input graph is a forest while analyzing what the algorithm does to the edges incident on $u$.  This implies that even if $G$ contains cycles, we can recover the guarantee of~\Cref{lem:trees:main} for all the good nodes, which leads us to the following lemma.

\begin{lem}
\label{lem:trees:main:good}
Suppose that the graph $G$ has maximum degree $\Delta \geq \frac{(100 \log n)}{\epsilon^4}$. Then w.h.p.~we have:
$$\text{deg}_F(v) = O(\epsilon \Delta) \text{ for all nodes } v \in U_G.$$
\end{lem}

\subsection{The Final (Static) Algorithm: {\sc Nibble} on Subsampled Graphs}
\label{sec:main:static:final}

Throughout~\Cref{sec:main:static:final}, we use  two parameters: $\gamma := 1/(30T)$ and $\Delta' := \Delta^{\gamma}$. We also assume that:
$$\Delta \geq (100 \log n/\epsilon^4)^{(30/\epsilon) \log (1/\epsilon)}.$$
\shay{I suppose this lower bound on $\Delta$ is worse than prev works? do KLS22 also have such a lb on $\Delta$? The reason we need such a large lb is due to Lem 2.18 in conjunction to the need of negative association? Would be interesting to include a short discussion on that.}
  We now combine~\Cref{alg:nibble} with the subsampling technique of~\cite{KulkarniLSST22}, which leads to our final algorithm on general graphs. This consists of two steps.

\medskip
\noindent {\bf Step I:} Set $\eta := \Delta/\Delta'$, and partition the input graph $G = (V, E)$ into $\eta$ subgraphs $\mathcal{G}_1 = (V, \mathcal{E}_1), \ldots, \mathcal{G}_{\eta} = (V, \mathcal{E}_{\eta})$, by placing each edge $e \in E$ uniformly and independently at random into one of the subsets $\mathcal{E}_1, \ldots, \mathcal{E}_{\eta}$. To ease notation, for all $j \in [\eta]$ we let $\mathcal{U}_j := U_{\mathcal{G}_j}$ and $\mathcal{B}_j := B_{\mathcal{G}_j}$ denote the sets of good and bad nodes in the subgraph $\mathcal{G}_j$ (see~\Cref{def:good:nodes}). We now derive two important properties of this subsampling step that will be crucially used in Step II.

\begin{claim}
\label{cl:subsampling:chernoff}
W.h.p., for every $j \in [\eta]$ we have $\Delta\left(\mathcal{G}_j\right) \leq (1+\epsilon)\Delta'$.\shay{the proof of this claim is straightforward and can be omitted}
\end{claim}

\begin{proof}(Sketch)
Fix an index $j \in [\eta]$. Any edge $e \in E$ appears in $\mathcal{G}_j$ with probability $1/\eta$. Hence, the expected degree of any node $v \in V$ in $\mathcal{G}_j$ is at most $\Delta \cdot (1/\eta) = \Delta'$. Since $\Delta$ is sufficiently large, from a straightforward application of Chernoff bound we infer that w.h.p.~$\text{deg}_{\mathcal{G}_j}(v) \leq (1+\epsilon)\cdot \Delta'$. The claim now follows from a union bound over all $v \in V$ and $j \in [\eta]$.
\end{proof}

\begin{claim}
\label{cl:subsampling:bad}
Consider the graph $G^{\star} = (V, E^{\star})$,  where $E^{\star} := \bigcup_{j=1}^{\eta} \left\{ (u, v) \in \mathcal{E}_j : \{ u, v\} \cap \mathcal{B}_j \neq \varnothing \right\}$ is the set of all edges that are incident on bad nodes in their corresponding subsampled graphs. Then w.h.p., we have $\Delta\left( G^{\star} \right) = o(\Delta)$.
\end{claim}

We defer the proof of~\Cref{cl:subsampling:bad} to~\Cref{sec:cl:subsampling:bad}, and proceed to the next step of our algorithm.

\medskip
\noindent {\bf Step II:} For each $j \in [\eta]$, taking~\Cref{cl:subsampling:chernoff} into account, we  invoke~\Cref{alg:nibble} on  $\mathcal{G}_j$ with a palette $\mathcal{C}_j$ (i.e., the set $\mathcal{C}$ in~\Cref{alg:palette} of~\Cref{alg:nibble}) of $(1+\epsilon)^2\Delta'$ colors. The palettes $\mathcal{C}_1, \ldots, \mathcal{C}_j$ are mutually disjoint.  For each $j \in [\eta]$, let $\mathcal{F}_j$ be the set of {\em failed edges} at the end of the corresponding invocation of the {\sc Nibble} algorithm on $\mathcal{G}_j$ (i.e., the set $F$ in~\Cref{alg:failed:2} of~\Cref{alg:nibble}). Let $\mathcal{F} := \bigcup_{j=1}^{\eta} \mathcal{F}_j$ be the collection of these failed edges over all $j \in [\eta]$, and let $\mathcal{G}_{\mathcal{F}} := (V, \mathcal{F})$. We color the edges in $\mathcal{G}_{\mathcal{F}}$ using the folklore algorithm from~\Cref{sec:perspective},\shay{(*) should mention what coloring greedily means (and that it takes linear time)} and an extra set of $O(\Delta(\mathcal{G}_{\mathcal{F}}))$ colors.

\medskip
It is easy to verify that the above algorithm returns a proper edge coloring of  $G$, and that the number of colors it uses is at most: $O(\Delta(\mathcal{G}_{\mathcal{F}})) + \sum_{j=1}^{\eta} \left|\mathcal{C}_j \right| = O(\Delta(\mathcal{G}_{\mathcal{F}})) + (1+\epsilon)^2 \Delta' \cdot \eta = O(\Delta(\mathcal{G}_{\mathcal{F}})) + (1+\epsilon)^2 \Delta' \cdot (\Delta/\Delta') = O(\Delta(\mathcal{G}_{\mathcal{F}})) + (1+\epsilon)^2 \Delta$, w.h.p. We upper bound $\Delta(\mathcal{G}_{\mathcal{F}})$ in~\Cref{cl:failed:bound:1} below. This implies the main result in this section, which is summarized in~\Cref{cor:failed:bound:1}.

\begin{claim}
\label{cl:failed:bound:1}
We have $\Delta\left( \mathcal{G}_{\mathcal{F}} \right) = O(\epsilon \Delta)$ w.h.p.
\end{claim}

\begin{proof}
Throughout the proof, fix any node $v \in V$. Partition the set of indices $[\eta]$ into two subsets - $J_{g} := \{ j \in [\eta] : v \in \mathcal{U}_j\}$ and $J_b := \{ j \in [\eta] : v \in \mathcal{B}_j\}$ - depending on whether or not $v$ is a good node in the corresponding subsampled graph. For each $j \in J_g$,~\Cref{cl:subsampling:chernoff} and~\Cref{lem:trees:main:good} imply that $\text{deg}_{\mathcal{F}_j}(v) = O(\epsilon \Delta')$ w.h.p.\shay{how can you apply Lem 2.13 on a subgraph corresponding to a good node? there is implicit conditioning here} Thus, summing over all $j \in J_g$, we derive the following bound w.h.p.
\begin{equation}
\label{eq:good:bound}
\sum_{j \in J_g} \text{deg}_{\mathcal{F}_j}(v) = |J_g| \cdot O(\epsilon \Delta') \leq \eta \cdot O(\epsilon \Delta') = O(\epsilon \Delta).
\end{equation}
Next, observe that every edge $(u, v) \in \bigcup_{j \in J_b} \mathcal{F}_j$, by definition, contributes to the set $E^{\star}$. Thus, from~\Cref{cl:subsampling:bad}, we infer that w.h.p.
\begin{equation}
\label{eq:good:bound:2}
\sum_{j \in J_b} \text{deg}_{\mathcal{F}_j}(v) \leq \text{deg}_{G^{\star}}(v) = o(\Delta).
\end{equation}
From~\Cref{eq:good:bound} and~\Cref{eq:good:bound:2}, we get $\text{deg}_{\mathcal{G}_{\mathcal{F}}}(v) = O(\epsilon \Delta
)$ w.h.p.\shay{in this transition, the implicit assumption is that
$\eps$ is fixed; the proof of Lem 2.18 gives an upper bound of say $O(\Delta^{1/4})$, which is $O(\eps \Delta)$  for all $\eps$ due to the lb on $\Delta$, and it might make sense not to assume that $\eps$ is fixed here} The claim now follows from a union bound over all $v \in V$.
\end{proof}

\begin{cor}
\label{cor:failed:bound:1} W.h.p., the above algorithm returns a $(1+O(\epsilon))\Delta$-coloring of the input graph $G$. 
\end{cor}

It is easy to verify that the subsampling step can be implemented in $O(m)$ time, because all we need to do is decide for each edge $e \in E$ which subgraph $\mathcal{G}_j$ it should appear in. We now implement the {\sc Nibble} algorithm in each subsampled graph $\mathcal{G}_j$ using the approach outlined in~\Cref{sec:forest:runtime}. Putting everything together, this implies a $O_{\epsilon}(m)$ time algorithm for $(1+\epsilon)\Delta$-edge coloring in a general graph. See~\Cref{sec:app:datastructstatic} for details.\shay{(*) didn't address the runtime; IMO should add a para saying that it can be implemented in linear time using appropriate data structures (crucially relying on not using near-regularization gadgets), and perhaps add something concrete in a few words, and say we omit most details since that will be anyway subsumed by the constant update time result}

\subsubsection{Proof (Sketch) of~\Cref{cl:subsampling:bad}}
\label{sec:cl:subsampling:bad}

We will use the following lemma, which is an immediate corollary of Lemma 4.2 in~\cite{KulkarniLSST22}.\shay{(*) I suppose this lemma is known (related to high girth random graphs) and shouldn't be difficult to prove, and we may want give credit to the first paper/book with such a bound; or perhaps we should rephrase the sentence preceding lem 2.18, saying that we use the following known result in random graphs (see, e.g., Lem 4.2 in KLS22)}

\begin{lem}\label{lem:kill small cycles new}
     Let $G'$ be a subgraph of $G$ obtained by sampling each edge in $G$ independently with probability $D/\Delta$, where $D \geq 2$. Then the probability that the $g$-neighborhood of a node $u$ contains a cycle in $G'$ is at most $3D^{5g}/\Delta$.
\end{lem}

Consider any node $u \in V$. For each $j \in [\eta]$, define $X^u_j \in \{0,1\}$ to be the indicator random variable for the event that the $(T+2)$-neighborhood of $u$ in the graph $\mathcal G_j$ contains a cycle. Thus, we have $X^u_j = 1$ whenever at least one of the neighbors of $u$ in $\mathcal{G}_j$ is contained in $\mathcal B_j$. It follows that:
\begin{equation}
\label{eq:girth}\deg_{G^\star}(u) \leq \sum_{j \in [\eta]} X^u_j \cdot |N(u) \cap \mathcal E_j| \leq \sum_{j \in [\eta]} X^u_j \cdot \Delta(\mathcal{G}_j).
\end{equation}
Each edge $e \in E$ is sampled in $\mathcal{G}_j$ independently with probability $1/\eta = \Delta'/\Delta$, and we have that $\Delta' \geq 2$ from the definitions of our parameters. Thus, setting $D = \Delta'$ and $g = T+2$,~\Cref{lem:kill small cycles new} gives us: $\mathbb{E}\left[X_j^u\right] = \Pr\left[ X^j_u = 1 \right] \leq 3(\Delta')^{5(T+3)}/\Delta= o(1)$.\shay{is it $T+2$? also, may want to be more precise than $o(1)$, say $O(\Delta^{-3/4})$ (see my prev comment)} By linearity of expectation, we get:
\begin{equation}
\label{eq:girth:2}
\mathbb{E}\left[\sum_{j \in [\eta]} X_j^u\right] =  \eta \cdot o(1) = o(\Delta/\Delta').
\end{equation}

With some extra effort, we can show that the random variables $\{ X^u_j \}_j$ are negatively associated (see~\Cref{def:NA}). This allows us to apply a Chernoff bound w.r.t.~\Cref{eq:girth:2},  and derive that: $\sum_{j \in [\eta]} X^u_j = o(\Delta/\Delta')$ w.h.p. Next,  by~\Cref{cl:subsampling:chernoff}, w.h.p.~we have: $\Delta\left(\mathcal{G}_j\right) = O(\Delta')$ for all $j \in [\eta]$. These last two observations, along with~\Cref{eq:girth}, imply that $\text{deg}_{G^{\star}}(u) = o(\Delta)$ w.h.p.\shay{and here we'll get say $O(\Delta^{1/4})$, which is $O(\eps \Delta)$ for any $\eps$} The claim now follows from a union bound over all nodes $u \in V$.

\section{Overview of our Dynamic Algorithm}
\label{sec:overview:dynamic}

In this section, we present an overview of our dynamic algorithm for edge coloring (see~\Cref{thm:dynamic:main}). As in~\Cref{sec:overview:static}, to convey the main insights behind our approach, the arguments here will be often informal and lacking in low-level details. The complete, formal proofs from this section are deferred to~\Cref{sec:app:recourse} and~\Cref{sec:app:datastructs}.

To simplify the presentation, throughout~\Cref{sec:overview:dynamic} we will assume that the input graph $G = (V, E)$ always remains a forest with maximum degree at most $\Delta$, throughout the sequence of updates. We will essentially dynamize the static algorithm from~\Cref{sec:main:static:nibble}. Thus, it is reasonable to expect that once we have a fast dynamic implementation of~\Cref{alg:nibble} on forests, we can extend this to work on general graphs using the subsampling framework from~\Cref{sec:main:static:final}.\footnote{Note that it is very easy to implement the subsampling based partitioning from~\Cref{sec:main:static:final} in a dynamic setting: When an edge $e$ is inserted, just sample an index $j \in [\eta]$ u.a.r.~and add $e$ to the subgraph $\mathcal{G}_j$.}

\medskip
\noindent {\bf Organization.} In~\Cref{sec:main:static:recourse}, we present our dynamic algorithm and analyze its {\em recourse} -- the number of changes it makes to the colors of edges per update -- ignoring any concern about efficient data structures (see~\Cref{thm:main:recourse}).~\Cref{sec:main:static:updatetime} shows that a slightly modified version of the dynamic algorithm from~\Cref{sec:main:static:recourse} can be implemented with fast data structures in $O_{\epsilon}(1)$ update time.

\subsection{Bounding the Recourse}
\label{sec:main:static:recourse}

We start by presenting our dynamic algorithm.

\medskip
\noindent {\bf Preprocessing.} We refer to an unordered pair of nodes as a {\em potential edge}. For each potential edge $e \in {V \choose 2}$, we independently sample an index $i_e \in [T+1]$ from a {\em capped geometric distribution} with success probability $\epsilon$, and an (infinite-length) {\em color-sequence} $c_e$. Specifically, for each integer $j \geq 1$, the $j^{th}$ color in the sequence $c_e$ is denoted by $c_e(j)$ and is sampled independently and u.a.r.~from the palette $\mathcal{C} = [(1+\epsilon)\Delta]$. At preprocessing, we sample and then fix these indices $\{i_e\}_e$ and the color-sequences $\{c_e\}_e$ for every potential edge $e \in {V \choose 2}$.  Note that they uniquely determine the output of~\Cref{alg:nibble} on any given input graph $G = (V, E)$, as follows. (1) Each edge $e \in E$ selects itself in round $i_e$ (i.e., $e \in S_{i_e}$ if $i_e \leq T$ and $e \in F_{T+1}$ otherwise). (2) While considering an edge $e \in S_i$ in round $i$ we simply scan through the sequence $c_e$ until we find the smallest index $\ell_e \in \mathbb{Z}^+$ such that $c_e(\ell_e) \in P_i(e)$, and set $\tilde{\chi}(e) := c_e(\ell_e)$. If no such index $\ell_e$ exists, i.e., if $P_i(e) = \varnothing$, then we set $\ell_e := 0$ and $\tilde{\chi}(e) := \bot$.

\medskip
\noindent {\bf Handling the sequence of updates.} We use the superscript $t$ to denote the status of any object at time $t$. For instance, $G^{(t)} = \left( V, E^{(t)} \right)$ denotes the input graph just after the $t^{th}$ update. We maintain the tentative coloring $\tilde{\chi}^{(t)} : E^{(t)} \rightarrow \mathcal{C}$ , which is obtained by executing~\Cref{alg:nibble} on $G^{(t)}$, with the same random choices that were fixed at preprocessing.  We also collect the subgraph $G_F^{(t)} := \left(V, F^{(t)}\right)$ consisting of all the failed edges and maintain a $O\left(\Delta(G_F^{(t)})\right)$-coloring $\psi^{(t)}$ in $G_F^{(t)}$ using an extra palette of colors that is mutually disjoint with $\mathcal{C}$.  The final coloring $\chi$ is then defined as follows. For all $e \in E^{(t)}$, we have $\chi^{(t)}(e) = \tilde{\chi}^{(t)}(e)$ if $e \notin F^{(t)}$, and $\chi^{(t)}(e) = \psi^{(t)}(e)$ otherwise.

\begin{thm}
\label{thm:main:recourse}
The dynamic algorithm presented above has an expected recourse of $O(1/\epsilon^4)$ per update, and at each time $t$ w.h.p.~maintains a proper $(1+O(\epsilon))\Delta$-edge coloring of $G^{(t)}$.
\end{thm}

We devote the rest of~\Cref{sec:main:static:recourse} to the proof of~\Cref{thm:main:recourse}. Towards this end, first observe that the total number of colors used by the algorithm is $|\C| + O\left(\Delta(G_F^{(t)})\right) = (1+\epsilon)\Delta + O\left(\Delta(G_F^{(t)})\right)$, which equals $(1+O(\epsilon)) \Delta$ w.h.p., according to~\Cref{lem:trees:main}. Thus, it now remains to bound the expected recourse of the algorithm. Let $A^{(t)} := \left\{ e \in E^{(t-1)} \cup E^{(t)} : \tilde{\chi}^{(t-1)}(e) \neq \tilde{\chi}^{(t)}(e)\right\}$ denote the set of edges that change their tentative color due to the $t^{th}$ update, where we define $\tilde \chi^{(t)}(e) = \bot$\martin{edges not in the graph are just defined as uncolored} for all $e \notin E^{(t)}$.~\Cref{thm:main:recourse} then follows from~\Cref{lem:recourse:1} and~\Cref{lem:recourse:2}.

\begin{lem}
\label{lem:recourse:1}
The recourse of the algorithm at time $t$ is at most $O(1) +O\left(\left| A^{(t)}\right|\right)$. 
\end{lem}

\begin{proof}
W.l.o.g., suppose that the $t^{th}$ update involves the insertion of an edge $e^{\star}$ (we can analyze deletions analogously). Thus, we have $E^{(t)} := E^{(t-1)} \cup \{e^{\star}\}$. Note that $\chi^{(t')}(e) \in \{ \tilde{\chi}^{(t')}(e), \psi^{(t')}(e) \}$ for all $e \in E^{(t)}$ and $t' \in \{t-1, t\}$.\footnote{For each of $\chi$, $\tilde \chi$, and $\psi$, we define the color of an edge $e$ not in the graph as $\bot$.} Let $R^{(t)}$ be the recourse of the algorithm for the $t^{th}$ update, and let $R_{\tilde{\chi}}^{(t)}$ and $R_{\psi}^{(t)}$ respectively denote the number of changes to the colorings $\tilde{\chi}$ and $\psi$ 
due to the $t^{th}$ update. Thus, we have $R^{(t)} \leq R_{\tilde{\chi}}^{(t)} + R_{\psi}^{(t)} = |A^{(t)}| + R_{\psi}^{(t)}$. It now   remains to show that $R_{\psi}^{(t)} = O(1) + O\left( | A^{(t)} | \right)$. Towards this end, we first observe that $R_{\psi}^{(t)} \leq | F^{(t-1)} \oplus F^{(t)}|$,\footnote{The symbol $\oplus$ denotes the symmetric difference between two sets.}  because it is trivial to maintain a $O(\Delta)$-coloring in a dynamic graph with maximum degree $\Delta$ with a recourse of at most one per update.\sayan{Give a forward pointer to Appendix.} The lemma now follows from~\Cref{cl:main:recourse:1}.
\begin{equation}
\label{cl:main:recourse:1}
| F^{(t-1)} \oplus F^{(t)} | = O(1) + O\left( |A^{(t)}| \right).\end{equation}

We devote the rest of the proof towards showing why~\Cref{cl:main:recourse:1} holds.   Consider any subset  $E' \subseteq {V \choose 2}$ and any coloring $\chi' : E' \rightarrow \mathcal{C} \cup \{ \bot \}$. We say that an $f \in E'$ is a {\em failed edge} w.r.t.~$(E', \chi')$ iff either $\chi'(f) = \bot$ or it has a neighboring edge $e \in E'$ (i.e., $e$ and $f$ shares an endpoint) with the same color (i.e., $\chi'(e) = \chi'(f)$).  For instance, $F^{(t)}$ is the set of failed edges w.r.t.~$(E^{(t)}, \tilde{\chi}^{(t)})$.

Let $A^{(t)} \cup \{ e^{\star}\} := \{ e_1, \ldots, e_{\ell}\}$, and for each $r \in [0,\ell]$, let $F(r)$ be  the set of failed edges w.r.t.~$E^{(t)}$ and the coloring where each edge $e \in \{e_1, \ldots, e_{r}\}$ receives the color $\tilde{\chi}^{(t)}$ and each edge $e \in E^{(t)} \setminus \{ e_1, \ldots, e_r\}$ receives the color $\tilde{\chi}^{(t-1)}$.  Intuitively, consider a process where we start with the coloring $\tilde{\chi}^{(t-1)}$, and then change the colors of the edges $\{e_1, \ldots, e_{\ell}\}$ from $\tilde{\chi}^{(t-1)}$ to $\tilde{\chi}^{(t)}$, one at a time. The sets $F(r)$  track how $F^{(t-1)}$ evolves into $F^{(t)}$ during this process. It is easy to see that
\[ | F^{(t-1)} \oplus F^{(t)}  | \leq |F(0) \oplus F(1)| + | F(1) \oplus F(2) | + \dots + | F(\ell-1) \oplus F(\ell)|, \]
since each  $f \in F^{(t-1)} \oplus F^{(t)}$ must either be added to or removed from the set of failed edges after some edge $e \in \{e_1, \ldots, e_{\ell}\}$  changes its color during the process described above. Now, fix any $r \in [\ell]$,  let $e_{r} = (u,v)$, and consider the event when we change the color of $e_{r}$ from $\tilde \chi^{(t-1)}(e_{r})$ to $\tilde \chi^{(t)}(e_{r})$ during the above process (assume for now that neither color is $\bot$). Due to this event,  at most $2$ edges can get added to the set $F$. This is because the only edges that will be added to $F$ after this color change are edges incident to $u$ and $v$ with color $\tilde \chi^{(t)}(e_{r})$ that are not already in $F$, and there can be at most one such edge incident on each of $u$ and $v$. By an analogous argument, at most $2$ edges can get removed from $F$ due to this event. It follows that $|F(r -1) \oplus F(r)| \leq 4$. A similar argument shows that if $\bot \in \left\{ \tilde{\chi}^{(t-1)}(e_{r}), \tilde{\chi}^{(t)}(e_{r}) \right\}$, then $|F(r -1) \oplus F(r)| \leq 3$.  Thus, in both cases, we have $|F(r-1) \oplus F(r)| = O(1)$. Summing over all $r \in [\ell]$, this gives us $|F^{(t-1)} \oplus F^{(t)}| = O(\ell)$. ~\Cref{cl:main:recourse:1} now follows because $\ell \leq 1 + |A^{(t)}|$.  
\end{proof}

\begin{lem}
\label{lem:recourse:2}
At each time $t$, we have $\mathbb{E}\left[ |A^{(t)}| \right] = O(1/\epsilon^4)$.
\end{lem}

\begin{proof}
As in the proof of~\Cref{lem:recourse:1}, w.l.o.g.~suppose that the $t^{th}$ update involves the insertion of an edge $e^{\star}$ (we can analyze deletions analogously). Thus, we have $E^{(t)} := E^{(t-1)} \cup \{e^{\star}\}$. Let $A_{i}^{(t)} := A^{(t)} \cap S_i^{(t)}$ for all $i \in [T]$.  For the rest of the proof, fix the rounds $\{i_{e}\}_e$ of all potential edges $e \in {V \choose 2}$, and condition on the high-probability event $\mathcal{Z}$ (see~\Cref{lem:trees:bound:1}) on both  $G^{(t-1)}$ and $G^{(t)}$. 

Since the edge $e^\star$ gets selected in round $i_{e^{\star}}$,  due to the $t^{th}$ update no edge $e \in S_{\leq i_{e^\star}}^{(t)} \setminus \{e^\star\}$ changes its tentative color.  In other words, we have $|A_{i}^{(t)}| = 0$ for all $i < i_{e^{\star}}$ and $|A_{i^{e^\star}}^{(t)}| \leq 1$. We will show:
\begin{equation}
\label{eq:recourse:main}
\mathbb{E}\left[ | A_i^{(t)} | \right] \leq 4\epsilon \cdot \mathbb{E}\left[ | A_{< i}^{(t)}|\right] \text{ for all } i \in [i_{e^\star} + 1, T].
\end{equation}
\Cref{eq:recourse:main} implies that $\mathbb{E}[|A^{(t)}|] \leq (1+4\epsilon)^T \leq e^{4\epsilon T} = O(1/\epsilon^4)$. Thus, it now remains to prove~\Cref{eq:recourse:main}.  Towards this end, consider the following scenario where: (i) we fixed the rounds $\{i_{e}\}_e$ and color-sequences $\{c_e\}_e$ of all potential edges at preprocessing, (ii) we have already computed the outcome of~\Cref{alg:nibble} on $G^{(t-1)}$ in accordance with these choices $\{i_e, c_e\}_e$, and (iii) now we are running~\Cref{alg:nibble} again on $G^{(t)}$ with the same choices $\{i_e, c_e\}_e$ and observing which edges get added to $A^{(t)}$ as~\Cref{alg:nibble} implements the rounds $i = 1, \ldots, T$ on $G^{(t)}$, in this order. 

It is easy to observe that no edge gets added to $A^{(t)}$ during rounds $1, \ldots, i_{e^{\star}}-1$. At round $i = i_{e^{\star}}$, the edge $e^\star$ is now present and it receives a tentative color (as long as its palette is non-empty). Now, consider any subsequent round $i \in \{i_{e^{\star}}+1, \ldots, T\}$. At the start of round $i$, we have already determined the palette $P_i^{(t)}(v)$ for each node $v \in V$. During round $i$, an edge $e = (u, v) \in S_i^{(t)}$ might get added to $A^{(t)}$, but this can happen only due to one of the following two reasons.
\begin{itemize}
\item (I) $c_e(\ell_e^{(t-1)}) \notin P_i^{(t)}(u) \cap P_i^{(t)}(v)$. So, there is an edge $e' \in N^{(t)}_{<i}(e)$ with $\tilde{\chi}^{(t)}(e') = c_{e}( \ell_{e}^{(t-1)}) = \tilde{\chi}^{(t-1)}(e)$.  Since  $\tilde{\chi}^{(t-1)}$ was the output of~\Cref{alg:nibble} on $G^{(t-1)}$, we  have $e' \in A^{(t)}$. In this case, we say that $e'$ is {\em type-I-responsible} for $e$ being added to $A^{(t)}$.  
\item (II) There is an $\ell \in \left[ \ell_e^{(t-1)} - 1\right]$ such that $c_e(\ell) \in P_i^{(t)}(u) \cap P_i^{(t)}(v)$. Let $\ell'$ be the smallest such  $\ell$. Note that~\Cref{alg:nibble} will set $\tilde{\chi}^{(t)}(e) := c_e(\ell')$ and $\ell_e^{(t)} := \ell'$. As $\tilde{\chi}^{(t-1)}$ was the output of~\Cref{alg:nibble} on $G^{(t)}$, there is an  edge $e' \in N^{(t-1)}_{<i}(e)$ with $\tilde{\chi}^{(t-1)}(e') = c_e(\ell') = \tilde{\chi}^{(t)}(e)$. But since $c_e(\ell') \in P_i^{(t)}(u) \cap P_i^{(t)}(v)$, we now have $\tilde{\chi}^{(t)}(e') \neq c_e(\ell')$, which implies that $e' \in A^{(t)}$. In this case, we say that $e'$ is {\em type-II-responsible} for $e$ being added to $A^{(t)}$. 
\end{itemize}
Motivated by the preceding two observations, we now define the following sets for each edge $e' \in  E^{(t)}$: $$\Gamma_i(e') := \left\{ e \in N_i^{(t)}(e')  \, \middle| \, \tilde{\chi}^{(t)}(e') = \tilde{\chi}^{(t-1)}(e) \right\}, \  \Lambda_{i}(e') := \left\{ e \in N_i^{(t)}(e')  \, \middle| \, \tilde{\chi}^{(t-1)}(e') = \tilde{\chi}^{(t)}(e) \right\}.$$

To summarize,  for every edge $e$ that gets added to $A^{(t)}_i$, there is some edge $e' \in A^{(t)}_{<i}$ that is either type-I or type-II responsible for $e$. Furthermore, if $e'$ is type-I (resp.~type-II) responsible for $e$, then we must have $e \in \Gamma_i(e')$ (resp.~$e \in \Lambda_i(e')$).  This leads us to the following observation.

\begin{equation}
\label{eq:bound:recourse:main:101}
|A_i^{(t)}| \ \leq  \sum_{e' \in A_{<i}^{(t)}} \left( |\Gamma_i(e')| + \left| \Lambda_{i}(e') \right| \right) \text{ for all } i \in [i_{e^{\star}}+1, T].
\end{equation}

We now make the following claim, whose proof is deferred to~\Cref{sec:cl:recourse:bound:main:201}.

\begin{claim}
\label{cl:recourse:bound:main:201}  Consider any round $i \in [i_{e^{\star}}+1, T]$ and any edge $e \in S_{<i}^{(t)}$. Then we have: $$\mathbb{E}\left[ |\Gamma_i(e)| + |\Lambda_{i}(e)| \, \middle| \, e \in A_{<i}^{(t)} \right] \leq 4\epsilon.$$
\end{claim}

For the rest of the proof, fix any round $i \in [i_{e^{\star}}+1, T]$. Observe that:
\begin{eqnarray*}
\mathbb{E}\left[ |A_i^{(t)}| \right] & \leq & \sum_{e \in S_{<i}^{(t)}} \mathbb{E}\left[ |\Gamma_i(e)| + |\Lambda_i(e)| \, \middle| \, e \in A_{<i}^{(t)} \right] \cdot \Pr\left[ e \in A_{<i}^{(t)} \right] \\
& \leq & 4\epsilon \cdot \sum_{e \in S_{<i}^{(t)}} \Pr\left[ e \in A_{<i}^{(t)}\right]  =  4\epsilon \cdot \mathbb{E}\left[|A_{<i}^{(t)}|\right].
\end{eqnarray*}
In the above derivation, the first inequality follows from~\Cref{eq:bound:recourse:main:101},  whereas the second inequality follows from~\Cref{cl:recourse:bound:main:201}. Thus, ~\Cref{eq:recourse:main} holds, and this concludes the proof of the lemma.
\end{proof}

\subsubsection{Proof of~\Cref{cl:recourse:bound:main:201}}
\label{sec:cl:recourse:bound:main:201}

The occurrence of  the event $\left\{ e \in A^{(t)}_{< i} \right\}$ depends only on the color-sequences of those edges $e' \in S^{(t)}_{<i}$ that are in the same connected component as $e$ in  $G^{(t)}_{<i} := (V, S^{(t)}_{<i})$.  For the rest of the proof,  we fix these relevant color-sequences in such a way that the event $\left\{ e \in A^{(t)}_{< i} \right\}$ occurs. Below, we will show that $ \mathbb{E}\left[ |\Gamma_i(e)| \right] \leq 2 \epsilon$. The proof for $ \mathbb{E}\left[ |\Lambda_i(e)| \right] \leq 2\epsilon$ is completely analogous, and taken together, they imply~\Cref{cl:recourse:bound:main:201}.

By linearity of expectation, we have:
\begin{equation}
\label{eq:recourse:key:1}
\mathbb E \left[ |\Gamma_i^{(t)}(e) | \right] = \sum_{f \in N^{(t)}_{i}(e)} \Pr \left[\tilde \chi^{(t-1)}(f) = \tilde \chi^{(t)}(e) \right].
\end{equation}
Let $e = (u,v)$. The palettes $P_i^{(t-1)}(u)$ and $P^{(t-1)}_i(v)$ are completely determined by the color-squences we have fixed, because  $G^{(t-1)}_{<i}$ is a subgraph of $G^{(t)}_{<i}$ as we are considering the scenario where the $t^{th}$ update is an edge-insertion. Consider an edge $f = (u,w) \in N^{(t-1)}_i(e)$ (note that $N^{(t)}_i(e) = N^{(t-1)}_i(e)$ since $i_{e^\star} < i$). Since the graph $G^{(t)}_i$ contains no cycles, $e$ and $w$ lie in different connected components of $G^{(t)}_{<i}$. Thus, the color $\tilde{\chi}^{(t)}(e)$ is  independent of the palette $P_i^{(t-1)}(w)$. Let $c = \tilde \chi^{(t)}(e)$.  We can now apply \Cref{lem:main:strong sym} and \Cref{cor:lem:trees:bound:1} to get:\sayan{We need to make this more formal before arxiv upload.} 
\begin{equation}
\label{eq:recourse:key:2}
\Pr \left[\tilde \chi^{(t-1)}(f) = c  \right] \leq \frac{1}{|P_i^{(t-1)}(u)|} \leq \frac{1}{(1+\epsilon) (1-\epsilon)^{i-1}\Delta}.
\end{equation}
From~\Cref{eq:recourse:key:1} and~\Cref{eq:recourse:key:2}, we now derive that:
\begin{eqnarray*}
\mathbb E \left[ |\Gamma_i^{(t)}(e) | \right] & = &\sum_{f \in N^{(t)}_{i}(e)} \Pr \left[\tilde \chi^{(t-1)}(f) = c \right] \\
& = &  \sum_{f \in N^{(t-1)}_{i}(u)} \Pr \left[ \tilde \chi^{(t-1)}(f) = c \right] + \sum_{f \in N^{(t-1)}_{i}(v)} \Pr \left[ \tilde \chi^{(t-1)}(f) = c \right] \\
& \leq & \frac{|N^{(t-1)}_{i}(u)| + |N^{(t-1)}_{i}(v)|}{(1+\epsilon)(1-\epsilon)^{i-1}\Delta} \leq \frac{2\epsilon(1+\epsilon)(1-\epsilon)^{i-1}\Delta}{(1+\epsilon) (1-\epsilon)^{i-1}\Delta}  = 2\epsilon.
\end{eqnarray*}
The penultimate step in the above derivation follows from~\Cref{lem:trees:bound:1}.

\subsection{Bounding the Update Time}
\label{sec:main:static:updatetime}

We now outline how to implement a  modified version of the algorithm in~\Cref{sec:main:static:recourse},  which leads  to dynamic  $(1+\epsilon)\Delta$-edge coloring in $O(\text{poly}(1/\epsilon))$ update time (see~\Cref{thm:dynamic:main}).

\medskip
\noindent {\bf A ``template'' algorithm.} We start with a  {\em template algorithm}, which differs from the one in~\Cref{sec:main:static:recourse} in only one aspect: The color-sequence of every potential edge is now of length $K := \lceil (8/\epsilon^2) \log (1/\epsilon) \rceil$. Specifically, at preprocessing each potential edge $e$ constructs the color-sequence $c_e$ by sampling $K$, as opposed to infinitely many, colors $c_e(1), \ldots, c_e(K)$ from the palette $\mathcal{C} = [(1+\epsilon)\Delta]$, independently and u.a.r. Subsequently, the algorithm handles the graph $G^{(t)}$ at time $t$ as follows. While executing round $i \in [T]$ of the {\sc Nibble} algorithm on $G^{(t)}$,  an edge $e \in S_i^{(t)}$ picks the minimum index $\ell \in [K]$ s.t.~$c_e(\ell) \in P_i^{(t)}(u) \cap P_i^{(t)}(v)$, and sets $\tilde{\chi}^{(t)}(e) := c_e(\ell)$ and $\ell_e^{(t)} := \ell$. If there is no such  $\ell$, then it sets $\tilde{\chi}^{(t)}(e) := \bot$ and $\ell_e^{(t)} := 0$. Everything else remains the same as  in~\Cref{sec:main:static:recourse}. We now give an intuitive justification as to why the above modification does not (meaningfully) change the guarantees derived in~\Cref{sec:main:static:recourse} (see~\Cref{app:static} for a formal argument).

\medskip
\noindent {\bf An intuitive justification.} 
Fix the rounds $\{i_e\}_e$ of all potential edges $e \in {V \choose 2}$, and condition on the high-probability event $\mathcal{Z}$ (see~\Cref{lem:trees:bound:1}) on the current graph $G^{(t)}$. Consider any round  $i \in [T]$  and any edge $e 
\in S_i^{(t)}$. By \Cref{cor:edge-palette}, we have that $|P_i(e)|\geq \epsilon^2(1 + \epsilon) \Delta/8$ w.h.p.
As $|\mathcal C| = (1+\epsilon)\Delta$, the probability that the color-sequence $c_e$ does not contain some color from $P_i(e)$ is given by: $\left(1-\frac{\epsilon^2(1 + \epsilon)\Delta}{8 \cdot |\mathcal{C}|}\right)^K \leq \left(1- \frac{\epsilon^2}{8}\right)^K \leq \left(1-\frac{\epsilon^2}{8}\right)^{(8/\epsilon^2) \log (1/\epsilon)} \leq \epsilon$.  To summarize, with probability at most $\epsilon$, the template algorithm mistakenly sets $\tilde{\chi}^{(t)}(e) = \bot$. But with probability $1-\epsilon$, it correctly sets $\tilde{\chi}^{(t)}(e)$ to be a color chosen u.a.r.~from $P_i^{(t)}(e)$.  Thus,  the template algorithm is  almost similar to the original algorithm, except that each edge $e \in E^{(t)}$, before even being considered for receiving a tentative color, gets added to the failed set $F^{(t)}$ with probability $\epsilon$. This increases the maximum degree of the subgraph $G_F^{(t)} := (V, F^{(t)})$ by at most $\epsilon \Delta$ (we can show that this bound holds w.h.p.). Since anyway we maintain a $O\left( G_F^{(t)} \right)$-coloring in $G_F^{(t)}$  using a separate palette, the total number of colors needed by the algorithm continues to remain $(1+O(\epsilon))\Delta$.

\medskip
\noindent {\bf Data structures.} In order to highlight the main ideas behind our data structures, we start by making two simplifications. (I) We allow for a preprocessing time of $O_{\epsilon}(n^2)$. This means that we can afford to fix the rounds $\{i_e\}_e$ and the color-sequences $\{ c_e\}_e$ for every potential edge $e \in {V \choose 2}$ at preprocessing. (II) We focus only on maintaining the tentative coloring $\tilde{\chi}^{(t)}$ on $G^{(t)}$. This allows us to ignore keeping track of the set of failed edges $F^{(t)}$, and the coloring $\psi^{(t)}$ on $G_F^{(t)} := (V, F^{(t)})$. 

Towards the end of this section, we explain how we can easily get rid of the preprocessing time, provided we start with an empty graph $G^{(0)} = (V, E^{(0)})$. We defer presenting the data structures which maintain the coloring $\psi^{(t)}$ to the full version in~\Cref{sec:app:datastructs}.

We are now ready to define the key data structure. For all nodes $v \in V$, rounds $i \in [T]$ and colors $c \in \mathcal{C}$, let $\mathcal{L}_{v, i}^{(t)}(c) := \left\{ e \in N_i^{(t)}(v) \, \middle| \, \exists k \in [K] \text{ s.t. } c_e(\ell_k) = c \right\}$ denote the set of all neighboring edges $e \in N_i^{(t)}(v)$ such that the color $c$ appears at least once in  the color-sequence $c_e$. We store all these sets $\left\{ \mathcal{L}_{v, i}^{(t)}(c)\right\}_{v, i, c}$ as hash-tables. In the two claims below, we bound the sizes of these sets. 

\begin{claim}
\label{cl:ds:1}
For each $v \in V, i \in [T]$ and $c \in \mathcal{C}$, we have $\mathbb{E}\left[ | \mathcal{L}^{(t)}_{v, i}(c)| \right] = K = O_{\epsilon}(1)$. 
\end{claim}

\begin{proof}(Sketch)
Consider any edge $e \in N_i^{(t)}(v)$. The expected number of times the color $c$ appears in the sequence $c_e$ is given by $K/|\mathcal{C}| = K/((1+\epsilon)\Delta) \leq K/\Delta$. Hence, by Markov's inequality, we have $\Pr\left[e \in \mathcal{L}_{v, i}^{(t)}(c)\right] \leq K/\Delta$.  Since $|N_i^{(t)}(v)| \leq \Delta$, it follows that $\mathbb{E}\left[ | \mathcal{L}_{v, i}^{(t)}(c) | \right] \leq (K/\Delta) \cdot \Delta = K$.
\end{proof}

\begin{claim}
\label{cl:ds:2}
We always have $\sum_{v \in V, i \in [T], c \in \mathcal{C}} |\mathcal{L}_{v, i}^{(t)}(c)| = O\left(KT \cdot |E^{(t)}|\right) = O_{\epsilon}\left( |E^{(t)}| \right)$. Further, insertion/deletion of an edge $e$ in $G$ can lead to at most $2KT = O_{\epsilon}(1)$ changes in the sets $\left\{ \mathcal{L}_{v, i}^{(t)}(c)\right\}_{v, i, c}$.
\end{claim}

\begin{proof}(Sketch)
Each edge $e = (u, w) \in E^{(t)}$ can appear in at most $2KT$ sets $\left\{ \mathcal{L}_{v, i}^{(t)}(c)\right\}_{v, i, c}$, one for each of its endpoints $x \in \{u, w\}$, for each round $i \in [T]$, and most importantly, for each of the (at most) $K$ colors that appear in its color-sequence $c_e$. For the same reason, it is also the case that the insertion/deletion of an edge $(u, v)$ in $G$ can  affect at most $2KT$ of the sets $\left\{ \mathcal{L}_{v, i}^{(t)}(c)\right\}_{v, i, c}$.
\end{proof}

\Cref{cl:ds:2} implies that we can maintain the hash-tables for $\left\{ \mathcal{L}_{v, i}^{(t)}(c)\right\}_{v, i, c}$ in $O_{\epsilon}(1)$ update time (note that these hash-table don't depend at all on the tentative coloring $\tilde{\chi}^{(t)}$), and that the total space complexity of this data structure is $O_{\epsilon}\left( |E^{(t)}|\right)$. We now show how using this data structure we can implement the template algorithm in expected update time proportional to its recourse, which in turn, is $O(1/\epsilon^4)$ according to~\Cref{thm:main:recourse}.

\medskip
\noindent {\bf Modifying the  coloring $\tilde{\chi}$ after the $t^{th}$ update.} Recall the notations from~\Cref{sec:main:static:recourse}, and consider the following scenario. We already have the coloring $\tilde{\chi}^{(t-1)}$ on $G^{(t-1)}$, when we receive the $t^{th}$ update to $G$. Our goal now is to change the coloring $\tilde{\chi}^{(t-1)}$ into $\tilde{\chi}^{(t)}$.\footnote{Note that $\tilde{\chi}^{(t-1)}$ (resp.~$\tilde{\chi}^{(t)}$) is the output of the template algorithm on $G^{(t-1)}$ (resp.~$G^{(t)}$).} To achieve this goal, all we need to do is consider the rounds $i = 1, \ldots, T$, one at a time. And while we are at round $i \in [T]$, we need to identify the set $A_i^{(t)}$ and change the color of each edge $e \in A_i^{(t)}$ appropriately. \sayan{Here I assume the edge being inserted is part of $A_i^{(t)}$. Need to change the proof of~\Cref{lem:recourse:2} accordingly.}

We claim that once we have identified an edge $e \in A_i^{(t)}$ during round $i$, changing its color takes only $O(K) = O_{\epsilon}(1)$ time. To see why this is true,  note that we can easily maintain the complements of the palettes $\left\{\overline{P_i(v)}\right\}_v$ for all nodes $v \in V$ as hash-tables, where $\overline{P_i(v)} := \mathcal{C} \setminus P_i(v)$.\footnote{For each color $c \in \overline{P_i(v)}$, maintain a counter which denotes the number of edges in $N_{<i}(v)$ that receive $c$ as their tentative color, and update these counters accordingly whenever an edge changes its tentative color.} Now, for the edge $e = (u, v)  \in A_i^{(t)}$,  we scan through its color-sequence $c_e$ and identify the smallest index (if any) $\ell \in [K]$ such that $c_e(\ell) \notin \overline{P_i^{(t)}(u)} \cup \overline{P_i^{(t)}(v)}$.  This takes $O_{\epsilon}(1)$ time\martin{this is actually $O_\epsilon(1)$, we dont have direct access to the palette, only the blocked colors at each layer. Im happy to not get into this here and leave it as $O(1)$ though.} per index in $[K]$, because at the start of round $i$ we already have access to the hash tables for $\overline{P_i^{(t)}(u)}$ and $\overline{P^{(t)}_i(v)}$. 

Thus, the time spent on implementing round $i$  is dominated by the time spent on identifying the set $A_i^{(t)}$ in the first place.  We will now show that using our data structures we can perform this task in expected time $O_{\epsilon}\left( |A_{<i}^{(t)}|\right)$.  Summing over $i \in [T]$, this will lead to an expected update time of $\sum_{i\in [T]} O_{\epsilon}\left( |A_{<i}^{(t)}|\right) = O_{\epsilon}\left(\sum_{i \in [T]} T \cdot |A_i^{(t)}| \right) = O_{\epsilon}(|A^{(t)}|)$, which is $O_{\epsilon}(1)$ in expectation, according to~\Cref{lem:recourse:2}. It now only remains to prove the claim below.

\begin{claim}
\label{cl:ds:main}
While implementing round $i$ of the template algorithm after the $t^{th}$ update, we can identify the set $A_i^{(t)}$ in expected time $O_{\epsilon}(|A_{<i}^{(t)}|)$. 
\end{claim}

\begin{proof}(Sketch)
From the proof of~\Cref{lem:recourse:2} (specifically, from the discussion in the paragraph preceding~\Cref{eq:bound:recourse:main:101}), it follows that an edge $e \in S_i^{(t)}$ belongs to $A_i^{(t)}$ only if it has a neighbor $e' \in N_{<i}^{(t)}(e) \cap A_{<i}^{(t)}$ such that either $e \in \Gamma_i^{(t)}(e')$ or $e \in \Lambda_i^{(t)}(e')$. Note that in the former (resp.~latter) case, the color $\tilde{\chi}^{(t)}(e')$ (resp.~$\tilde{\chi}^{(t-1)}(e')$) must  appear at least once in the color-sequence $c_e$. In other words, an edge $e = (u, v) \in S_i^{(t)}$ appears in $A_i^{(t)}$ only if it has a neighboring edge $e' = (v, w) \in N_{<i}^{(t)}(e) \cap A_{<i}^{(t)}$ such that $e \in \mathcal{L}_{v, i}^{(t)}\left( \tilde{\chi}^{(t)}(e')\right) \cup \mathcal{L}_{v, i}^{(t)}\left( \tilde{\chi}^{(t-1)}(e')\right)$. This motivates the following approach for identifying the set $A_i^{(t)}$ at the start of round $i$.  Construct a set $\mathcal{A}_i^{(t)}$ by taking the union, over all $e' = (v, w) \in A_{<i}^{(t)}$, of the sets $\mathcal{L}_{v, i}^{(t)}\left( \tilde{\chi}^{(t)}(e')\right) \cup \mathcal{L}_{v, i}^{(t)}\left( \tilde{\chi}^{(t-1)}(e')\right)$. By~\Cref{cl:ds:1}, this takes $O_{\epsilon}\left( |A_{<i}^{(t)}| \right)$ expected time.  We also know for sure that  $A_i^{(t)} \subseteq \mathcal{A}_i^{(t)}$, and $\mathbb{E}\left[ |\mathcal{A}_i^{(t)}| \right] = O_{\epsilon}\left(|A_{<i}^{(t)}|\right)$. 

We now scan through all the edges $e \in \mathcal{A}_i^{(t)}$, and for each of them determine whether or not it belongs to $A_i^{(t)}$. This takes $O(K) = O_{\epsilon}(1)$ expected time per edge in $\mathcal{A}_i^{(t)}$, because given an edge $e = (u, v) \in \mathcal{A}_i^{(t)}$ all we need to do is to check which of the colors $\{ c_e(\ell)\}_{\ell \in [K]}$ belong to $P_i^{(t)}(u) \cap P_i^{(t)}(v) = \mathcal{C} \setminus \left( \overline{P_i^{(t)}(u)} \cup \overline{P_i^{(t)}(v)}\right)$,  and this takes $O(1)$ expected time per color. 

To summarize, we can indeed identify the set $A_i^{(t)}$ in $O_{\epsilon}\left(|A_{<i}^{(t)}|\right)$ expected time.
\end{proof}

\medskip
\noindent {\bf Getting rid of the preprocessing time.} Suppose that we start with an empty graph $G^{(0)} := (V, E^{(0)})$ before any update arrives. We can get rid of the $O_{\epsilon}(n^2)$ preprocessing by making the following simple observation: We don't need to fix the rounds and color-sequences of potential edges in advance. We can do that {\em on the fly}. Specifically, when an edge $e$ gets inserted, we sample its round $i_e$ and color-sequence $c_e$ in $O_{\epsilon}(1)$ time.  We work with $(i_e, c_e)$ as long as the edge $e$ is present in the input graph. When the edge $e$ gets deleted, we remove all information/storage about $(i_e, c_e)$. If the edge $e$ gets inserted again at some point in future, then we sample a fresh pair $(i_e, c_e)$ and work with this new sample.  It is easy to verify that all our recourse and update time bounds, as well as the guarantee on the total number of colors used, continue to hold if we make this simple modification. The only effect it has is that now we don't need to worry about preprocessing time at all, provided we are starting with an empty graph.

\bibliographystyle{alpha}
\bibliography{bibl}

\section*{Acknowledgements}

Shay Solomon is funded by the European Union (ERC, DynOpt, 101043159).
Views and opinions expressed are however those of the author(s) only and do not necessarily reflect those of the European Union or the European Research Council. Neither the European Union nor the granting authority can be held responsible for them.
Shay Solomon and Nadav Panski are supported by the Israel Science Foundation (ISF) grant No.1991/1.
Shay Solomon is also supported by a grant from the United States-Israel Binational Science Foundation (BSF), Jerusalem, Israel, and the United States National Science Foundation (NSF).

\appendix

\section{Our Static Algorithm (Full Version)}\label{app:static}
Throughout this paper, we fix some constant $\epsilon$ such that $0 < \epsilon \leq 1/10$, and define parameters $T := \lfloor(1 / \epsilon) \log(1/\epsilon) \rfloor$ and $K := \lceil (8/\epsilon^2 )\log(1/\epsilon)\rceil$. In this appendix, we describe our static algorithm in full detail and provide a complete analysis. The main result in this appendix is \Cref{app:thm:static}, which is restated below.

\begin{thm}
    Let $\epsilon \in (0,1/10)$ be a constant. Then, given a graph $G$ and a parameter $\Delta \geq (100 \log n/\epsilon^4)^{(30/\epsilon) \log (1/\epsilon)}$ such that the maximum degree of $G$ is at most $\Delta$, the algorithm \textsc{StaticColor} produces a $(1 + 61\epsilon)\Delta$-edge coloring of $G$ with probability at least $1 - 8/n^{6}$.
\end{thm}

\subsection{Algorithm Description}



\textbf{The algorithm \textnormal{\textsc{StaticColor}}.} Our main algorithm, \textsc{StaticColor}, takes as input a graph $G=(V,E)$ with $n$ nodes, and a parameter $\Delta \in \mathbb N$. The algorithm uses two subroutines: \textsc{Partition} and \textsc{Nibble}. We first use the algorithm \textsc{Partition} to split the input graph $G$ into $\eta = \lceil \Delta^{1 - 1/(30T)} \rceil$ many graphs $\mathcal G_1,...,\mathcal G_\eta$ by assigning the edges of $G$ to one of the $\mathcal G_j$ independently and uniformly at random. We then proceed to call the algorithm \textsc{Nibble} on each of the $\mathcal G_j$, where we use $(1 + \epsilon)^2\Delta^{1/(30T)}$ many colors to color most of the edges in $\mathcal G_j$. Finally, we take all of the edges that failed to be colored by our calls to \textsc{Nibble} and greedily color them. Algorithm \ref{app:alg:static} gives the pseudocode for \textsc{StaticColor}.

\begin{algorithm}[H]
 \caption{\textsc{StaticColor}$(G, \epsilon)$}\label{app:alg:static}
\begin{algorithmic}[1]
    \State $\mathcal G_1,\dots, \mathcal G_\eta \leftarrow \textsc{Partition}(G, \Delta, \Delta^{1/(30T)})$
    \For{$j = 1,...,\eta$}
        \State $\mathcal F_j \leftarrow \textsc{Nibble}(\mathcal G_j,(1 + \epsilon)\Delta^{1/(30T)},\epsilon)$
    \EndFor
    \State Greedily color the edges in $H \leftarrow \mathcal F_1 \cup \dots \cup \mathcal F_\eta$ using $3\Delta(H)$ many colors
\end{algorithmic}
\end{algorithm}

\noindent
For analytic purposes, and also to make the algorithm easier to present and dynamize, it will be useful for us to fix the randomness used by our algorithm while making calls to \textsc{Partition} and \textsc{Nibble} in advance. Hence, we will describe how all the relevant randomness is generated in advance while describing each part of our algorithm. The assumption can be removed later.

\medskip
\noindent \textbf{The algorithm \textnormal{\textsc{Partition}}.} This algorithm takes as input a graph $G$ and some parameters $\Delta$ and $\Delta'$ such that $\Delta' \leq \Delta$,
and outputs a partition $\mathcal G_1,..., \mathcal G_\eta$ of $G$, where $\eta = \lceil \Delta / \Delta' \rceil$, obtained by assigning the edges of $G$ to one of the $\mathcal G_j$ independently and uniformly at random. The node sets $V(\mathcal G_j)$ in all of the $\mathcal G_j$ are the same as the node set $V$ of the graph $G$.
In order to fix the randomness in advance, for each \textit{potential edge} of the graph $e \in \binom{V}{2}$, we sample an index $j_e \in [\eta]$ independently and uniformly at random, and place $e$ into the graph $\mathcal G_{j_e}$. Notice that after fixing the random variables $\{j_e\}_{e}$ the partition of the graph depends only on what edges are present in $G$. Algorithm \ref{app:alg:split} gives the pseudocode for \textsc{Partition}.

\begin{algorithm}[H]
\caption{\textsc{Partition}$(G, \Delta, \Delta')$}\label{app:alg:split}
\begin{algorithmic}[1]
    \State $\eta \leftarrow \lceil \Delta / \Delta' \rceil$
    \State $\mathcal E_j \leftarrow \varnothing$ for all $j \in [\eta]$
    \For{$e \in E$}
        \State Place $e$ into one of $\mathcal E_1,\dots, \mathcal E_\eta$ independently and u.a.r.\;
    \EndFor
    $\mathcal G_j \leftarrow (V, \mathcal E_j)$ for all $j \in [\eta]$
    \State \Return{$\mathcal G_1,\dots, \mathcal G_\eta$}
\end{algorithmic}
\end{algorithm}

\medskip
\noindent \textbf{The algorithm \textnormal{\textsc{Nibble}}.} This algorithm takes as input a graph $G$, a parameter $\Delta$, and uses a palette $\mathcal C$ of $\lceil (1 + \epsilon)\Delta \rceil$ colors. The {\sc Nibble} algorithm  runs for $T :=  \lfloor(1 / \epsilon) \log(1/\epsilon) \rfloor$ {\em rounds}.

At the start of round $i \in [T]$, we have a subset of edges $E_i \subseteq E$ such that the algorithm has already assigned tentative colors to the remaining edges $E \setminus E_i$. We denote this {\em tentative} partial coloring by $\tilde{\chi} : E \setminus E_i \rightarrow \C \cup \{ \bot \}$, which need not necessarily be proper. For each node $u \in V$, we refer to the set of colors $P_i(u) :=  \C \setminus \tilde \chi(N(u) \setminus E_i)$ as the {\em palette} of $u$ at the start of round $i$, where $N(u) \subseteq E$ denotes the set of edges incident on $u$ in $G$. In words, the palette $P_i(u)$ consists of the set of colors that were {\em not} tentatively assigned to any incident edge of $u$ in previous rounds. We also define $P_i(u, v) := P_i(u) \cap P_i(v)$ to be the {\em palette} of any edge $(u, v) \in E_i$ at the start of round $i$.

We start by initializing $E_1 \leftarrow E$, and $P_1(u) \leftarrow \C$  for all $u \in V$. Subsequently, for $i = 1, \ldots, T$, we implement round $i$ as follows. Each edge $e \in E_i$ {\em selects} itself independently with probability $\epsilon$. Let $S_i \subseteq E_i$ be the set of selected edges.
Next, in parallel, each edge $e \in S_i$ samples $K := \lceil (8/\epsilon^2) \log(1/\epsilon) \rceil$ colors $c_1,\dots,c_K$ independently and uniformly at random from $\mathcal C$ and sets its tentative color $\tilde \chi(e)$ to some $c_\ell$ which is contained in $P_i(e)$ if such a color exists. Otherwise, we set $\tilde{\chi}(e) \leftarrow \bot$.
At this point, we define the collection $F_i \subseteq S_i$ of {\em failed} edges in round $i$. We say that an edge $e = (u, v) \in S_i$ {\em fails} in round $i$ iff either (i) $\tilde{\chi}(e) = \bot$, or (ii) there is a neighboring edge $f \in (N(u) \cup N(v)) \cap S_i$  which was also selected in round $i$ and received the same tentative color as the edge $e$ (i.e., $\tilde{\chi}(e) = \tilde{\chi}(f)$). Let $F_i \subseteq S_i$ denote this collection of failed edges (in round $i$). We now set $E_{i+1} \leftarrow E_i \setminus S_i$  and proceed to the next round $i+1$.

In order to fix the randomness in advance, for each potential edge of the graph $e$, we sample a round for the edge $e$ independently from $\textsc{CappedGeo}(\epsilon, T+1)$, which we denote by $i_e$.\footnote{Thus, we have $\Pr[i_e = i] = (1-\epsilon)^{i-1}\epsilon$ for all $i \in [T]$ and $\Pr[i_e = T+1] = (1-\epsilon)^T$.} We also assume that for each potential edge $e$ we have some colors $c_e(1),...,c_e(K)$ where each $c_e(\ell)$ is sampled uniformly at random and independently from $\mathcal C$. We define $\ell_e$ to be $\min \{\ell \, | \, c_e(\ell) \in P_{i_e}(e)\}$ (taking the convention that $\min \varnothing = 0$) and note that $\tilde \chi(e) = c_e(\ell_e)$ (as long as $c_e \cap P_{i_e}(e) \neq \varnothing$). Algorithm~\ref{app:alg:nibble} gives the pseudocode for \textsc{Nibble}.

To ease notations, at the end of the last round $T$ we define $F_{T+1} \leftarrow E_{T+1}$, and $\tilde{\chi}(e) \leftarrow \bot$ for all $e \in F_{T+1}$. We let  $F := \bigcup_{i=1}^{T+1} F_i$ denote the set of failed edges across all the rounds. We denote $N(u) \cap S_i$ by $N_i(u)$. It is easy to check that the tentative coloring $\tilde{\chi}$, when restricted to the edge-set $E \setminus F$, is already proper.

\begin{algorithm}[H]
\caption{\textsc{Nibble}$(G, \Delta, \epsilon)$}\label{app:alg:nibble}
\begin{algorithmic}[1]
    \State $\chi(e) \leftarrow \perp$ and $\tilde \chi(e) \leftarrow \perp$ for all $e \in E(G)$
    \State $E_1 \leftarrow E(G)$
    \For{$i = 1,...,T$}
        \For{$e \in S_i$}
            \State $\ell_e \leftarrow \min \{\ell \; | \; c_e(\ell) \in P_i(e)\}$
            \If{$1 \leq \ell_e \leq K$}
                \State $\tilde \chi(e) \leftarrow c_e(\ell_e)$
            \EndIf
        \EndFor
        \State $F_i \leftarrow \{ e \in S_i \, | \, \exists f \in N(e) \cap S_i \textrm{ such that } \tilde \chi(f) = \tilde \chi(e)\} \cup  \{ e \in S_i \, | \, \tilde \chi(e) = \perp \} $
	\State $\chi(e) \leftarrow \tilde \chi(e)$ for all $e \in S_i \setminus F_i$
        \State$E_{i+1} \leftarrow E_i \setminus S_i$
    \EndFor
    \State $F_{T+1} \leftarrow E_{T+1}$ 
    \State$ F \leftarrow \bigcup_{i=1}^{T+1} F_i$
    \State \Return {$F$}
\end{algorithmic}
\end{algorithm}

\noindent We now prove the following result which describes the behavior of \textsc{StaticColor}.

\begin{thm}\label{app:thm:static}
    Let $\epsilon \in (0,1/10)$ be a constant. Then, given a graph $G$ and a parameter $\Delta \geq (100 \log n/\epsilon^4)^{(30/\epsilon) \log (1/\epsilon)}$ such that the maximum degree of $G$ is at most $\Delta$, the algorithm \textsc{StaticColor} produces a $(1 + 61\epsilon)\Delta$-edge coloring of $G$ with probability at least $1 - 8/n^{6}$.
\end{thm}


\subsection{Analysis on Locally Treelike Graphs}\label{sec:anal}

For the rest of \Cref{sec:anal}, fix some graph $G=(V,E)$ of maximum degree at most $\Delta$ such that $\Delta \geq (100 \log n )/ \epsilon^4$. Let $\mathcal{N}_G(u, j) := G\left[\{ v \in V : \text{dist}_G(u, v) \leq j\}\right]$ denote the $j$-hop neighborhood of any node $u \in V$.\footnote{We use the symbole $\text{dist}_G(u, v)$ to denote the distance between $u$ and $v$ in $G$. Furthermore, for any subset of nodes $V' \subseteq V$, the symbol $G[V']$ denotes the subgraph of $G$ induced by $V'$ and $\mathcal{N}_G(V', j)$ denotes $\bigcup_{u \in V'}\mathcal{N}_G(u, j)$.} We refer to $\mathcal{N}_G(u, j)$ as the $j$-neighborhood of $u$. Let $U \subseteq V$ be the set of nodes $u \in V$ such that the $(T + 1)$-neighborhood of $u$ is a tree. We refer to the nodes in $U$ as \emph{good} and the nodes in $V \setminus U$ as \emph{bad}. Now suppose we run \Cref{app:alg:nibble} on input $(G, \Delta, \epsilon)$. We now proceed to analyze the behavior of the algorithm on this input.

\subsubsection{Locality of the Nibble Method}


\begin{lem}\label{lem:locality of nibble}
    Let $u \in V$, $i \in [T]$. Then $P_i(u)$ depends only on $\mathcal N(u, i-1)$.
\end{lem}

\begin{proof}
We first recall the way our algorithm generates random bits in advance. We assign each \emph{potential edge} $e \in \binom{V}{2}$ a round $i_e$ sampled from $\textsc{CappedGeo}(\epsilon, T+1)$ independently. Furthermore, we assign each potential edge a sequence of colors $c_e(1),\dots c_e(K)$ generated by sampling colors independently and u.a.r from $\mathcal C$.
We now argue inductively that the set $P_i(u)$ depends only on $\mathcal N(u, i-1)$, i.e. that the structure of the graph outside of the $(i-1)$-neighborhood of $u$ has no effect on the palette of $u$ during the first $i-1$ rounds of the nibble method.

Clearly, $P_1(u) = \mathcal C$ for all nodes $u$ and hence does not depend on anything. Now, for the induction step, suppose that $P_i(u)$ depends only on $\mathcal N(u, i-1)$ for all nodes $u$. Given some node $u$, note that $P_{i+1}(u)$ depends on (i) the rounds of the edges incident on $u$, (ii) the sequences of colors assigned to the edges incident on $u$, and (iii) the palettes $P_i(v)$ of the nodes adjacent to $u$. Since (i) and (ii) are fixed in advance, $P_{i+1}(u)$ is completely determined by which nodes are its neighbors and their palettes at the end of iteration $i$. Hence, by the induction hypothesis, it follows that $P_{i+1}(u)$ depends on  $\bigcup_{v \in N(u)} \mathcal N(v, i-1)$ which is exactly the $i$-neighborhood of $u$, $\mathcal N(u,i)$.
\end{proof}

\noindent
It immediately follows from Lemma \ref{lem:locality of nibble} that, for an edge $e$, the palette  $P_i(e)$ depends only on $\mathcal N(e, i-1)$.
Since the tentative color assigned to $e$ depends only on the round $i_e$ and the palette $P_{i_e}(e)$, it also follows that $\tilde \chi(e)$ depends only on $\mathcal N(e, i_e-1) \subseteq \mathcal N(e, T-1) $. Hence, given some node $u$, the tentative colors assigned to edges in $N(u)$ depend only on $\mathcal N(N(u), T-1) \subseteq \mathcal N(u, T)$. Since the edges in $N(u)$ that fail depend on the tentative colors of the edges in $N^2(U) := N(N(U))$, while analyzing the palettes, tentative colors, and failure status of the edges in $N(u)$ we can assume that the input graph is $\mathcal N(N^2(u), T-1) = \mathcal N(u, T+1)$. In the case where $u$ is good, this will allow us to significantly simplify the analysis.

\subsubsection{The Symmetry of Algorithm \ref{alg:nibble}}\label{app:symmetry}

We now prove some lemmas that will be crucial for analyzing the types of colorings generated by our algorithm. For the rest of \Cref{app:symmetry} we fix the random bits that determine the rounds of all edges in $G$. The following lemma formalizes the main ``symmetry'' property of \textsc{Nibble} which says that it treats all colors equally.

\begin{lem}\label{lem:strong sym}
    For all $u \in V$, $i \in [T]$, $C \subseteq \mathcal C$, and permutations $\pi : \mathcal C \longrightarrow \mathcal C$, we have that
    \[\Pr \left[ P_i(u) = C \right] = \Pr \left[ \pi \left( P_i(u) \right) = C\right].\]
\end{lem}

\begin{proof}
    We prove this lemma via a coupling argument.
    Let $\mathcal A$ denote the algorithm \textsc{Nibble}. Given some permutation $\pi$ of $\mathcal C$, we define a new algorithm $\mathcal A^\pi$, which behaves in the exact same way as $\mathcal A$, except that given the color sequence $c_e$ for the potential edge $e$, it uses the color sequences $\pi \circ c_e$ instead. We denote the palette $P_i(u)$ produced by algorithm $\mathcal A'$ by $P^{(\mathcal A')}_i(u)$.
    We now prove by induction that for all $u \in V$, $i \in [T]$, and permutations $\pi$ of $\mathcal C$, we have that
    \begin{equation}\label{eq:sym}
        P^{(\mathcal A^\pi)}_i(u) = \pi \left(P^{(\mathcal A)}_i(u)\right).
    \end{equation}
    Fix some permutation $\pi$ of $\mathcal C$. It's easy to see that (\ref{eq:sym}) holds for all $u \in V$ when $i = 1$ since the palettes all equal $\mathcal C$ and $\pi$ is a permutation. Now suppose that (\ref{eq:sym}) holds for some $i \in [T-1]$ and all $u \in V$. Then, for all $e = (u,v) \in S_i$, we have that
    \[P^{(\mathcal A^\pi)}_i(e) = P^{(\mathcal A^\pi)}_i(u) \cap P^{(\mathcal A^\pi)}_i(v) = \pi \left(P^{(\mathcal A)}_i(u)\right) \cap \pi \left(P^{(\mathcal A)}_i(v)\right) \]
    \[= \pi \left(P^{(\mathcal A)}_i(u) \cap P^{(\mathcal A)}_i(v)\right) = \pi \left(P^{(\mathcal A)}_i(e)\right),\]
    thus, it follows that
    \[c_e(\ell) \in P^{(\mathcal A)}_i(e) \text{ iff } \pi(c_e(\ell)) \in \pi \left( P^{(\mathcal A)}_i(e) \right) \text{ iff } \pi(c_e(\ell)) \in  P^{(\mathcal A^\pi)}_i(e). \]
    We get that $\ell^{(\mathcal A)}_e = \ell^{(\mathcal A^\pi)}_e$ for all $e \in S_i$, which implies that 
    \[\tilde \chi^{(\mathcal A^\pi)}(e) = \pi \left(c_e(\ell^{(\mathcal A^\pi)}_e) \right) = \pi \left(c_e(\ell^{(\mathcal A)}_e) \right) = \pi \left(\tilde \chi^{(\mathcal A)}(e) \right)\]
    for all $e \in S_i$.\footnote{Here $\ell^{(\mathcal A')}_e$ and $\tilde \chi^{(\mathcal A')}(e)$ denote the color index $\ell_e$ and tentative color $\tilde \chi(e)$ produced by algorithm $\mathcal A'$.}
    Finally, it follows that
    \[ \pi \left(P^{(\mathcal A)}_{i+1}(u) \right) = \pi \left(P^{(\mathcal A)}_{i}(u) \setminus \tilde \chi^{(\mathcal A)}(N_i(u)) \right) = \pi \left(P^{(\mathcal A)}_{i}(u) \right) \setminus \pi \left( \tilde \chi^{(\mathcal A)}(N_i(u)) \right)  \]
    \[ = P^{(\mathcal A^\pi)}_{i}(u) \setminus \tilde \chi^{(\mathcal A^\pi)}(N_i(u)) =  P^{(\mathcal A^\pi)}_{i+1}(u). \]
    This concludes the induction. By now noticing that algorithms $\mathcal A$ and $\mathcal A^\pi$ are actually the same (since applying a permutation to an independent uniform sample returns an independent uniform sample) we get that, for any $C \subseteq \mathcal C$,
    \[\Pr \left[ P^{(\mathcal A)}_i(u) = C \right] = \Pr \left[ P^{(\mathcal A^\pi)}_i(u) = C \right] = \Pr \left[ \pi \left( P^{(\mathcal A)}_i(u) \right) = C\right]\]
    and the lemma follows.
\end{proof}

\noindent
Fix any $i \in [T]$ and let $X^w_c$ be the indicator for the event that $c \in P_i(w)$ for a color $c$ and node $w$. 

\begin{lem}\label{lem:symmetry of nibble}
    Let $u \in V$. Given that $\sum_c X^u_c = \gamma$ for some $\gamma \in \mathbb N$, we have that $\{X^u_c\}_c$ is a permutation distribution (see \Cref{def:perm dist}) where exactly $\gamma$ of the $X^u_c$ are $1$ and the rest are $0$. 
\end{lem}

\begin{proof}
    Since the random variables $X^u_c$ are all indicators and we know that  $\sum_c X^u_c = \gamma$, it follows that exactly $\gamma$ of them are $1$ and the rest are $0$. It then follows from \Cref{lem:strong sym} that, for any $C, C' \subseteq \mathcal C$ of size $\gamma$, $\Pr \left[ P_i(u) = C \right] = \Pr \left[ P_i(u) = C' \right]$, and hence this collection of random variables forms a permutation distribution.
\end{proof}

\begin{lem}\label{lem:symmetry of nibble 4}
    Let $\mathfrak X_1,\dots, \mathfrak  X_\ell$ be families of random variables that depend on disjoint collections of random bits. Then these families of random variables are mutually independent.
\end{lem}

\begin{lem}\label{lem:symmetry of nibble 5}
    Let $u_1,\dots,u_\ell \in V$ be nodes that are mutually disconnected in the graph $(V, S_{< i})$. Then we have that the families of random variables $\{X^{u_1}_c\}_c,\dots,\{X^{u_\ell}_c\}_c$ depend on disjoint collections of random bits.
\end{lem}


\begin{lem}\label{lem:symmetry of nibble 3}
    Let $e = (u,v) \in S_i$, such that $e \in N^2(U)$ and $u$ and $v$ are not connected in the graph $(V, S_{< i})$. Suppose we fix the random bits used in the first $i-1$ rounds that determine the palette $P_i(u)$ (but not those that determine $P_i(v)$) and let $c$ be a color that does not depend on the random bits that determine $P_i(u)$.
    Then we have that $\Pr [\tilde \chi(e) = c] \leq 1/|P_i(u)|$.
\end{lem}

\begin{proof}
    First, note that by Lemma \ref{lem:locality of nibble} we can assume that the input graph is $\mathcal N(e, T-1)$. Since $e \in N^2(U)$, $\mathcal N(e, T-1)$ is a subgraph of $\mathcal N(w, T+1)$ for some $w \in U$, and hence $\mathcal N(e, T-1)$ is a tree. It follows that $u$ and $v$ are not connected in the graph $(V, S_{<i})$.
    Note that since $P_i(u)$ and $P_i(v)$ are functions of $\{X^u_c\}_c$ and $\{X^v_c\}_c$ respectively, which depend on disjoint sets of random bits by Lemma \ref{lem:symmetry of nibble 5}, it is possible to fix the random bits this way.
    In the event that $e$ does not fail, $\tilde \chi(e)$ is a uniform sample from $P_i(u) \cap P_i(v)$ by the construction of our algorithm.
    By Lemma \ref{lem:symmetry of nibble 4}, the families of random variables $\{X_c^u\}_c$ and $\{X_c^v\}_c$ are independent, so by Lemma \ref{lem:symmetry of nibble} we get that $P_i(v)$ is a uniform random subset of $\mathcal C$. Hence, as long as $P_i(u) \cap P_i(v) \neq \varnothing$, sampling a color uniformly at random from $P_i(u) \cap P_i(v)$ is the same as sampling a color uniformly at random from $P_i(u)$. It follows that
    \[\Pr \left[ \tilde \chi(e) = c \right] \leq \Pr \left[\tilde \chi(e) = c \, | \, c \in P_i(u) \right] \leq \frac{1}{|P_i(u)|}.\]
\end{proof}

\subsubsection{Concentration of Basic Quantities}\label{app:concetration}

We now establish the concentration of some basic quantities. We first start by analyzing how many edges incident on a node are sampled during a round. For $u \in V$ and $i \in [T]$, let $N_i(u)$ denote the set of edges $N(u) \cap S_i$, and let $N_{\geq i}(u)$ denote the set of edges $\cup_{i' \geq i} N_{i'}(u)$.

\begin{lem}\label{lem:degree concentration}
For all $u \in V$, $i \in [T]$, we have that
\[|N_i(u)| < (\epsilon + \epsilon^2) (1 - \epsilon)^{i-1} \Delta \]
with probability at least $1 - 1/n^{14}$.
\end{lem}

\begin{proof}
Let $u \in V$ and $i \in [T]$. Since the round of each edge is sampled from the capped geometric distribution, it follows that
\[\mathbb E[|N_i(u)| ] = \sum_{e \in N(u)} \Pr[e \in S_i] = \epsilon (1 - \epsilon)^{i-1} |N(u)| \leq \epsilon (1 - \epsilon)^{i-1} \Delta.\]
Since the rounds of edges are sampled independently, we can apply a Chernoff bound to get concentration. It follows that
\[ \Pr \left[ |N_i(u)| \geq (1 + \epsilon) \cdot \epsilon (1 - \epsilon)^{i-1} \Delta \right] \leq \exp({-\epsilon^2 \cdot \epsilon (1 - \epsilon)^{i-1} \Delta/3}). \]
Now note that
\[(1 - \epsilon)^{i-1} \geq e^{-\epsilon T} (1 - \epsilon^2 T) \geq \epsilon (1 - \epsilon \log(1 / \epsilon)) \geq \epsilon/2,  \]
and hence we have that
\[\exp({-\epsilon^2 \cdot \epsilon (1 - \epsilon)^{i-1} \Delta/3}) \leq \exp({-\epsilon^4 \Delta/6}) \leq \exp({-16\log n}) = 1/n^{16}.\]
The result follows by union bounding over all $u \in V$, $i \in [T]$, and noting that $T \leq 1/\epsilon^4 \leq n$.
\end{proof}

\noindent We now define an event $\mathcal Z$ which occurs if and only if $|N_i(u)| < (\epsilon + \epsilon^2) (1 - \epsilon)^{i-1} \Delta$ for all $u \in V$, $i \in [T]$. By Lemma \ref{lem:degree concentration}, this event occurs with probability at least $1 - 1/n^{14}$. For the rest of \Cref{app:concetration}, we fix all of the random bits used to determine the rounds of the potential edges and assume that event $\mathcal Z$ occurs. We implicitly condition all probabilities on event $\mathcal Z$ unless explicitly stated otherwise. Hence, when taking expectations and probabilities, we are doing so over the randomness of the color sequences assigned to edges.

\begin{lem}\label{lem: node palette concentration}
For all $u \in V$, $i \in [T]$, we have that
\[|P_i(u)| > (1 + \epsilon) (1 - \epsilon)^{i-1} \Delta. \]
\end{lem}

\begin{proof}
    Given any $u \in V$, $i \in [T]$, we have that
    \[|N_{<i}(u)| = \sum_{j=1}^{i-1} |N_j(u)| < (\epsilon + \epsilon^2)\Delta \sum_{j=1}^{i-1} (1 - \epsilon)^{j-1} = (1 + \epsilon)\Delta \cdot (1 - (1 - \epsilon)^{i-1}). \]
    It follows that $|P_i(u)| \geq (1 + \epsilon)\Delta - |N_{< i}(u)| \geq (1 + \epsilon) (1 - \epsilon)^{i-1} \Delta$.
\end{proof}

\begin{lem}\label{lem: edge palette concentration}
For all $e \in E$, $i \in [T]$ such that $e \in E_{i}$ and $e \in N^2(U)$, we have that
\[|P_i(e)| > (1 - \epsilon^2) (1 - \epsilon)^{2(i-1)} \Delta \]
with probability at least $1 - 1/n^{9}$.
\end{lem}

\begin{proof}
    Given any such edge $e = (u,v)$, we have that
    \[|P_i(e)| = \sum_{c \in [(1 + \epsilon) \Delta]} X^u_c \cdot X^v_c\]
    where $X^w_c$ is the indicator for the event that $c \in P_i(w)$ for a color $c$ and a node $w$. Since one of the endpoints of $e$ is distance at most $2$ from a good node, we get that $\mathcal N(e, T-1)$ is a tree.
    Since we only want to analyse the palette $P_i(e)$, we can assume that the input graph is $\mathcal N(e, T-1)$, which is a tree. Hence, the connected components containing $u$ and $v$ in the graph $(V, S_{<i})$ are not connected, so by Lemmas~\ref{lem:symmetry of nibble 4} and \ref{lem:symmetry of nibble 5} it follows that $\{X^u_c\}_c$ and $\{X^v_c\}_c$ are independent families of random variables.
    Note that, by Lemma~\ref{lem:symmetry of nibble}, given that $|P_i(u)| = \gamma$ for some $\gamma \in \mathbb N$, $\{X^u_c\}_c$ is a permutation distribution where exactly $\gamma$ of the $X^u_c$ are $1$ and the rest are $0$. Since we know that $|P_i(v)| > (1 + \epsilon)(1 - \epsilon)^{i-1}\Delta$, we can define a collection of random variables $\{Y^u_c\}_c$ by taking a uniform random subset $\Pi^u$ of size $(1 + \epsilon)(1 - \epsilon)^{i-1}\Delta$ of the set $\{c \in [(1 + \epsilon)\Delta] \, | \, X^u_c = 1\}$ and letting $Y^u_c$ indicate whether $c \in \Pi^u$. It follows that $\{Y^u_c\}_c$ is a permutation distribution where exactly $(1 + \epsilon)(1 - \epsilon)^{i-1}\Delta$ of the $Y^c_u$ are 1, and hence by Proposition~\ref{prop:perm dist NA} is an NA collection of indicator random variables such that $Y^u_c \leq X^u_c$ for each $c$. We define $\{Y^v_c\}_c$ in the exact same way. Since $\{X^u_c\}_c$ and $\{X^v_c\}_c$ are independent families, it follows that $\{Y^u_c\}_c$ and $\{Y^v_c\}_c$ are independent families, and we can apply closure under products (Proposition~\ref{prop:closure NA}) to get that $\{Y^u_c, Y_c^v\}_c$ is also a family of NA random variables. By then applying disjoint monotone aggregation (Proposition~\ref{prop:closure NA}), it follows that $\{Y_c^u \cdot Y_c^v\}_c$ are NA. Hence, we can apply a Chernoff bound. Letting $Y_c =  Y_c^u \cdot Y_c^v$, we first note that
    \[\mathbb E \left[ \sum_{c} Y_c \right] = \sum_{c} \mathbb E \left[Y_c^u\right] \cdot \mathbb E \left[ Y_c^v \right] = (1 + \epsilon)\Delta \cdot \frac{(1 + \epsilon)(1 - \epsilon)^{i-1}\Delta}{(1 + \epsilon)\Delta} \cdot \frac{(1 + \epsilon)(1 - \epsilon)^{i-1}\Delta}{(1 + \epsilon)\Delta}\]
    \[ = (1 + \epsilon)(1 - \epsilon)^{2(i-1)}\Delta. \]
    It follows that
    \[\Pr \left [ \sum_{c} Y_c \leq (1 - \epsilon) \cdot (1 + \epsilon)(1 - \epsilon)^{2(i-1)}\Delta  \right] \leq \exp(-\epsilon^2 \cdot (1 + \epsilon)(1 - \epsilon)^{2(i-1)}\Delta /2 ).\]
    As we saw in the proof of Lemma \ref{lem:degree concentration}, $(1 - \epsilon)^{i-1} \geq \epsilon/2$, so it follows that
    \[\exp(-\epsilon^2 \cdot (1 + \epsilon)(1 - \epsilon)^{2(i-1)}\Delta /2 ) \leq \exp(-\epsilon^4 \Delta /8 ) \leq \exp(-12\log n ) \leq 1/n^{12}. \]
    The result follows by union bounding over all $e \in E$, $i \in [T]$, and noting that $|P_i(e)| \geq \sum_c Y_c$.
\end{proof}

\subsubsection{Analyzing the Failed Edges}\label{sec:anal failed}

Given some good node $u \in U$, we now bound the number of edges incident on $u$ that fail to be colored, either because the algorithm failed to find a color in its palette or because of conflicts with neighboring edges. For $u \in V$, $i \in [T+1]$, we denote by $F_i(u)$ the set $N_i(u)\cap F_i$ of edges incident on $u$ that fail during iteration $i$. We begin by bounding the number of edges in $F_{T+1}(u)$, and then proceed to bound $F_i(u)$ for $i \in [T]$. The following lemma is \emph{not} conditioned on event $\mathcal Z$.

\begin{lem}\label{lem:F_{T+1} bound}
   Let $u \in V$. Then we have that $|F_{T+1}(u)| \leq \epsilon(1 +\epsilon)\Delta$ with probability at least $1 - 1/n^{33}$.
\end{lem}

\begin{proof}
    Let $u \in V$. Given some $e \in N(u)$, we can see that the probability that $e$ is never selected to be colored during a round is $(1 - \epsilon)^T$. Since all of these edges are sampled to be colored independently, it follows that 
    \[\mathbb E[|F_{T+1}(u)|] = (1 - \epsilon)^T |N(u)| \leq e^{-\epsilon T} \Delta = \epsilon \Delta.\]  By applying Chernoff bounds, it follows that
    \[\Pr \left[|F_{T+1}(u)| \geq (1 + \epsilon) \cdot \epsilon \Delta \right] \leq \exp(-\epsilon\Delta \cdot \epsilon^2/3) \leq \exp(-33 \log n) = 1/n^{33}.\]
\end{proof}

\noindent
For the rest of \Cref{sec:anal failed}, we again fix all of the random bits used to determine the rounds of the potential edges and assume that event $\mathcal Z$ occurs, implicitly conditioning all probabilities on $\mathcal Z$.
 
For some $i \in [T]$, we now categorize the failed edges in $F_i(u)$ into 3 types and bound each of these individually. We say that edges $e$ and $f$ \textit{conflict} if they share an endpoint and are assigned the same tentative color. Note that conflicting edges must receive their tentative colors on the same round. Given some edge $e \in F_i(u)$, we place $e$ into $F'_i(u)$ iff there exists some edge $f \in F_i(u)$ such that $e$ and $f$ conflict and we place $e$ into $F''_i(u)$ iff there exists some edge $f \in F_i \setminus F_i(u)$ such that $e$ and $f$ conflict. $F_i'(u) \cup F_i''(u)$ capture all of the edges that fail due to conflicts, i.e. $F_i'(u) \cup F_i''(u) = \{e \in N_i(u) \, | \, \exists f \in N_i(e) \textnormal{ such that } \tilde \chi(e) = \tilde \chi(f)\}$. Finally, we let $F_i'''(u)$ be the edges $e \in F_i(u)$ that fail because the algorithm does not find a color in its palette, i.e. $c_e \cap P_{i}(e) = \varnothing$. Clearly $F_i(u) = F'_i(u) \cup F''_i(u) \cup F'''_i(u)$. Note that these sets are not necessarily disjoint. We now proceed to bound $|F_i(u)|$ by individually bounding $|F'_i(u)|$, $|F''_i(u)|$ and $|F'''_i(u)|$. We split up the bound in this way because the techniques required to establish concentration on the sizes of these sets are slightly different.

\begin{lem}\label{lem: F' concentration}
    Let $u \in U$, $i \in [T]$. Then we have that $\mathbb |F'_i(u)| \leq 2(\epsilon^2 + \epsilon^3)(1 - \epsilon)^{i-1}\Delta$ with probability at least $1 - 1/n^{25}$.
\end{lem}

\begin{proof}
    Let $u \in U$, $i \in [T]$, and arbitrarily fix all of the random bits used in the first $i-1$ rounds that determine the palette $P_i(u)$. Let $u_1,\dots,u_\ell$ be the nodes that are connected to $u$ by an edge in $N_i(u)$. Recall that we can assume the input graph is $\mathcal N(u,T+1)$, which is a tree. Since the graph is a tree, the nodes $u, u_1,\dots,u_\ell$ are all disconnected from each other in $(V, S_{<i})$, and hence
    $\{X^c_{u_1}\}_c,\dots,\{X^c_{u_\ell}\}_c$ are mutually independent families of random variables by Lemmas \ref{lem:symmetry of nibble 4} and \ref{lem:symmetry of nibble 5}.
    
    Let $e \in N_i(u)$. Given some $f \in N_i(u) \setminus \{e\}$, the probability that $e$ and $f$ conflict is the probability that $e$ and $f$ are assigned the same color.
    Letting $c = \tilde \chi(e)$, we have 
    by Lemma \ref{lem:symmetry of nibble 3} that $\Pr[\tilde \chi(f) = c ] \leq 1/|P_i(u)|$.
    Let $f,f' \in N_i(u) \setminus \{e\}$ be distinct edges such that $f=(u,v)$ and $f=(u,v')$. Then the event $\tilde \chi(f) = c$ depends on the random bits that determine $P_i(v)$ and the random bits used to sample $\tilde \chi(f)$ (recall that the bits that determine $P_i(u)$ are fixed). Hence, by Lemma \ref{lem:symmetry of nibble 4}, the events $\tilde \chi(f) = c$ and $\tilde \chi(f') = c$ are independent since they depend on distinct random bits. It follows that 
    \[\Pr[e \notin F'_i(u)] =  \prod_{f \in N_i(u) \setminus \{e\}} \Pr \left[ \tilde \chi(f) \neq c \right] \geq \prod_{f \in N_i(u) \setminus \{e\}} \left(1 - \frac{1}{|P_i(u)|} \right) \]
    \[\geq \left(1 - \frac{1}{(1 + \epsilon)(1 - \epsilon)^{i-1}\Delta} \right)^{\epsilon(1 + \epsilon)(1 - \epsilon)^{i-1}\Delta} \geq 1 - \frac{\epsilon(1 + \epsilon)(1 - \epsilon)^{i-1}\Delta}{(1 + \epsilon)(1 - \epsilon)^{i-1}\Delta} = 1 - \epsilon,\]
    where the last inequality follows from Bernoulli's inequality. By linearity of expectation, we get that
    \[\mathbb E[|F'_i(u)|] = \sum_{e \in N_i(u)} \Pr[e \in F_i'(u)] \leq \epsilon |N_i(u)| \leq (\epsilon^2 + \epsilon^3)(1 - \epsilon)^{i-1}\Delta. \]
    We now establish concentration. Let $N_i(u) = \{e_1,...,e_\ell\}$ and $\phi'(\tilde \chi(e_1),...,\tilde \chi(e_\ell))$ be the function that counts the number of edges in $F'_i(u)$, i.e. $\phi' = |F'_i(u)|$. Since the function $\phi'$ is Lipschitz with all constants $2$ (Definition~\ref{def:lippy}) and $\tilde \chi(e_1),...,\tilde \chi(e_\ell) \in P_i(u)$ are mutually independent as they depend on distinct random bits, we can apply the method of bounded differences (Proposition~\ref{lem:bounded diff}) to get that
    \[\Pr \left[\phi' \geq \mathbb E [\phi'] + t \right] \leq \exp{\left(-\frac{t^2}{2|N_i(u)|} \right)} \]
    for all $t > 0$. By setting $t = \epsilon^2(1 + \epsilon)(1 - \epsilon)^{i-1}\Delta$ we get that $|F'_i(u)| \geq 2(\epsilon^2 + \epsilon^3)(1 - \epsilon)^{i-1}\Delta$ with probability at most
    \[\exp{\left(-\frac{\epsilon^4(1 + \epsilon)^2(1 - \epsilon)^{2(i-1)}\Delta^2}{2(\epsilon + \epsilon^2)(1 - \epsilon)^{i-1}\Delta} \right)} = \exp{\left(-\frac{\epsilon^3(1 + \epsilon)(1 - \epsilon)^{i-1}\Delta}{2} \right)} \leq \exp{\left(-\frac{1}{4} \epsilon^4\Delta \right)} \leq \frac{1}{n^{25}}.\]
\end{proof}

\begin{lem}\label{lem: F'' concentration}
Let $u \in U$, $i \in [T]$. Then we have that $|F_i''(u)| \leq 2(\epsilon^2 + \epsilon^3)(1 - \epsilon)^{i-1}\Delta$ with probability at least $1 - 1/n^{33}$.
\end{lem}

\begin{proof}
    Let $u \in U$, $i \in [T]$, and $u_1,\dots,u_\ell$ be the nodes that are connected to $u$ by an edge in $N_i(u)$.
    Arbitrarily fix all of the random bits used in the first $i-1$ rounds that determine the palettes $P_i(u),P_i(u_1),\dots,P_i(u_\ell)$. 
    Let $v \in \{u_1,\dots,u_\ell\}$, $e = (u,v)$, and $v_1,\dots,v_{\ell'}$ be the nodes that are connected to $v$ by an edge in $N_i(u) \setminus \{e\}$. Recall that we can assume the input graph is $\mathcal N(u,T+1)$, which is a tree. Since the graph is a tree, the nodes $v, v_1,\dots,v_\ell$ are all disconnected from each other in $(V, S_{<i})$, and hence
    $\{X^c_{v_1}\}_c,\dots,\{X^c_{v_\ell}\}_c$ are mutually independent families of random variables by Lemmas \ref{lem:symmetry of nibble 4} and \ref{lem:symmetry of nibble 5}.
    
    Given some $f \in N_i(v) \setminus \{e\}$, the probability that $e$ and $f$ conflict is the probability that $e$ and $f$ are assigned the same tentative color.
    Letting $c = \tilde \chi(e)$, we have 
    by Lemma \ref{lem:symmetry of nibble 3} that $\Pr[\tilde \chi(f) = c] \leq 1/|P_i(v)|$.
    Let $f,f' \in N_i(v) \setminus \{e\}$ be distinct edges such that $f=(v,w)$ and $f=(v,w')$. Then the event $\tilde \chi(f) = c$ depends on the random bits that determine $P_i(w)$ and the random bits used to sample $\tilde \chi(f)$ (recall that the bits that determine $P_i(v)$ are fixed). Hence, by Lemma \ref{lem:symmetry of nibble 4}, the events $\tilde \chi(f) = c$ and $\tilde \chi(f') = c$ are independent since they depend on distinct random bits. It follows that 
    \[\Pr[e \notin F''_i(u)] =  \prod_{f \in N_i(v) \setminus \{e\}} \Pr \left[ \tilde \chi(f) \neq c \right] \geq \prod_{f \in N_i(v) \setminus \{e\}} \left(1 - \frac{1}{|P_i(v)|} \right) \]
    \[\geq \left(1 - \frac{1}{(1 + \epsilon)(1 - \epsilon)^{i-1}\Delta} \right)^{\epsilon(1 + \epsilon)(1 - \epsilon)^{i-1}\Delta} \geq 1 - \frac{\epsilon(1 + \epsilon)(1 - \epsilon)^{i-1}\Delta}{(1 + \epsilon)(1 - \epsilon)^{i-1}\Delta} = 1 - \epsilon,\]
    where the last inequality follows from Bernoulli's inequality. By linearity of expectation, we get that
    \[\mathbb E[|F''_i(u)|] = \sum_{e \in N_i(u)} \Pr[e \in F_i''(u)] \leq \epsilon |N_i(u)| \leq (\epsilon^2 + \epsilon^3)(1 - \epsilon)^{i-1}\Delta. \]
    We now establish concentration. Given $e = (u,v),e' = (u,v') \in N_i(u)$, the fact that the events $e  \in F''_i(u)$ and $e' \in F''_i(u)$ depend on disjoint collections of random bits follow the details above and the fact that the graph is a tree. Hence, these events are independent. It follows that we can apply a Chernoff bound.
    \[\Pr[ |F''_i(u)| \geq 2(\epsilon^2 + \epsilon^3)(1 - \epsilon)^{i-1}\Delta ] \leq \exp \left( -\frac{1}{3} (\epsilon^2 + \epsilon^3)(1 - \epsilon)^{i-1}\Delta  \right) \leq \exp \left( -\frac{\epsilon^3}{6} (1 + \epsilon)\Delta  \right) \leq 1/n^{33}.\]
\end{proof}

\begin{lem}\label{lem: F''' concentration}

Let $u \in U$, $i \in [T]$. Then we have that $|F_i'''(u)| \leq 2(\epsilon^2 + \epsilon^3)(1 - \epsilon)^{i-1}\Delta$ with probability at least $1 - 2/n^{9}$.

\end{lem}

\begin{proof}
    We begin by fixing all of the random bits used in the first $i-1$ rounds such that for all $e \in N^2(U)$ we have $|P_i(e)| > (1 - \epsilon^2) (1 - \epsilon)^{2(i-1)} \Delta$. Note that by Lemma \ref{lem: edge palette concentration} this event occurs with probability at least $1 - 1/n^9$.
    Given some edge $e \in N_i(u)$, we now have that
    \[\Pr[ e \in F'''(u) ] = \left(1 - \frac{|P_i(e)|}{(1 + \epsilon)\Delta} \right)^{8\log(1/\epsilon)/\epsilon^2}\leq \left(1 - (1 - \epsilon)(1 - \epsilon)^{2(i-1)} \right)^{8\log(1/\epsilon)/\epsilon^2}  \]
    \[\leq \left(1 - \epsilon^2/8 \right)^{8\log(1/\epsilon)/\epsilon^2} \leq \epsilon, \]
    where the first equality follows from the fact that we are drawing colors from $[(1 + \epsilon)\Delta]$ independently and u.a.r while avoiding $P_i(e)$ and the second inequality follows from the fact that $(1 - \epsilon)^{i-1} \geq \epsilon/2$. By linearity of expectation, it follows that
    \[\mathbb E[|F'''_i(u)|] = \sum_{e \in N_i(u)} \Pr[e \in F'''_i(u)] \leq \epsilon |N_i(u)| < (\epsilon^2 + \epsilon^3)(1 - \epsilon)^{i-1}\Delta.\]
    Given two distinct edges $e = (u,v),f = (u,w) \in N_i(u)$,
    the random bits used by the algorithm to sample the colors $\tilde \chi(e)$ and $\tilde \chi(f)$ are distinct, so we can apply Lemma \ref{lem:symmetry of nibble 4} to get that
    the events $e \in F'''_i(u)$ and $f \in F'''_i(u)$ are independent.
    Hence, we can apply Chernoff bounds to get that
    \[\Pr[ |F'''_i(u)| \geq 2(\epsilon^2 + \epsilon^3)(1 - \epsilon)^{i-1}\Delta ] \leq \exp \left( -\frac{1}{3} (\epsilon^2 + \epsilon^3)(1 - \epsilon)^{i-1}\Delta  \right) \leq \exp \left( -\frac{\epsilon^3}{6} (1 + \epsilon)\Delta  \right) \leq 1/n^{33}.\]
\end{proof}

\noindent Let $G_F$ denote the subgraph of $G$ consisting of the edges contained in $F$. We now remove the conditioning on event $\mathcal Z$.

\begin{lem}\label{lem:max deg G_F}
We have that $\deg_{G_F}(u) \leq 7 \epsilon (1 + \epsilon) \Delta$ for all $u \in U$ with probability at least $1 - 6/n^{7}$.
\end{lem}

\begin{proof}
Let $u \in U$. Then, by Lemmas \ref{lem:F_{T+1} bound}, \ref{lem: F' concentration}, \ref{lem: F'' concentration} and \ref{lem: F''' concentration} we get that
 \[\deg_{G_F}(u) = |F_{T+1}(u)| + \sum_{i=1}^T |F_i(u)| \leq |F_{T+1}(u)| +  \sum_{i=1}^T \left( |F'_i(u)| + |F''_i(u)| + |F'''_i(u)| \right)\]
 \[ \leq \epsilon(1 +\epsilon)\Delta +  6(\epsilon^2 + \epsilon^3)\Delta \sum_{i=1}^T(1 - \epsilon)^{i-1} = \epsilon(1 +\epsilon)\Delta + 6(\epsilon^2 + \epsilon^3)\Delta \cdot \frac{1 - (1 - \epsilon)^{T}}{\epsilon}  \leq 7 \epsilon (1 + \epsilon) \Delta. \]
 with probability at least $1 - 1/n^{33} - T/n^{25} - T/n^{33} - 2T/n^{9} \geq 1 - 5/n^{8}$. The lemma follows by union bounding over all nodes in $U$ and removing the conditioning on event $\mathcal Z$. 
\end{proof}

\subsection{Properties of Algorithm \ref{app:alg:split}}

Let $G = (V,E)$ be a graph with maximum degree at most $\Delta$ and $\mathcal G_1,\dots, \mathcal G_\eta$ be the subgraphs obtained from calling Algorithm \ref{app:alg:split} on $G$ with some $2 \leq \Delta' \leq \Delta$.

\begin{lem}\label{lem:G_i degree bound}
    $\Delta(\mathcal G_j) \leq \Delta' + 10\sqrt{\Delta' \log n}$ for all $j \in [\eta]$ with probability at least $1 - 1/n^{31}$.
\end{lem}

\begin{proof}
    Let $j \in [\eta]$ and $u \in V$. Let $X^e_j$ be the indicator for the event that some edge $e \in E$ is contained in the graph $\mathcal G_j$. Then clearly $\deg_{\mathcal G_j}(u) = \sum_{e \in N(u)} X^e_j$, and hence
    \[\mathbb E[\deg_{\mathcal G_j}(u)] = \sum_{e \in N_G(u)} \Pr[e \in \mathcal E_j] \leq \frac{\Delta}{\eta} \leq \Delta'.\]
    By applying Chernoff bounds, we get that
    \[ \Pr \left[ \deg_{\mathcal G_j}(u) > \Delta' + 10\sqrt{\Delta' \log n} \right] \leq \exp \left( -\frac{1}{3} \cdot \frac{100 \log n}{\Delta'} \cdot \Delta' \right) \leq \frac{1}{n^{33}}, \]
    which implies that $\deg_{\mathcal G_j}(u) \leq \Delta' + 10\sqrt{\Delta' \log n}$ with probability at least $1 - 1/n^{33}$. We can union bound over all $u \in V$ and $j \in [\eta]$ to get that $\Delta(\mathcal G_j) \leq \Delta' + 10\sqrt{\Delta' \log n}$ for all $j \in [\eta]$ with probability at least $1 - n \eta/n^{33} \geq 1 - 1/n^{31}$.
\end{proof}

\noindent
The following lemma is proven by \cite{KulkarniLSST22}.

\begin{lem}[Lemma 4.2, \cite{KulkarniLSST22}]\label{lem:kill small cycles original}
     Let $G'$ be a subgraph of $G$ obtained by sampling each edge in $G$ independently with probability $D/\Delta$, where $D \geq 2$. Then the probability that the $g$-neighborhood of an edge $e$ in $G'$ contains a cycle is at most $3D^{5g}/\Delta$.
\end{lem}

\noindent We will use the following lemma, which follows immediately from Lemma \ref{lem:kill small cycles original}.

\begin{lem}\label{lem:kill small cycles}
     Let $G'$ be a subgraph of $G$ obtained by sampling each edge in $G$ independently with probability $D/\Delta$, where $D \geq 2$. Then the probability that the $g$-neighborhood of a node $u$ contains a cycle in $G'$ is at most $3D^{5g}/\Delta$.
\end{lem}

\begin{lem}\label{lem:deg G star}
    Let $G^\star$ be the subgraph of $G$ that contains an edge $e$ iff $e \in \mathcal E_j$ is incident on a node $u$ such that the $g$-neighborhood of $e$ in $\mathcal G_j$ is not a tree. Then $\Delta(G^\star) \leq (\Delta' + 10\sqrt{\Delta' \log n}) \cdot (6(\Delta')^{5(g+1)} + 10\sqrt{(\Delta/\Delta') \log n})$ with probability at least $1 - 1/n^{30}$.
\end{lem}

\begin{proof}
    Given some $u \in V$, $j \in [\eta]$, define $X^u_j$ to be the indicator for the event that the $(g+1)$-neighborhood of $u$ in the graph $\mathcal G_j$ is not a tree.
    It immediately follows that
    \[\deg_{G^\star}(u) \leq \sum_{j \in [\eta]} X^u_j \cdot |N(u) \cap \mathcal E_j| \leq \sum_{j \in [\eta]} X^u_j \cdot \Delta(\mathcal G_j) \leq \max_{j \in [\eta]}\Delta(\mathcal G_j) \cdot \sum_{j \in [\eta]} X^u_j.\]
    By Lemma \ref{lem:G_i degree bound}, we have that $\max_j \Delta(\mathcal G_j) \leq \Delta' + 10\sqrt{\Delta' \log n}$ with probability at least $1 - 1/n^{31}$. 
    It follows from Lemma \ref{lem:kill small cycles} that the $g+1$-neighborhood of $u$ in the graph $\mathcal G_j$ is not a tree with probability at most $ 3(\Delta / \eta)^{5(g+1)}/\Delta \leq 3(\Delta')^{5(g+1)}/\Delta$. Hence, letting $X^u$ denote $\sum_{j \in [\eta]} X^u_j$,
    \[\mathbb E [ X^u ] \leq \sum_{j \in [\eta]} \mathbb E[X^u_j] \leq \sum_{j \in [\eta]} 3(\Delta')^{5(g+1)}/\Delta \leq 6(\Delta')^{5(g+1)} \]
    holds for all $u \in V$.
    In order to establish concentration, we first establish that, for any fixed $u$, the random variables $\{X^u_j\}_j$ are NA (see Definition~\ref{def:NA}). Given some $e \in E$, let $X^e_j$ indicate the event that $e \in \mathcal E_j$. For any fixed $j$, the random variables $\{X^e_j\}_j$ are NA by Proposition~\ref{prop:0-1 NA}. Since the families of random variables $\{X^{e_1}_j\}_j,\dots,\{X^{e_m}_j\}_j$ are mutually independent, it follows by closure under products (Proposition~\ref{prop:closure NA}) that $\{X^e_j\}_{j,e}$ are NA. Finally, for any $u\in V$, since $X^u_j$ is a monotonically increasing function of $X^{e_1}_j,\dots,X^{e_m}_j$, it follows by disjoint monotone aggregation (Proposition~\ref{prop:closure NA}) that $\{X^u_j\}_j$ are NA. Hence, we can apply Hoeffding bounds for NA random variables to get that
    \[\Pr \left[X^u > 6(\Delta')^{5(g+1)} + 10\sqrt{(\Delta/\Delta') \log n} \right] \leq \exp \left(-2 \cdot \frac{100 (\Delta / \Delta') \log n}{\eta} \right) \leq \frac{1}{n^{100}}.\]
    By union bounding over all $u \in V$, it follows that, with probability at least $1 - 1/n^{99}$,
    \[X^u \leq 6(\Delta')^{5(g+1)} + 10\sqrt{(\Delta/\Delta') \log n}\]
    for all $u \in V$. Putting everything together and applying a union bound we get that with probability at least $1 - 1/n^{30}$
    \[\Delta(G^\star) \leq \left(\Delta' + 10\sqrt{\Delta' \log n} \right) \cdot \left(6(\Delta')^{5(g+1)} + 10\sqrt{(\Delta/\Delta') \log n} \right).\]
\end{proof}

\subsection{Edge Coloring the Subsampled Graphs}

Let $G = (V,E)$ be a graph with maximum degree at most $\Delta$ such that $\Delta \geq (100 \log n / \epsilon^4)^{30 T}$. Let $\gamma = 1/(30T)$ and $\Delta' = \Delta^\gamma$. Now suppose we run Algorithm~\ref{app:alg:static} on input $(G, \Delta, \epsilon)$. Let $\mathcal G_1,\dots , \mathcal G_\eta$ denote the partition of $G$ produced by the algorithm. Recall that, for $j \in [\eta]$, $\mathcal F_j$ denotes the set of failed edges in $\mathcal G_j$. Let $\mathcal F = \bigcup_{j \in [\eta]} \mathcal F_j$. For notational convenience, we identify $\mathcal F$ with the graph $(V, \mathcal F)$ for the rest of this section. Finally, let $G^\star$ denote the graph that contains an edge $e \in \mathcal E_j$ if and only if one of the endpoints of $e$ is bad with respect to $\mathcal G_j$. Then we have the total number of colors used by our algorithm is
\[ \sum_{j \in [\eta]} (1 + \epsilon)^2 \Delta' +  3 \Delta ( \mathcal F ), \]
where the factor of $3$ in the second term comes from the greedy algorithm.
We now proceed to show that, with high probability, this expression is upper bounded by $(1 + O(\epsilon))\Delta$. We begin with the following simple observation.

\begin{obs}\label{lem:Delta(G_i)}
    For all $j \in [\eta]$, $\Delta(\mathcal G_j) \leq (1 + \epsilon)\Delta'$
    with probability at least $1 - 1/n^{31}$.
\end{obs}

\begin{proof}
By Lemma \ref{lem:G_i degree bound}, we have that for all $j \in [\eta]$
\[ \Delta(\mathcal G_j) - \Delta' \leq 10\sqrt{\Delta' \log n} \leq 10 \epsilon^2 \Delta' / (10 \sqrt{2} ) \leq \epsilon\Delta'\]
with probability at least $1 - 1/n^{31}$.
\end{proof}

\noindent
For the rest of this section, we assume that the event in the statement of Observation \ref{lem:Delta(G_i)} occurs and implicitly condition all probabilities on this event unless explicitly stated otherwise. Note that we also have $(1 + \epsilon)\Delta' \geq (100 \log n)/\epsilon^2$.

\begin{lem}\label{lem: phase 1 bound}
    We have that $\sum_{j \in [\eta]} (1 + \epsilon)^2 \Delta' \leq (1 + 4 \epsilon) \Delta$.
\end{lem}

\begin{proof}
\[ \sum_{j \in [\eta]} (1 + \epsilon)^2 \Delta' \leq (1 + \epsilon)^2 \Delta' \eta \leq (1 + \epsilon)^2 \Delta + (1 + \epsilon)^2 \leq (1 + 4 \epsilon) \Delta.\]
\end{proof}



\begin{lem}
    Let $u \in V$ and let $J^\star_u = \{ j \in [\eta] \, | \, u \textnormal{ is good with respect to } \mathcal G_j \}$. Then we have that
    \[\deg_{\mathcal F}(u) \leq \Delta(G^\star) + \sum_{j \in J^\star_u} \deg_{\mathcal F_j}(u)\]
\end{lem}

\begin{proof}
    We have that
    \[\deg_{\mathcal F}(u) = \sum_{j \in [\eta]} \deg_{\mathcal F_j}(u) = \sum_{j \in J^\star_u} \deg_{\mathcal F_j}(u) + \sum_{j \in [\eta] \setminus J^\star_u} \deg_{\mathcal F_j}(u).\]
    By the definition of the graph $G^\star$, for any $j \in [\eta] \setminus J^\star_u$, all of the edges incident on $u$ in $\mathcal G_j$ are contained in $G^\star$. Since the $\mathcal G_j$ are edge-disjoint, it follows that
    \[\sum_{j \in [\eta] \setminus J^\star_u} \deg_{\mathcal F_j}(u) \leq \sum_{j \in [\eta] \setminus J^\star_u} \deg_{\mathcal G_j}(u) \leq \deg_{G^\star}(u) \leq \Delta(G^\star).\]
\end{proof}

\begin{lem}
    For all $j \in [\eta]$, $\deg_{\mathcal F_j}(u) \leq 9 \epsilon \Delta'$ for all $u \in V$ such that $u$ is good with respect to $\mathcal G_j$ with probability at least $1 - 6/n^{6}$.
\end{lem}

\begin{proof}
    Let $j \in [\eta]$. It follows from Lemma \ref{lem:max deg G_F} that, for all $u \in V$ that are good with respect to $\mathcal G_j$, $\deg_{\mathcal F_j}(u) \leq 7\epsilon(1 + \epsilon)^2 \Delta' \leq 9 \epsilon \Delta'$ with probability at least $1 - 6/n^{7}$. The result follows by union bounding over all $j \in [\eta]$.
\end{proof}

\begin{lem}
We have that $\Delta(G^\star) \leq \epsilon \Delta$ with probability at least $1 - 1/n^{30}$.
\end{lem}

\begin{proof}
By Lemma \ref{lem:deg G star}, with probability at least $1 - 1/n^{30}$ we have that
\begin{align*}
\Delta(G^\star) &\leq \left(\Delta^\gamma + \Delta^{\gamma/2} \cdot 10\sqrt{\log n} \right) \cdot \left(\Delta^{5\gamma (T+2)} \cdot 6 + \Delta^{(1 - \gamma)/2} \cdot 10\sqrt{\log n} \right)\\
& \leq \Delta^{5 \gamma T + 11\gamma} \cdot 6 +  \Delta^{(1 + \gamma)/2} \cdot 10\sqrt{ \log n} + \Delta^{5 \gamma T + 11\gamma} \cdot 60\sqrt{\log n} + \Delta^{1/2} \cdot 100 \log n\\
&\leq 100 \log n \cdot \left( \Delta^{5\gamma T + 11\gamma} + \Delta^{(1 + \gamma)/2} + \Delta^{1/2} \right)\\
&\leq \left(\Delta^{5/30 + 11/(30T)} + \Delta^{1/2 + 1/(60T)}\right) \cdot 200 \log n\\
&\leq \left(\Delta^{5/30 + 11/(30T)} + \Delta^{1/2 + 1/(60T)}\right) \cdot 300 \epsilon^4 \Delta^{1/(30T)} / 200\\
&\leq 4 \epsilon^4 \Delta\\
&\leq \epsilon \Delta.\\
\end{align*}
\end{proof}

\begin{lem}\label{lem: delta F bound}
    We have that $\Delta(\mathcal F) \leq 19\epsilon\Delta$ with probability at least $1 - 7/n^{6}$.
\end{lem}

\begin{proof}
    It follows from the preceding lemmas that for all $u \in V$ we have that
    \[\deg_{\mathcal F}(u) \leq \Delta(G^\star) + \sum_{j \in J^\star_u} \deg_{\mathcal F_j}(u) \leq \epsilon\Delta + |J^\star_u| \cdot 9 \epsilon \Delta' \leq 19\epsilon\Delta \]
    with probability at least $1 - 1/n^{30} - 6/n^6 \geq 1 - 7/n^6$. It immediately follows that $\Delta(\mathcal F) \leq 19\epsilon\Delta$ with the same probability.
\end{proof}


\noindent
We are not ready to complete the proof of Theorem~\ref{app:thm:static}.

\begin{proof}[Proof of Theorem~\ref{app:thm:static}]
    First, observe that our algorithm uses at most
    \[ \sum_{i \in [\eta]} (1 + \epsilon)^2 \Delta' +  3 \Delta ( \mathcal F )\]
    colors. It now follows from Lemmas \ref{lem: phase 1 bound} and \ref{lem: delta F bound} that
    \[\sum_{i \in [\eta]} (1 + \epsilon)^2 \Delta' + 3\Delta ( \mathcal F ) \leq (1 + 4\epsilon)\Delta + 3 \cdot 19 \epsilon \Delta = (1 + 61\epsilon)\Delta\]
    with probability at least $1 - 1/n^{31} - 7/n^{6} \geq 1 - 8/n^{6}$, where the terms in the probability come from removing the conditioning on the event from Observation~\ref{lem:Delta(G_i)} and the probability in Lemma~\ref{lem: delta F bound}.
\end{proof}

\section{Our Dynamic Algorithm with Constant Recourse (Full Version)}
\label{sec:app:recourse}

In this appendix, we describe our dynamic algorithm. For now, we omit all details of implementation and deal only with bounding the \emph{recourse} of our dynamic algorithm. We design the data structures that allow us to implement this procedure efficiently in the proceeding sections. The main result in this appendix is the following theorem.

\begin{thm}
    There exists a dynamic algorithm that, given a dynamic graph $G$ that is initially empty and evolves by a sequence of $\kappa$ edge insertions and deletions by means of an oblivious adversary, and a parameter $\Delta \geq (100 \log n/\epsilon^4)^{(30/\epsilon) \log (1/\epsilon)}$ such that the maximum degree of $G$ is at most $\Delta$ at all times,
    maintains a $(1 + 61\epsilon)\Delta$-edge coloring of $G$ and
    has an expected worst-case recourse of $O(1/\epsilon^4)$.
\end{thm}

\medskip
\noindent \textbf{The dynamic setting.}
In the dynamic setting, we have a graph $G = (V,E)$ that undergoes updates via a sequence of edge insertions and deletions, while the set of nodes remains fixed. Our task is to explicitly maintain an edge coloring $\chi$ of $G$ as it is updated, where $\Delta$ is a fixed upper bound on the maximum degree of the graph of $G$ at any point. Let $\sigma_1,...,\sigma_\kappa$ denote the sequence of updates, and $G^{(t)} = (V, E^{(t)})$ denote the state of the graph $G$ after the first $t$ updates. We assume that the graph $G$ is initially empty, i.e. $G^{(0)} = (V, \varnothing)$. Given some dynamic edge coloring algorithm, its \emph{update time} is the time it takes to handle an update, and its \emph{recourse} is the number of edges that change color during an update. In this paper, we assume that all adversaries are \emph{oblivious}. In other words, the update $\sigma_t$ does not depend on the random bits used by our algorithm while handling the updates $\sigma_1,\dots,\sigma_{t-1}$.

\medskip
\noindent \textbf{Our algorithm.}
Informally, our dynamic algorithm works by maintaining the output of \textsc{StaticColor} on the dynamic graph $G$ as it undergoes edge insertions and deletions. In other words, we maintain the invariant that at any point in time the coloring generated by our dynamic algorithm on the current input graph $G$ is the same as the coloring obtained by running \textsc{StaticColor} on $G$---regardless of what updates have occurred previously. Since we fix the randomness for every potential edge in advance, the tentative colors assigned by \textsc{StaticColor} to the edges in phase 2 (while running Algorithm \ref{app:alg:nibble} on the subgraphs $\mathcal G_j$) are completely determined by which edges are present in $G$. Hence, the tentative color $\tilde \chi^{(t)}(e)$ of each edge $e \in E^{(t)}$ is well defined. It follows from Theorem~\ref{app:thm:static} that as long as $\Delta \geq (100 \log n/\epsilon^4)^{(30/\epsilon) \log (1/\epsilon)}$, given some $t \in [\kappa]$, the edge coloring $\chi^{(t)}$ (defined by combining the tentative colorings $\tilde \chi^{(t)}$ for the edges in each $\mathcal G_j$ that don't fail and the greedy coloring of $H$) uses at most $(1 + 61\epsilon)\Delta$ colors with probability at least $1 - O(1/n^6)$.

We design data structures that allow us to maintain these tentative colorings explicitly, along with the sets of edges that fail, for each graph $\mathcal G_j$ as edges are inserted and deleted from $G$. We then dynamically maintain a greedy edge coloring of the graph $H$ that uses $3\Delta(H)$ many colors and only changes the colors of $O(1)$ many edges every time an edge is inserted or deleted from $H$. This level of detail is sufficient to upper bound the recourse of our dynamic algorithm. We begin by bounding the recourse of our algorithm and defer the details of how to implement this efficiently to the proceeding sections.

\subsection{Recourse analysis}

We now proceed to upper bound the expected recourse of our algorithm. We achieve this by arguing that in expectation the tentative colorings produced by $\textsc{StaticColor}$ on two graphs that differ only by one edge assign different colors to at most $O(1/\epsilon^4)$ many edges. By showing that the number of edges that are added and removed from $H$ is at most a constant factor larger than the number of edges that change their tentative colors, it then follows that the expected recourse of our algorithm is $O(1/\epsilon^4)$.

Let $t \in [\kappa]$. Suppose that the $t^{th}$ update corresponds to the insertion or deletion of an edge $e$ contained in $\mathcal G_j$. Then clearly only edges contained in $\mathcal G_j$ can change their tentative colors during this update. Since the colors used by the tentative colorings in each of $\mathcal G_1,\dots,\mathcal G_\eta$ are distinct, it also follows that only edges in $\mathcal G_j$ can be added or removed from $H$ during the update. It follows that while bounding the recourse of the $t^{th}$ update it is sufficient to only consider the edges contained in the graph $\mathcal G_j$. It is not too difficult to see that the recourse of the $t^{th}$ update is upper bounded by
\[O \left( \middle| \mathcal F^{(t-1)}_j \oplus \mathcal F^{(t)}_j \middle| \right) + \left| \left\{e \in \mathcal E^{(t)}_j \, \middle| \, \tilde \chi^{(t-1)}(e) \neq \tilde \chi^{(t)}(e) \right\} \right|,\]
where the terms correspond to the changes in the greedy coloring caused by edge insertions and deletions from $H$ due to edges that failed at time $t-1$ but not at time $t$ and vice versa, and the edges that change their tentative colors, respectively.

\medskip
\noindent \textbf{Notation.}
In order to emphasize that we are looking at the state of an object $X$ \textit{directly after the $t^{th}$ update} we add the superscript $X^{(t)}$. For example, $P_i^{(t)}(e)$ is the set $P_i(e)$ after the $t^{th}$ update, i.e. when the input graph to our algorithm is $G^{(t)}$, where $P_i(e)$ is the palette of the edge $e$ during round $i$ as defined in Appendix~\ref{app:static}. Since while bounding the recourse of the $t^{th}$ update it is sufficient to only consider the edges contained in the graph $\mathcal G_j$, when we use notation from Algorithm~\ref{app:alg:nibble}, it is implicit that the notation is referring to objects from the call to this algorithm on the graph $\mathcal G_j$. We will now introduce some definitions that will allow us to analyze the way that changes in the tentative coloring propagate through the rounds of Algorithm~\ref{app:alg:nibble} after an update.

\medskip
\noindent \textbf{The sets $\Gamma(e)$ and $\Lambda(e)$.}
The following definitions capture the notion of an edge $e$ having an effect on the tentative color assigned to some edge $f$ in a subsequent round and are crucial for efficiently identifying the edges whose tentative colors change after an update.

\begin{definition}
    Given some $t \in [\kappa]$, $i \in [T]$, and an edge $e \in S_i^{(t-1)} \cup S_i^{(t)}$, we define the sets of edges $\Gamma^{(t)}(e)$ and $\Lambda^{(t)}(e)$ by
    \[\Gamma^{(t)}(e) :=  \left\{ f \in N^{(t)}_{>i}(e) \, \middle| \, \tilde \chi^{(t-1)}(f) = \tilde \chi^{(t)}(e) \right\},\]
    \[\Lambda^{(t)}(e) :=  \left\{ f \in N^{(t)}_{>i}(e) \, \middle| \, \tilde \chi^{(t)}(f) = \tilde \chi^{(t-1)}(e) \right\}.\]
\end{definition}

\noindent
We denote the sets $S^{(t)}_i \cap \Gamma^{(t)}(e)$ and $S^{(t)}_i \cap \Lambda^{(t)}(e)$ by $\Gamma^{(t)}_i(e)$ and $\Lambda^{(t)}_i(e)$ respectively. Informally, one can think of the edges in $\Gamma^{(t)}(e)$ as the edges that change their color during the $t^{th}$ update because $e$ now occupies their color, and the edges in $\Lambda^{(t)}(e)$ as the edges that change their color during the $t^{th}$ update because $e$ no longer occupies a color that they would rather have assigned to them.

\medskip
\noindent \textbf{Type $A$ and $B$ dirty edges.}
For some $t \in [\kappa]$, $t \in [T]$, we define the sets of edges $A_i^{(t)}$ and $B^{(t)}_i$ by
\[A_i^{(t)} := \left\{ e \in S^{(t-1)}_i \cup S^{(t)}_i \, \middle| \, \tilde \chi^{(t-1)}(e) \neq \tilde \chi^{(t)}(e) \right\},\]
\[B_i^{(t)} := F_i^{(t-1)} \oplus F_i^{(t)}. \]
In words, $A_i^{(t)}$ is the set of edges in round $i$ that change their tentative color during the $t^{th}$ update, and $B^{(t)}_i$ is the set of edges in round $i$ that fail at either time $t-1$ or $t$, but not both. We let $A^{(t)}$ and $B^{(t)}$ denote the sets $\bigcup_{i}A^{(t)}_i$ and $\bigcup_{i}B^{(t)}_i$ respectively, and refer to the edges in $A^{(t)}$ and $B^{(t)}$ as $A$-dirty and $B$-dirty respectively. We define an edge $e$ as being \textit{dirty} with respect to the $t^{th}$ update if the color $\chi(e)$ changes during the $t^{th}$ update and denote the set of such edges by $D^{(t)}$, and $D^{(t)} \cap S^{(t)}_i$ by $D_i^{(t)}$. The recourse of the $t^{th}$ update is precisely $|D^{(t)}|$.

\subsubsection{Basic Facts}\label{app:sec:basic facts}
We now give some basic facts about these definitions. Let $t \in [\kappa]$, $e^\star$ be the edge that is either inserted or deleted during the $t^{th}$ update, and let $i^\star$ denote $i_{e^\star}$.

\begin{lem}\label{app:lem:DleqAB}
We have that $|D^{(t)}| \leq O(|B^{(t)}|) + |A^{(t)}|$.
\end{lem}

\begin{proof}
    It is sufficient to argue that at most $O(|B^{(t)}|)$ many edges not in $A^{(t)}$ change their colors during the $t^{th}$ update. Clearly, any such edge must be contained in $F^{(t-1)} \cup F^{(t)}$. Since at most $O(|F^{(t-1)} \oplus F^{(t)}|)$ many edges in $F^{(t-1)} \cup F^{(t)}$ change their colors during the $t^{th}$ update (recall that our dynamic greedy algorithm changes the colors of at most $O(1)$ many edges in $H$ when adding or removing an edge from $H$) the lemma follows.
\end{proof}

\begin{lem}\label{app:lem:new B bound}
    We have that $|B^{(t)}| \leq 4|A^{(t)}| + 1$.
\end{lem}

\begin{proof}
    Consider the set of all the failed edges
    \[F^{(t)} = \left\{ e \in E^{(t)} \, \middle| \, \exists f \in N^{(t)}_{i_e}(e) \textnormal{ such that } \tilde \chi^{(t)}(e) = \tilde \chi^{(t)}(f) \right\} \cup \left\{e \in E^{(t)} \, \middle| \, \tilde \chi^{(t)}(e) = \perp \right\}\]
    that is defined by our tentative coloring $\tilde \chi^{(t)}$. We want to get an upper bound on the size of $B^{(t)} = F^{(t-1)} \oplus F^{(t)}$. Suppose we start with the set $F^{(t-1)}$ and are given the set of edges that change their tentative colors during the $t^{th}$ update, $A^{(t)}$. Let $A^{(t)} = \{e_1,\dots,e_\ell\}$. Now suppose we update the tentative colors of the first $r$ edges $e_1,\dots,e_r \in A^{(t)}$, and let $F(r)$ be the set of failed edges defined with respect to the tentative coloring where an edge $f \in E^{(t-1)} \cup E^{(t)} \setminus \{e_1,\dots,e_r\}$ receives color $\tilde \chi^{(t-1)}(f)$ and the edges $e_1,\dots,e_r$ receive colors $\tilde \chi^{(t)}(e_1),\dots,\tilde \chi^{(t)}(e_r)$ respectively. We can see that
    \[ |F^{(t-1)} \oplus F^{(t)}| \leq |F^{(t-1)} \oplus F(0)| + |F(0) \oplus F(1)| + | F(1) \oplus F(2) | + \dots + | F(\ell-1) \oplus F(\ell)|. \]
    Let $r \in [\ell]$ and let $e_{r} = (u,v)$, and consider how $F(r-1)$ changes into $F(r)$ after we change the tentative color of $e_{r}$ from $\tilde \chi^{(t-1)}(e_{r})$ to $\tilde \chi^{(t)}(e_{r})$ (assume for now that neither color is $\bot$). We have that $|F(r) \setminus F(r-1)| \leq 2$. This is because the only edges in $F(r) \setminus F(r-1)$ are edges incident to $u$ and $v$ with color $\tilde \chi^{(t)}(e_{r})$ that are not already in $F(r-1)$, and there can be at most one such edge incident on each of $u$ and $v$. By an analogous argument, $|F(r-1) \setminus F(r)| \leq 2$. It follows that $|F(r -1) \oplus F(r)| \leq 4$. A similar argument shows that if one of these colors is $\bot$ then $|F(r -1) \oplus F(r)| \leq 3$. Finally, we note that $F^{(t-1)} \oplus F(0) \subseteq \{e^\star\}$ and the lemma follows.
\end{proof}

\begin{cor}\label{app:lemDleqA}
    We have that $|D^{(t)}| \leq O(|A^{(t)}|) + O(1)$.
\end{cor}

\begin{proof}
    Follows immediately from Lemmas \ref{app:lem:DleqAB} and \ref{app:lem:new B bound}.
\end{proof}

\begin{lem}\label{app:lem:fact 1}
For all  $e \notin A^{(t)}$, we have that $\Gamma^{(t)}(e) = \varnothing$ and $\Lambda^{(t)}(e) = \varnothing$.
\end{lem}

\begin{proof}
    Since $e \notin A^{(t)}$, we have that $\tilde \chi^{(t)}(e) = \tilde \chi^{(t-1)}(e)$. Hence, $\Gamma^{(t)}(e) = \{ f \in N^{(t)}_{<i_e}(e) \, | \, \tilde \chi^{(t-1)}(f) = \tilde \chi^{(t-1)}(e) \}$ which is clearly empty since, for any $f \in \Gamma^{(t)}(e)$, $e$ and $f$ share an endpoint and $i_e < i_f$, and hence cannot have the same tentative color. Similarly, the set $\Lambda^{(t)}(e)$ is $\{ f \in N^{(t)}_{>i_e}(e) \, | \, \tilde \chi^{(t)}(f) = \tilde \chi^{(t)}(e) \}$ and is also empty by the same arguments.
\end{proof}

\begin{lem}\label{app:lem:fact 2}
For all  $e \in E^{(t-1)} \cup E^{(t)}$, we have that $\Gamma^{(t)}(e) \subseteq A^{(t)}$ and $\Lambda^{(t)}(e) \subseteq A^{(t)}$.
\end{lem}

\begin{proof}
    Let $e \in E^{(t)} \cup E^{(t-1)}$ and $f \in \Gamma^{(t)}(e)$. Then we know that $\tilde \chi^{(t)}(e) = \tilde \chi^{(t-1)}(f)$. Since $e$ and $f$ share an endpoint and $i_e < i_f$, we also have that $\tilde \chi^{(t)}(e) \neq \tilde \chi^{(t)}(f)$. Hence, it follows that $\tilde \chi^{(t)}(f) \neq \tilde \chi^{(t-1)}(f)$ and so $f \in A^{(t)}$. It follows that $\Gamma^{(t)}(e) \subseteq A^{(t)}$. By an analogous argument, we have that $\Lambda^{(t)}(e) \subseteq A^{(t)}$.
\end{proof}

\begin{lem}\label{app:lem:fact 3}
For all $i$ such that $i^\star < i \leq T$, we have that $A^{(t)}_i \subseteq \Gamma^{(t)}( A^{(t)}_{<i} ) \cup \Lambda^{(t)}( A^{(t)}_{<i})$.
\end{lem}

\begin{proof}
    We prove this lemma by showing that $\{ e \in A^{(t)}_i \, | \, \ell^{(t)}_e > \ell^{(t-1)}_e \} \subseteq \Gamma^{(t)}(A^{(t)}_{<i})$ and $\{ e \in A^{(t)}_i \, | \, \ell^{(t)}_e < \ell^{(t-1)}_e \} \subseteq \Lambda^{(t)}(A^{(t)}_{<i})$, which implies that
\begin{align*}
    A^{(t)}_i &= \{ e \in A^{(t)}_i \, | \, \ell^{(t)}_e > \ell^{(t-1)}_e \} \sqcup  \{ e \in A^{(t)}_i \, | \, \ell^{(t)}_e < \ell^{(t-1)}_e\}\\
    &\subseteq \Gamma^{(t)}(A^{(t)}_{<i}) \cup \Lambda^{(t)}(A^{(t)}_{<i}).\\
\end{align*}
Let $e \in A^{(t)}_i$ such that $\ell^{(t)}_e > \ell^{(t-1)}_e$ and $i = i_e$. Then there exists some $f \in N^{(t)}(e) \cap S^{(t)}_{< i}$ such that $\tilde \chi^{(t)}(f) = c_e(\ell_e^{(t-1)}) = \tilde \chi^{(t-1)}(e)$. Hence, $e \in \Gamma^{(t)}(f)$. Since $e$ and $f$ share an endpoint and $i_f < i_e$, we must have that $\tilde \chi^{(t-1)}(f) \neq \tilde \chi^{(t-1)}(e) = \tilde \chi^{(t)}(f)$ and so $f \in A^{(t)}_{<i}$. It follows that $e \in \Gamma^{(t)}(A^{(t)}_{<i})$ and hence $\{ e \in A^{(t)}_i \, | \, \ell^{(t)}_e > \ell^{(t-1)}_e \} \subseteq \Gamma^{(t)}(A^{(t)}_{<i})$. By a similar argument, we get that $\{ e \in A^{(t)}_i \, | \, \ell^{(t)}_e < \ell^{(t-1)}_e \} \subseteq \Lambda^{(t)}(A^{(t)}_{<i})$.
\end{proof}

\begin{cor}\label{app:lem:fact 4}
    For all $i$ such that $i^\star < i \leq T$, we have that $A^{(t)}_i = \Gamma_i^{(t)}( A^{(t)}_{<i} ) \cup \Lambda_i^{(t)}( A^{(t)}_{<i})$.
\end{cor}

\begin{proof}
    It follows by Lemmas~\ref{app:lem:fact 2} and \ref{app:lem:fact 3} that $A^{(t)}_i \subseteq \Gamma^{(t)}( A^{(t)}_{<i} ) \cup \Lambda^{(t)}( A^{(t)}_{<i}) \subseteq A^{(t)}$. By intersecting with $S^{(t)}_i$, it follows that $A^{(t)}_i \subseteq \Gamma_i^{(t)}( A^{(t)}_{<i} ) \cup \Lambda_i^{(t)}( A^{(t)}_{<i}) \subseteq A_i^{(t)}$.
\end{proof}

\subsubsection{Bounding the Expected Recourse}

We can now bound the expected recourse of our dynamic algorithm by showing that $\mathbb E[|A^{(t)}|] = O(1/\epsilon^4)$ for all $t \in [\kappa]$. By then applying Corollary~\ref{app:lemDleqA}, this immediately implies that the expected recourse of our algorithm while handling an update is $O(1/\epsilon^4)$. We devote the rest of this section to proving the following lemma.

\begin{lem}\label{app:thm:AB bound}
    For all $t \in [\kappa]$, $\mathbb E \left[ \middle| A^{(t)} \middle| \right] \leq 1/\epsilon^4 + o(1)$.
\end{lem}

\subsection{Proof of Lemma~\ref{app:thm:AB bound}}

\noindent
For the remainder of this section, we fix some $t \in [\kappa]$. Let $e^\star$ denote the edge that is either inserted or deleted during the $t^{th}$ update and let $i^\star$ denote $i_{e^\star}$. We can make the following observation.

\begin{obs}\label{app:obs:A start}
$A^{(t)}_i = \varnothing$ for all $0 \leq i < i^\star$ and $A^{(t)}_{i^\star} \subseteq \{e^\star\}$.
\end{obs}

\noindent
We now establish a relationship between the expected sizes of the sets $A^{(t)}_i$ when $i$ is larger than $i^\star$, and use this to bound the expected size of $A^{(t)}$. 

Let $\mathcal E$ be the event that $\mathcal N(e^\star,2T+2)$ is a tree at time $t$ and $t-1$. Recall that since we are only concerned with bounding the recourse at time $t$, we only consider the edges in the graph $\mathcal G_j$, where $e^\star$ is contained in $\mathcal G_j$. Hence, $\mathcal N(e^\star, 2T+2)$ denotes $\mathcal N_{\mathcal G_j}(e^\star, 2T+2)$ in this context. Let $\mathcal Z^{(t)}$ denote the event that $\mathcal Z$ (defined in \Cref{app:concetration}) occurs at time $t$.
Let $\mathcal Y^{(t)}$ denote the event that the statement in Observation~\ref{lem:Delta(G_i)} occurs at time $t$. We first argue that we can assume that the event $\mathcal E \cap \mathcal Z^{(t-1)} \cap \mathcal Z^{(t)} \cap \mathcal Y^{(t-1)} \cap \mathcal Y^{(t)}$ occurs.

\begin{lem}\label{app:lem:recourse tree assumption}
    We have that
    \[\mathbb E[|A^{(t)}|] = \mathbb E[|A^{(t)}| \, | \, \mathcal E \cap \mathcal Z^{(t-1)} \cap \mathcal Z^{(t)} \cap \mathcal Y^{(t-1)} \cap \mathcal Y^{(t)}] + o(1).\]
\end{lem}

\begin{proof}
    By the law of total expectation, we have that
    \[ \mathbb E[|A^{(t)}|] = \mathbb E[|A^{(t)}| \, | \, \mathcal E] \cdot \Pr[\mathcal E] + \mathbb E[|A^{(t)}| \, | \, \neg {\mathcal E}] \cdot \Pr[\neg {\mathcal E}]. \]
    By Lemma~\ref{lem:kill small cycles original}, the event $\mathcal E$ does not occur with probability at most $3(\Delta')^{10T+10}/\Delta = 3 \Delta^{1/(3T)-2/3}$. By Lemma~\ref{lem:locality of nibble}, we can see that every edge in $A^{(t)}$ is contained in $\mathcal N(e^\star, T+1)$. Hence, the number of edges contained in $\mathcal N(e^\star, T+1)$ is an upper bound on $|A^{(t)}|$. By Observation \ref{lem:Delta(G_i)}, we have that $\Delta(\mathcal G_j) \leq (1 + \epsilon)\Delta' = (1 + \epsilon)\Delta^{1/(30T)}$ with probability at least $1 - 1/n^{31}$. It follows that $\mathcal N(e^\star, T+1)$ contains at most $2((1 + \epsilon)\Delta^{1/(30T)})^{T+1} \leq (3/\epsilon)\Delta^{1/15}$ many edges with the same probability. Putting everything together, we get that
    \[ \mathbb E[|A^{(t)}|] = \mathbb E[|A^{(t)}| \, | \, \mathcal E] \cdot \Pr[\mathcal E] + \mathbb E[|A^{(t)}| \, | \, \neg {\mathcal E}] \cdot \Pr[\neg {\mathcal E}] \leq \mathbb E[|A^{(t)}| \, | \, \mathcal E] + 9 \Delta^{-4/15} /\epsilon^2 \]
    with probability at least $1 - 1/n^{31}$. Noting that $|A^{(t)}| \leq |E| \leq n^2$, we have that
    \[ \mathbb E[|A^{(t)}|] \leq n^2 \cdot (1/n^{31}) + \mathbb E[|A^{(t)}| \, | \, \mathcal E] +  9 \Delta^{-4/15} /\epsilon^2 \leq \mathbb E[|A^{(t)}| \, | \, \mathcal E] +  9 \Delta^{-4/15} /\epsilon^2 + 1/n^{29}. \]
    Finally, letting $\mathcal X = \mathcal Z^{(t-1)} \cap \mathcal Z^{(t)} \cap \mathcal Y^{(t-1)} \cap \mathcal Y^{(t)}$, the result follows by noting that
    \begin{align*}
        \mathbb E[|A^{(t)}| \, | \, \mathcal E] &\leq \mathbb E[|A^{(t)}| \, | \, \mathcal E \cap \mathcal X] + n^2 \cdot \Pr[\neg \mathcal X]\\
        &\leq \mathbb E[|A^{(t)}| \, | \, \mathcal E \cap \mathcal X] + 4/n^{12}.
    \end{align*}
\end{proof}

\noindent For the rest of this section, we assume that $\mathcal N(e^\star,2T+2)$ is a tree at time $t$ and $t-1$. We also fix the random bits used by the algorithm to determine the partition of the graph so that event $\mathcal Y^{(t-1)} \cap \mathcal Y^{(t)}$ occurs and the random bits used by the algorithm to determine the rounds of edges so that the event $\mathcal Z^{(t-1)} \cap \mathcal Z^{(t)}$ occurs. Note that, in this context, the upper bound on $\Delta(\mathcal G_j)$ is $(1 + \epsilon)\Delta'$, where $\Delta' = \Delta^{1/(30T)}$. We implicitly condition all probabilities on these events unless stated otherwise. By Lemma \ref{lem:locality of nibble}, since all edges that change their tentative color or failed status during the $t^{th}$ update are contained in $\mathcal N^{(t)}(e^\star, T+1)$, we can assume that the input graph is $\mathcal N^{(t)}(e^\star, 2T+2)$ while trying to bound the number of such edges.

\begin{lem}\label{app:lem:Gamma size}
    Let $i \in [T]$ and $e \in S_{<i}^{(t)}$, then we have that $\mathbb E [ |\Gamma_i^{(t)}(e) | ] \leq 2\epsilon$.
\end{lem}

\begin{proof}
    We begin by observing that by linearity of expectation
    \[\mathbb E [ |\Gamma_i^{(t)}(e) | ] = \sum_{f \in N^{(t)}_{i}(e)} \Pr [\tilde \chi^{(t-1)}(f) = \tilde \chi^{(t)}(e) ].\]
    If $e \notin A_{<i}^{(t)}$, then by Lemma \ref{app:lem:fact 1} we have that $\Gamma^{(t)}(e) = \varnothing$.
    Assume that $e \in A_{<i}^{(t)}$. Then $e$ is contained in $\mathcal N(e^\star, T+1)$ and hence $\mathcal N(e, T+1) \subseteq \mathcal N(e^\star, 2T+2)$ is a tree. 
    Let $e = (u,v)$ and fix the random bits used by the algorithm in the first $i-1$ rounds that determine the palettes $P_i^{(t-1)}(u)$ and $P^{(t-1)}_i(v)$.
    Let $f = (u,w) \in N^{(t-1)}_i(e)$ (note that $N^{(t)}_i(e) = N^{(t-1)}_i(e)$ since $i_{e^\star} < i$). 
    Then $e$ and $w$ are disconnected in the graphs $(V, S^{(t-1)}_{<i})$ and $(V, S^{(t)}_{<i})$ since $\mathcal N^{(t)}(e^\star, 2T+2)$ and $\mathcal N^{(t-1)}(e^\star, 2T+2)$ are both trees and one is a subgraph of the other.
    Letting $c = \tilde \chi^{(t)}(e)$, we have that $c$ does not depend on the random bits that determine the palette $P_i^{(t-1)}(w)$.
    We can now apply Lemmas \ref{lem:symmetry of nibble 3} and \ref{lem: node palette concentration} to get that
    \[\Pr [\tilde \chi^{(t-1)}(f) = c  ] \leq \frac{1}{|P_i^{(t-1)}(u)|} \leq \frac{1}{(1+\epsilon)^2(1-\epsilon)^{i-1}\Delta'}.\]
    The same holds for all $f \in N^{(t-1)}_i(v)$.
    It follows that
    \[\mathbb E [ |\Gamma_i^{(t)}(e) | ] = \sum_{f \in N^{(t)}_{i}(e)} \Pr [\tilde \chi^{(t-1)}(f) = c ]\]
    \[\leq  \sum_{f \in N^{(t-1)}_{i}(u)} \Pr [ \tilde \chi^{(t-1)}(f) = c ] + \sum_{f \in N^{(t-1)}_{i}(v)} \Pr [ \tilde \chi^{(t-1)}(f) = c ] \]
    \[\leq \frac{|N^{(t-1)}_{i}(u)| + |N^{(t-1)}_{i}(v)|}{(1+\epsilon)^2(1-\epsilon)^{i-1}\Delta'} \leq \frac{2\epsilon(1+\epsilon)^2(1-\epsilon)^{i-1}\Delta'}{(1+\epsilon)^2(1-\epsilon)^{i-1}\Delta'}  = 2\epsilon.\]
\end{proof}

\begin{lem}\label{app:lem:Lambda size}
    Let $i \in [T]$ and $e \in S_{<i}^{(t)}$, then we have that $\mathbb E [ |\Lambda_i^{(t)}(e) | ] \leq 2\epsilon$.
\end{lem}

\begin{proof}
    We begin by observing that by linearity of expectation
    \[\mathbb E [ |\Lambda_i^{(t)}(e) | ] = \sum_{f \in N^{(t)}_{i}(e)} \Pr [\tilde \chi^{(t)}(f) = \tilde \chi^{(t-1)}(e) ].\]
    If $e \notin A_{<i}^{(t)}$, then by Lemma \ref{app:lem:fact 1} we have that $\Lambda^{(t)}(e) = \varnothing$.
    Assume that $e \in A_{<i}^{(t)}$. Then $e$ is contained in $\mathcal N(e^\star, T+1)$ and hence $\mathcal N(e, T+1) \subseteq \mathcal N(e^\star, 2T+2)$ is a tree.
    Let $e = (u,v)$ and fix the random bits used by the algorithm in the first $i-1$ rounds that determine the palettes $P_i^{(t)}(u)$ and $P^{(t)}_i(v)$.
    Let $f = (u,w) \in N^{(t)}_i(e)$. 
    Then $e$ and $w$ are disconnected in the graphs $(V, S^{(t-1)}_{<i})$ and $(V, S^{(t)}_{<i})$ since $\mathcal N^{(t)}(e^\star, 2T+2)$ and $\mathcal N^{(t-1)}(e^\star, 2T+2)$ are both trees and one is a subgraph of the other.
    Letting $c = \tilde \chi^{(t)}(e)$, we have that $c$ does not depend on the random bits that determine the palette $P_i^{(t)}(w)$.
    We can now apply Lemmas \ref{lem:symmetry of nibble 3} and \ref{lem: node palette concentration} to get that
    \[\Pr [\tilde \chi^{(t)}(f) = c ] \leq \frac{1}{|P^{(t)}_i(u)|} \leq \frac{1}{(1+\epsilon)^2(1-\epsilon)^{i-1}\Delta'}.\]
    The same holds for all $f \in N^{(t)}_i(v)$.
    It follows that
    \[\mathbb E [ |\Lambda_i^{(t)}(e) | ] = \sum_{f \in N^{(t)}_{i}(e)} \Pr [ \tilde \chi^{(t)}(f) = c ]\]
    \[\leq  \sum_{f \in N^{(t)}_{i}(u)} \Pr [\tilde \chi^{(t)}(f) = c ] + \sum_{f \in N^{(t)}_{i}(v)} \Pr [\tilde \chi^{(t)}(f) = c ]\]
    \[\leq \frac{|N^{(t)}_{i}(u)| + |N^{(t)}_{i}(v)|}{(1+\epsilon)^2(1-\epsilon)^{i-1}\Delta'} \leq \frac{2\epsilon(1+\epsilon)^2(1-\epsilon)^{i-1}\Delta'}{(1+\epsilon)^2(1-\epsilon)^{i-1}\Delta'}  = 2\epsilon.\]
\end{proof}

\begin{lem}\label{app:lem:A bound}
    For all $i$ such that $i^\star < i \leq T$, we have that $\mathbb E[|A^{(t)}_i|] \leq 4 \epsilon \cdot  \mathbb E[|A^{(t)}_{< i}|]$.
\end{lem}

\begin{proof}
    We know by Corollary \ref{app:lem:fact 4} that
    \[A^{(t)}_i = \Gamma_i^{(t)}( A^{(t)}_{<i} ) \cup \Lambda_i^{(t)}( A^{(t)}_{<i}) = \bigcup_{e \in A^{(t)}_{<i}} \left( \Gamma_i^{(t)}(e) \cup \Lambda_i^{(t)}(e) \right). \]
    From this, we can immediately deduce that
    \[\mathbb E[|A^{(t)}_i|] \leq \sum_{e \in A^{(t)}_{<i}} \left( \mathbb E[|\Gamma_i^{(t)}(e)|] + \mathbb E[|\Lambda_i^{(t)}(e)|] \right)\]
    by using linearity of expectation. It then follows from Lemmas \ref{app:lem:Gamma size} and \ref{app:lem:Lambda size} that $\mathbb E[|A^{(t)}_i|] \leq 4\epsilon \cdot |A^{(t)}_{<i}|$. The lemma follows by taking expectations on both sides.
\end{proof}

\begin{lem}\label{app:lem:A recurrence}
For all $i$ such that $i^\star < i \leq T$, we have that $\mathbb E[|A^{(t)}_{\leq i }|] \leq (1 + 4\epsilon) \cdot \mathbb E[|A^{(t)}_{\leq i - 1}|]$.
\end{lem}

\begin{proof}
    By applying Lemma \ref{app:lem:A bound} we get that
    \[\mathbb E[|A^{(t)}_{\leq i}|] = \mathbb E[|A^{(t)}_i|] + \mathbb E[|A^{(t)}_{< i}|] \leq 4\epsilon \cdot \mathbb E[|A^{(t)}_{< i}|] + \mathbb E[|A^{(t)}_{< i}|] = (1 + 4\epsilon) \cdot \mathbb E[|A^{(t)}_{\leq i - 1}|].\]
\end{proof}

\begin{lem}\label{app:lem:A bound 2}
We have that $\mathbb E[|A^{(t)}|] \leq 1 /\epsilon^4$.
\end{lem}

\begin{proof}
    It follows from Observation \ref{app:obs:A start} and Lemma \ref{app:lem:A recurrence} that
    \[\mathbb E[|A^{(t)}|] = \mathbb E[|A^{(t)}_{\leq T}|] \leq (1 + 4\epsilon)^{T - i^\star} \cdot \mathbb E[|A^{(t)}_{\leq i^\star}|] \leq (1 + 4\epsilon)^{T} \leq e^{4 \log(1/\epsilon)} = 1/\epsilon^4. \]
\end{proof}

\noindent
Finally, using Lemma \ref{app:lem:recourse tree assumption}, we remove the conditioning on the event $\mathcal E \cap \mathcal Z^{(t-1)} \cap \mathcal Z^{(t)} \cap \mathcal Y^{(t-1)} \cap \mathcal Y^{(t)}$ to get that
\[\mathbb E[|A^{(t)}|] = \mathbb E[|A^{(t)}| \, | \, \mathcal E \cap \mathcal Z^{(t-1)} \cap \mathcal Z^{(t)}] + o(1) = O(1/\epsilon^4).\]

\section{Implementing our Dynamic Algorithm (Full Version)}\label{sec:app:datastructs}

We now proceed to give a more detailed description of our algorithm, which we then implement with appropriate data structures. The main result in this appendix is \Cref{thm:main 2}, which is restated below.

\begin{thm}
    There exists a dynamic algorithm that, given a dynamic graph $G$ that is initially empty and evolves by a sequence of $\kappa$ edge insertions and deletions by means of an oblivious adversary, and a parameter $\Delta \geq (100 \log n/\epsilon^4)^{(30/\epsilon) \log (1/\epsilon)}$ such that the maximum degree of $G$ is at most $\Delta$ at all times,
    maintains a $(1 + 61\epsilon)\Delta$-edge coloring of $G$ and
    has an expected worst-case update time of $O(\log^4(1/\epsilon)/\epsilon^9)$.
\end{thm}

\subsection{Our Algorithm}

\textbf{High-level approach.} Let $t \in [\kappa]$, and suppose we have the edge coloring $\chi^{(t-1)}$ of $G^{(t-1)}$. Let $e^\star$ be the edge inserted/deleted at time $t$ and let $i^\star$ denote $i_{e^\star}$. We know that the edge $e^\star$ is $A$-dirty (as long as it doesn't fail to be assigned a tentative color) and that there are no other $A$-dirty edges in the first $i^\star$ rounds. The main idea behind our dynamic algorithm is to iterate through the rounds $i^\star+1,\dots,T$ and observe how the changes in the tentative coloring $\tilde \chi$ propagate through the rounds, finding all of the $A$-dirty edges. By designing an appropriate data structure that updates the failed edges as we update the tentative colors of edges, we show that we can explicitly maintain the failed edges by just finding the $A$-dirty edges and updating their tentative colors.
We find all of the $A$-dirty edges in round $i$ inductively by using Corollary \ref{app:lem:fact 4}. In other words, if we can find all of the $A$-dirty edges in the first $i-1$ rounds then we can find all of the $A$-dirty edges in round $i$ by using the fact that $A^{(t)}_i = \Gamma_i^{(t)}( A^{(t)}_{<i} ) \cup \Lambda_i^{(t)}( A^{(t)}_{<i})$.
Algorithm \ref{app:alg:update} gives the pseudocode for this high-level approach.

\begin{algorithm}[H]
\caption{\textsc{Update-Coloring}$(e^\star)$}\label{app:alg:update}
\begin{algorithmic}[1]
    \State Find $\tilde \chi^{(t)}(e^\star)$
    \If{$\tilde \chi^{(t)}(e^\star) = \tilde \chi^{(t-1)}(e^\star)$}
        \State \Return
    \EndIf
    \State $A_{i^\star}^{(t)} \leftarrow \{e^\star\}$
    \For{$i = i^\star + 1,\dots,T:$}
        \State $A^{(t)}_i \leftarrow \Gamma_i^{(t)}( A^{(t)}_{<i} ) \cup \Lambda_i^{(t)}( A^{(t)}_{<i})$
        \State Find $\tilde \chi^{(t)}(e)$ for all $e \in A^{(t)}_i$
    \EndFor
\end{algorithmic}
\end{algorithm}

\medskip
\noindent \textbf{Maintaining the tentative colorings.} The main technical challenge that we face is maintaining the tentative colorings for the graphs $\mathcal G_j$. In Section \ref{app:sec:key data structure}, we design a data structure that, given a dynamic graph $G'$, is capable of maintaining and updating a tentative coloring of $G'$ so that the output matches that of Algorithm~\ref{app:alg:nibble} when they use the same random bits. More precisely, our dynamic data structure will be given the round $i_e$ and color sequence $c_e$ for each edge $e$ in $G'$, and will maintain the color indices $\ell_e$ of each edge so that they match the corresponding indices produced by Algorithm~\ref{app:alg:nibble} when run on input $G'$. This defines a tentative coloring (by setting $\tilde \chi(e) = c_e(\ell_e)$) and a corresponding set of failed edges, which the data structure maintains explicitly. Given an edge $e$ to be inserted or deleted from $G'$, our data structure is capable of efficiently identifying all of the edges that need to change their color indices in order for the values maintained by the data structure to match the output produced by running Algorithm~\ref{app:alg:nibble} on the updated graph. We will create $\eta$ copies $\mathcal D_1,\dots,\mathcal D_\eta$ of this data structure and use them to maintain the tentative colorings of the graphs $\mathcal G_1,\dots, \mathcal G_\eta$ by feeding the graph $\mathcal G_j$ to $\mathcal D_j$. Since each $\mathcal D_j$ explicitly maintains the edges that fail while tentatively coloring $\mathcal G_j$, we can use them to explicitly maintain the graph $H$ consisting of all failed edges by adding or removing an edge $e$ from $H$ every time one of the $\mathcal D_j$ adds or removes $e$ from its collection of failed edges. In Section~\ref{sec:dynamic nibble output}, we prove the following lemma.

\begin{lem}\label{lem:dynamic data struc}
    There exists a dynamic data structure that, given a dynamic graph $G$ and a parameter $\Delta$, is capable of explicitly maintaining the tentative coloring $\tilde \chi$ and the set of failed edges $F$ defined by running Algorithm~\ref{app:alg:nibble} on $G$, and has an expected update time of $O(\log^4(1/\epsilon)/\epsilon^5 \cdot \mathbb E[|A|])$, where $A$ is the set of $A$-dirty edges during an update.
\end{lem}

\medskip
\noindent \textbf{Maintaining the greedy coloring.} We now show that there exists an algorithm that can dynamically maintain a greedy coloring of a graph and has $O(1)$ expected update time. Our algorithm has the property that each edge insertion or deletion can only change the colors of at most $O(1)$ many edges in the graph. Furthermore, this coloring is maintained explicitly and hence can be queried in $O(1)$ time. More precisely, in Section~\ref{sec:proving greedy}, we prove the following lemma.

\begin{lem}\label{app:thm:fast greedy 3}
    There exists a dynamic data structure that, given a dynamic graph $G = (V,E)$ that undergoes edge insertions and deletions, can explicitly maintain a $3\Delta(G)$-edge coloring of $G$, where $\Delta(G)$ is the current maximum degree of $G$, has an expected update time of $O(1)$, and only changes the colors of $O(1)$ many edges per update. 
\end{lem}

\medskip
\noindent \textbf{Putting everything together.} We can now use the dynamic algorithm that is described in Lemma~\ref{lem:dynamic data struc} to maintain the tentative coloring $\tilde \chi$ and the graph $H$ defined by running Algorithm~\ref{app:alg:static} on $G$, and then run the dynamic greedy algorithm described in Lemma~\ref{app:thm:fast greedy 3} on $H$ in order to efficiently maintain a $3\Delta(H)$-edge coloring of $H$. Recall that we create $\eta$ copies of this data structure and use them to maintain the tentative colorings and failed edges of each $\mathcal G_j$ individually. By Lemma~\ref{app:thm:AB bound}, we know that the expected type $A$ recourse (in the subsampled graph $\mathcal G_j$ that the update occurs) is $O(1/\epsilon^4)$, and hence the expected update time of our algorithm is $O(\log^4(1/\epsilon)/\epsilon^9)$. Since the graph $H$ is maintained explicitly, we insert or delete at most $O(\log^4(1/\epsilon)/\epsilon^9)$ many edges from $H$. Since our dynamic greedy algorithm has an expected update time of $O(1)$, this only leads to a $O(1)$ multiplicative factor in overhead to the update time. Finally, we note that the assumption that the randomness is generated in advance was made purely for analytic purposes so that we could argue that our dynamic algorithm and our static algorithm produce the same output when using the same random bits. However, if we generate the random bits for an edge $e$ on the fly when it is inserted into the graph (i.e. decide which graph $\mathcal G_j$ to place it into and sample its round $i_e$ and color sequence $c_e$ independently of all previous random processes) then the distribution of the random bits assigned to the edges in the graph at any given time does not change, and the same guarantees hold. Our main theorem follows.

\begin{thm}\label{thm:main}
    There exists a dynamic algorithm that, given a dynamic graph $G$ that is initially empty and evolves by a sequence of $\kappa$ edge insertions and deletions by means of an oblivious adversary, and a parameter $\Delta \geq (100 \log n/\epsilon^4)^{(30/\epsilon) \log (1/\epsilon)}$ such that the maximum degree of $G$ is at most $\Delta$ at all times,
    maintains a $(1 + 61\epsilon)\Delta$-edge coloring of $G$ with probability at least $1 - 8/n^6$ at each time $t \in [\kappa]$, and
    has an expected worst-case update time of $O(\log^4(1/\epsilon)/\epsilon^9)$.
\end{thm}

\noindent By our proof of Theorem~\ref{app:thm:static}, it follows that, at any point in time, the graph $H$ has a maximum degree of at most $19\epsilon\Delta$ with probability at least $1 - 8/n^6$. Hence, in the event that $\Delta(H) > 19\epsilon\Delta$, we can resample all of the randomness used by our algorithm by deleting and reinserting every edge (and resampling all of the random bits for each edge in the process). This will take $O_\epsilon(m)$ expected time, and we will have that $\Delta(H) > 19\epsilon\Delta$ with probability at most $8/n^6$ independently of the randomness that our algorithm had used before. We repeat this process of resampling all of the randomness until we have that $\Delta(H) \leq 19\epsilon\Delta$. In expectation, we do this at most $1/(1 - 8/n^6) - 1 = 1/\Omega(n^6)$ many times. Since we spend $O_\epsilon(m) \leq O_\epsilon(n^2)$ expected time resampling all the randomness each time we do this, it follows that this process takes $o(1)$ time in expectation. However, since our algorithm uses at most $(1 + 61\epsilon)\Delta$ colors in total as long as $\Delta(H) \leq 19\epsilon\Delta$, this ensures that we always maintain a $(1 + 61\epsilon)\Delta$-edge coloring of $G$, while only incurring an additive $o(1)$ factor in the expected worst-case update time. Thus, we have the following corollary.

\begin{cor}\label{thm:main 2}
    There exists a dynamic algorithm that, given a dynamic graph $G$ that is initially empty and evolves by a sequence of $\kappa$ edge insertions and deletions by means of an oblivious adversary, and a parameter $\Delta \geq (100 \log n/\epsilon^4)^{(30/\epsilon) \log (1/\epsilon)}$ such that the maximum degree of $G$ is at most $\Delta$ at all times,
    maintains a $(1 + 61\epsilon)\Delta$-edge coloring of $G$ and
    has an expected worst-case update time of $O(\log^4(1/\epsilon)/\epsilon^9)$.
\end{cor}

\noindent Finally, by amortizing over a sufficiently long sequence of updates and periodically resampling the random bits, we can get $O_\epsilon(1)$ amortized update time with high probability.

\begin{cor}\label{thm:main 3}
    There exists a dynamic algorithm that, given a dynamic graph $G$ that is initially empty and evolves by a sequence of $\kappa$ edge insertions and deletions by means of an oblivious adversary, and a parameter $\Delta \geq (100 \log n/\epsilon^4)^{(30/\epsilon) \log (1/\epsilon)}$ such that the maximum degree of $G$ is at most $\Delta$ at all times, 
    maintains a $(1 + 61\epsilon)\Delta$-edge coloring of $G$ and
    has an amortized update time of $O(\log^4(1/\epsilon)/\epsilon^9)$ over the whole update sequence w.h.p. as long as $\kappa$ is a sufficiently large polynomial in $n$.
\end{cor}

\begin{proof}
    We first make a further modification to the algorithm from Corollary~\ref{thm:main 2} so that, every $n^2$ updates, our algorithm samples new random bits for each edge in the same way as described above, taking $O_\epsilon(m)$ time in expectation. This splits up the update sequence $\sigma_1,\dots,\sigma_\kappa$ in \textit{epochs} of length $n^2$. Let $X_i$ be a random variable denoting the total update time of our algorithm during the $i^{th}$ epoch. We have that $\mathbb E[X_i] = O(n^2\log^4(1/\epsilon)/\epsilon^9)$. Now note that, given the data structures that we use to implement our algorithm, there exists some polynomial $p(n)$ such that $X_i$ is always at most $p(n)$. Let $X = \sum_i X_i$ be the total update time of our algorithm. Since we sample fresh randomness at the end of each epoch, the random variables $X_1,\dots,X_\tau$ are independent. Hence, we can apply Hoeffding bounds to get that $X = O(\kappa \log^4(1/\epsilon)/\epsilon^9)$ with probability at least $1 - \exp{\left(-\tau \cdot \Theta_\epsilon(n^4)/p(n)^2 \right)}$. Hence, as long as $\tau = \Theta(\kappa / n^2)$ is a sufficiently large polynomial in $n$, $X = O(\kappa \log^4(1/\epsilon)/\epsilon^9)$ and the amortized update time of our algorithm is $O(\log^4(1/\epsilon)/\epsilon^9)$ with high probability.
\end{proof}

\noindent The rest of this section is now devoted to proving Lemmas~\ref{lem:dynamic data struc} and \ref{app:thm:fast greedy 3}. In Section~\ref{app:sec:key data structure}, we design the key data structure that we need to maintain a tentative coloring and associated failed edges. In Section~\ref{sec:dynamic nibble output}, we show how to use this data structure in order to achieve this and prove Lemma~\ref{lem:dynamic data struc}. Finally, in Section~\ref{sec:proving greedy}, we prove a (slight generalization of) Lemma~\ref{app:thm:fast greedy 3}.

\subsection{Key Data Structure}\label{app:sec:key data structure}

In this section, we design a dynamic data structure that will allow us to dynamically maintain a tentative coloring and the corresponding set of failed edges generated by this coloring.
Given a dynamic graph $G = (V,E)$ and a parameter $T$, for each edge $e \in E$ our dynamic data structure maintains: the round of the edge, $i_e$, the sequence of colors of the edge, $c_e(1),\dots,c_e(K)$, and a \emph{color index} $\ell_e \in \{0,\dots,K\}$. This defines a partition of the edges $S_1,\dots,S_{T+1}$ where $S_i = \{e \in E \, | \, i_e = i\}$, and an edge coloring $\tilde \chi : E \longrightarrow \mathcal C \cup \{\perp\}$ of $G$ where $\tilde \chi(e) = c_e(\ell_e)$ and $c_e(0)$ is defined as $\perp$. The edge coloring $\tilde \chi$ then defines a set of failed edges $F = \{ e \in E \, | \, \exists f \in N_{i_e}(e) \textnormal{ such that } \tilde \chi(e) = \tilde \chi(f)\} \cup \{e \in E \, | \, \tilde \chi(e) = \perp\}$ and a collection of palettes $P_i(u) = [(1 + \epsilon)\Delta] \setminus \tilde \chi(N_{< i}(u))$, $P_i(u,v) = P_i(u) \cap P_i(v)$. Our data structure maintains all of these objects and supports the following update and query operations.

\medskip
\noindent \textbf{Updates:} The data structure can be updated with the following operations, as described below.

\begin{itemize}
    \item \textsc{Insert}$(e, i, (c(1),\dots,c(K)))$: Inserts the edge $e$ into the graph $G$ and assigns the edge $e$ round $i$, color sequence $c$, and color index $0$.
    \item \textsc{Delete}$(e)$: Deletes the edge $e$ from the graph $G$.
    \item \textsc{Set-Color-Index}$(e, \ell)$: Sets the color index of the edge $e$ to $\ell$.
    \item \textsc{Reset-Color}$(e)$: Sets the color index of the edge $e$ to the smallest value $\ell$ such that $c_e(\ell) \in P_i(e)$, and to $0$ if no such index exists. Returns the edge $e$ if the color index of $e$ changes, and $\textsc{null}$ otherwise.
\end{itemize}

\medskip
\noindent \textbf{Queries:} The data structure can answer the following types of queries, as described below.

\begin{itemize}
    \item \textsc{Node-Palette-Query}$(u, i, c)$: The input to this query is a node $u \in V$, an integer $i \in [T]$, and a color $c \in \mathcal C$. In response, the data structure outputs YES if $c \in P_i(u)$ and NO otherwise.
    \item \textsc{Edge-Palette-Query}$(e, i, c)$: The input to this query is an edge $e \in E$, an integer $i \in [T]$, and a color $c \in \mathcal C$. In response, the data structure outputs YES if $c \in P_i(e)$ and NO otherwise.
    \item \textsc{Failed-Edge-Query}$(e)$: The input to this query is an edge $e \in E$. In response, the data structure returns YES if $e \in F$ and NO otherwise.
    \item \textsc{Color-Query}$(e)$: The input to this query is an edge $e \in E$. In response, the data structure returns the tentative color $\tilde \chi(e)$.
\end{itemize}

\noindent We now show how to implement this data structure so that each of these updates and queries run in $O_\epsilon(1)$ expected time. Our data structure also maintains other crucial internal data structures, which we describe below.

\subsubsection{Implementation}

We create hashmaps $\textsc{Round}:E \longrightarrow [T]$ and $\textsc{Color-Index}:E \longrightarrow [K]$ where $\textsc{Round}(e) = i_e$ and $\textsc{Color-Index}(e) = \ell_e$, allowing is to set and retrieve the rounds and color indices of edges in $O(1)$ (expected) time. The color sequences of edges are implemented as arrays, and we store a hashmap \textsc{Color-Sequence} that maps an edge $e$ to the position of this array, allowing us to retrieve the color $c_e(\ell)$ for an edge $e \in E$ and $\ell \in [K]$ in $O(1)$ time. We implement the set $F$ using a hashmap \textsc{Failed}, allowing us to insert, delete, and query the membership of an edge $e$ in $O(1)$ time. Given some edge $e \in E$, we can then compute $\tilde \chi(e)$ in $O(1)$ time by retrieving $\ell_e$ and returning $c_e(\ell_e)$ if $\ell \neq 0$ and $\perp$ otherwise.

\medskip
\noindent \textbf{Internal data structures.} Our data structure also maintains the following internal data structures that will be crucial for implementing our algorithm. We maintain hashmaps
\[\phi :  V \times [T] \times \mathcal C \longrightarrow 2^{N_i(u)} \quad \textnormal{ and } \quad \Psi :  V \times [T] \times \mathcal C \longrightarrow 2^{N_i(u)},\]
such that, for $u \in V$, $i \in [T]$, and $c \in \mathcal C$,
\[\phi_{u,i}(c)= \{e \in N_i(u) \, | \, \tilde \chi(e) = c\} \quad \textnormal{ and } \quad \Psi_{u,i}(c) = \{e \in N_i(u) \, | \, c \in c_e \}.\]
We take the convention that $\phi_{u,i}(\perp) = \varnothing$ and $\Psi_{u,i}(\perp) = \varnothing$. We also implement each set $\phi_{u,i}(c)$ as a hashmap, but maintain pointers between the elements of the set forming a doubly linked list. This allows for $O(1)$ time insertions, deletions, and membership queries, while also allowing us to return all the elements in the set in $O(|\phi_{u,i}(c)|)$ time. We implement the sets $\Psi_{u,i}(c)$ in the exact same way. 
By implementing $\phi$ as a hashmap, we avoid the preprocessing time of having to initialize each set in $\{\phi_{u,i}(c)\}_{u,i,c}$ upon creating the data structure and only initialize the set $\phi_{u,i}(c)$ when we want to insert an edge into $\phi_{u,i}(c)$, which takes $O(1)$ time. Otherwise, $\phi_{u,i}(c)$ points to $\textsc{null}$, and we know that the set is empty. Similarly, when $\phi_{u,i}(c)$ becomes empty, we can delete the map in $O(1)$ time, so that we only store maps that are non-empty.

\medskip
\noindent \textbf{Initialization.} It follows that, upon initializing our data structure, we only need to create the maps that store the rounds, color sequences, color indices, and failed edges, as well as the $2$ maps $\phi$ and $\Psi$. This takes $O(1)$ time in total.

\subsubsection{Handling Updates} We now show how to maintain our data structure as we perform updates to the graph and the color indices.

\medskip
\noindent \textbf{Implementing \textnormal{\textsc{Set-Color-Index}}.} We first note that changing the color index $\ell_e$ of an edge $e$ will \emph{not} change the set $\Psi_{u,i}(c)$ or whether or not an edge $f \neq e$ is contained in the set $\phi_{u,i}(c)$ for any $u \in V$, $i \in [T]$, and $c \in \mathcal C$. Hence, in order to update the $\phi_{u,i}$ and $\Psi_{u,i}$ maps after a color index update for edge $e$, we only need to insert and remove $e$ from the appropriate sets $\phi_{u,i}(c)$. On the other hand, some $O(1)$ many edges neighboring $e$ might have to be added or removed from $F$ after a color index update for $e$. Algorithm~\ref{app:alg:DS set-color-index} shows how we implement the \textsc{Set-Color-Index} procedure in order to appropriately update all the data structures used by our algorithm.

\begin{algorithm}[H]
\caption{\textsc{Set-Color-Index}$(e, \ell)$}\label{app:alg:DS set-color-index}
\begin{algorithmic}[1]
    \State $\triangleright$ Let $u$ and $v$ be the endpoints of $e$ and $i = i_e$
    \State $\ell^{\textsc{prev}}_e \leftarrow \ell_e$ and $c^{\textsc{prev}} \leftarrow \tilde \chi(e)$ 
    \State $\ell_e \leftarrow \ell$ and $c \leftarrow \tilde \chi(e)$
    \State $\triangleright$ Update the $\phi$ maps
    \If{$c^{\textsc{prev}} \neq \perp$}
        \State $\phi_{u,i}(c^{\textsc{prev}}) \leftarrow \phi_{u,i}(c^{\textsc{prev}}) \setminus \{e\}$
        \State $\phi_{v,i}(c^{\textsc{prev}}) \leftarrow \phi_{v,i}(c^{\textsc{prev}}) \setminus \{e\}$
    \EndIf
    \If{$c \neq \perp$}
        \State $\phi_{u,i}(c) \leftarrow \phi_{u,i}(c) \cup \{e\}$
        \State $\phi_{v,i}(c) \leftarrow \phi_{v,i}(c) \cup \{e\}$
    \EndIf
    \State $\triangleright$ Update the set of failed edges $F$
    \State $X \leftarrow \{e\}$
    \If{$|\phi_{u,i}(c)| = 2$}
        \State $X \leftarrow X \cup \phi_{u,i}(c)$
    \EndIf
    \If{$|\phi_{v,i}(c)| = 2$}
        \State $X \leftarrow X \cup \phi_{v,i}(c)$
    \EndIf
    \If{$|\phi_{u,i}(c^{\textsc{prev}})| = 1$}
        \State $X \leftarrow X \cup \phi_{u,i}(c^{\textsc{prev}})$
    \EndIf
    \If{$|\phi_{v,i}(c^{\textsc{prev}})| = 1$}
        \State $X \leftarrow X \cup \phi_{v,i}(c^{\textsc{prev}})$
    \EndIf
    \For{$f \in X$}
        \If{$|\phi_{u,i}(\tilde \chi(f))| > 1$ \textnormal{or} $|\phi_{v,i}(\tilde \chi(f))| > 1$}
            \State $F \leftarrow F \cup \{f\}$
        \Else
            \State $F \leftarrow F \setminus \{f\}$
        \EndIf
    \EndFor
\end{algorithmic}
\end{algorithm}

\begin{lem}
    The implementation of \textsc{Set-Color-Index} given by Algorithm~\ref{app:alg:DS set-color-index} correctly updates the data structure and runs in time $O(1)$.
\end{lem}

\begin{proof}
    Before the update, edge $e$ is contained in both $\phi_{u,i}(c^{\textsc{prev}})$ and $\phi_{v,i}(c^{\textsc{prev}})$, and no other such sets. After the update, since $\tilde \chi(e)$ changes to $c$, we remove $e$ from these sets and add it to $\phi_{u,i}(c)$ and $\phi_{v,i}(c)$, taking $O(1)$ time.
    After performing these operations, the $\phi$ maps are now correctly updated for the new tentative coloring. In order to see that our algorithm correctly updates the set $F$, note that the only edges that might need to be removed from $F$ are those incident on $u$ and $v$ at round $i$ that have tentative color $c^{\textsc{prev}}$. However, if there is more than 1 edge incident on $u$ (resp. $v$) at round $i$ with tentative color $c^{\textsc{prev}}$ after updating the map $\phi$, then no edge incident on $u$ (resp. $v$) will be removed from $F$ as all these edges will fail. Similarly, the only edges that might need to be added to $F$ are those incident on $u$ and $v$ at round $i$ that have tentative color $c$. However, if there are more than 2 edges incident on $u$ (resp. $v$) at round $i$ with tentative color $c$ after updating the $\phi$ maps, then no edge incident on $u$ (resp. $v$) will be added to $F$ since these edges will have failed before the update. It follows that a total of $O(1)$ many edges might need to be added or removed from $F$ and that we can identify them in $O(1)$ time. We can scan through all of these edges and for each one determine in $O(1)$ time whether it should or shouldn't be in $F$ by checking if $|\phi_{u,i}(\tilde \chi(f))| > 1$ or $|\phi_{v,i}(\tilde \chi(f))| > 1$ which is true if and only $f$ fails.
\end{proof}

\medskip
\noindent \textbf{Implementing \textnormal{\textsc{Insert}} and \textnormal{\textsc{Delete}}.} We first note that inserting or deleting an edge $e=(u,v)$ that has tentative color $\perp$ will not cause any change in the map $\phi$. Furthermore, it will not cause any edge $f \neq e$ to added or removed from $F$. Hence, we can insert or delete such edges and update the data structures easily. In order to delete an edge $e$ that has $\tilde \chi(e) \neq \perp$, we can first call \textsc{Set-Color-Index}$(e, 0)$ and then assume we are deleting an edge with tentative color $\perp$. Algorithms \ref{app:alg:DS insertion} and \ref{app:alg:DS deletion} show how we can implement this. These algorithms clearly update all the data structures correctly and can be implemented to run in $O(K)$ time.

\begin{algorithm}[H]
\caption{\textsc{Insert}$(e, i, (c(1),\dots,c(K)))$}\label{app:alg:DS insertion}
\begin{algorithmic}[1]
    \State $i_e \leftarrow i$
    \State $c_e \leftarrow (c(1),\dots,c(K))$
    \State $\ell_e \leftarrow 0$
    \State $F \leftarrow F \cup \{e\}$
    \For{$\ell' = 1\dots K$}
        \State $\Psi_{u,i}(c_e(\ell')) \leftarrow \Psi_{u,i}(c_e(\ell')) \cup \{e\}$
        \State $\Psi_{v,i}(c_e(\ell')) \leftarrow \Psi_{v,i}(c_e(\ell')) \cup \{e\}$
    \EndFor
\end{algorithmic}
\end{algorithm}

\begin{algorithm}[H]
\caption{\textsc{Delete}$(e)$}\label{app:alg:DS deletion}
\begin{algorithmic}[1]
    \State $\textsc{Set-Color-Index}(e, 0)$
    \State $i_e \leftarrow \textsc{null}$
    \State $c_e \leftarrow \textsc{null}$
    \State $\ell_e \leftarrow \textsc{null}$
    \State $F \leftarrow F \setminus \{e\}$
    \For{$\ell' = 1\dots K$}
        \State $\Psi_{u,i}(c_e(\ell')) \leftarrow \Psi_{u,i}(c_e(\ell')) \setminus \{e\}$
        \State $\Psi_{v,i}(c_e(\ell')) \leftarrow \Psi_{v,i}(c_e(\ell')) \setminus \{e\}$
    \EndFor
\end{algorithmic}
\end{algorithm}

\medskip
\noindent \textbf{Implementing \textnormal{\textsc{Reset-color}}.} We implement the \textsc{Reset-color} update by scanning through the list of colors $c_e$ and finding the smallest $\ell$ such that $c_e(\ell) \in P_{i_e}(e)$ by making calls to \textsc{Edge-Palette-Query}. Once we identify the smallest such $\ell$, we call $\textsc{Set-Color-Index}(e, \ell)$ and return $e$ if $\ell_e$ changes. If we cannot find such an index $\ell$, we call $\textsc{Set-Color-Index}(e, 0)$ and return $e$ if $\ell_e$ changes. As we will see in Section \ref{app:sec: answer queries}, each edge palette query can be implemented to run in time $O(T)$. Since we make at most $K$ such queries, one call to \textsc{Set-Color-Index} that takes $O(1)$ time, and the rest of the operations can be implemented in $O(1)$ time, it follows that \textsc{Reset-color} can be implemented to run in $O(KT)$ time. Algorithm \ref{app:alg:DS reset color} shows how we can implement this algorithm.

\begin{algorithm}[H]
\caption{\textsc{Reset-Color}$(e)$}\label{app:alg:DS reset color}
\begin{algorithmic}[1]
    \State $\ell' \leftarrow \ell_e$
    \For{$\ell = 1,\dots,K$}
        \If{$\textsc{Edge-Palette-Query}(e, i_e, c_e(\ell)) = \textnormal{YES}$}
            \State $\textsc{Set-Color-Index}(e, \ell)$
             \If{$\ell' \neq \ell$}
                 \State \Return{e}
             \EndIf
              \State \Return{\textsc{null}}
        \EndIf
    \EndFor
     \State $\textsc{Set-Color-Index}(e, 0)$
    \If{$\ell' \neq 0$}
         \State \Return{e}
    \EndIf
     \State \Return{\textsc{null}}
\end{algorithmic}
\end{algorithm}

\subsubsection{Answering the Queries}\label{app:sec: answer queries}

Since we maintain the set $F$ as a hashmap, given some edge $e$, we can answer $\textsc{Failed-Edge-Query}$ on edge $e$ in $O(1)$ time by directly checking whether the edge $e$ is contained in $F$. Given some $u \in V$, $i \in [T]$, and $c \in \mathcal C$, we can check whether $c \in P_i(u) = [(1 + \epsilon)\Delta] \setminus \tilde \chi(N_{< i}(u))$ by checking whether $\phi_{u, i'}(c) = \varnothing$ for all $i' < i$, which can be done in $O(i) \leq O(T)$ time. Hence, we can answer the query \textsc{Node-Palette-Query} in $O(T)$ time. Given some $e = (u,v) \in E$, $i \in [T]$, and $c \in \mathcal C$, we can check whether $c \in P_i(e)$ by making 2 calls to \textsc{Node-Palette-Query} and checking whether \textsc{Node-Palette-Query} $c \in P_i(u)$ and $c \in P_i(v)$. Hence, we can answer the query \textsc{Edge-Palette-Query} in $O(T)$ time. The fact that \textsc{Color-Query} can be implemented in $O(1)$ time follows immediately from the implementation.

\subsection{Proof of Lemma~\ref{lem:dynamic data struc}}\label{sec:dynamic nibble output}

Let $G = (V,E)$ be a dynamic graph that undergoes updates via a sequence of edge insertions and deletions, and let $\Delta$ be an upper bound on the maximum degree of $G$ at any point in time. We now show how our data structure can be used to explicitly maintain the tentative coloring $\tilde \chi^{(t)}$ and set of failed edges $F^{(t)}$ defined by running Algorithm~\ref{app:alg:nibble} on graph $G^{(t)}$ with parameters $\Delta$ and $\epsilon$. Let $i_e$ and $c_e$ denote the round and color sequence sampled in advance for potential edge $e \in \binom{V}{2}$.

We first remark that, if we have that our data structure currently stores the graph $G^{(t-1)}$, each edge $e$ receives the color sequence $c_e$ and round $i_e$, and that the color index $\ell_e$ of each edge $e$ equals the color index $\ell_e^{(t-1)}$ defined by Algorithm~\ref{app:alg:nibble}, then clearly the tentative coloring $\tilde \chi$ and the set of failed edges $F$ maintained by our data structure equal $\tilde \chi^{(t-1)}$ and $F^{(t-1)}$. Hence, in order to be able to maintain $\tilde \chi^{(t-1)}$ and $F^{(t-1)}$, given the edge $e^\star$ that is either inserted or deleted during the $t^{th}$ update, we need to appropriately update the graph maintained by our data structure (by inserting or deleting $e^\star$) and then update the color indices so that $\ell_e = \ell_e^{(t)}$ for all edges $e$. Given that our data structure is in the state described above, we now show how to do this efficiently.

\medskip
\noindent \textbf{Updating the data structure.} When an edge $e^\star$ is inserted into the graph $G$, we run Algorithm~\ref{app:alg:insertion update} on input $e^\star$, which passes $e^\star$ to the data structure along with its round $i_{e^\star}$ and color sequence $c_{e^\star}$. When an edge $e^\star$ is deleted from the graph $G$, we run Algorithm \ref{app:alg:deletion update} on input $e^\star$, which removes $e^\star$ from the data structure. In both cases, the algorithm then determines if $e^\star$ is dirty and then passes it to Algorithm \ref{app:alg:update implementation} if this is the case. Algorithm \ref{app:alg:update implementation} takes the set $A^{(t)}_{i_{e^\star}}$ and round $i_{e^\star}$ as inputs, where $i_{e^\star}$ is the first round containing dirty edges, and propagates the changes in the tentative coloring through the rounds. After the update is complete, we have that $\ell_e = \ell_e^{(t)}$ for all edges $e$ in the graph.

\begin{algorithm}[H]
\caption{\textsc{Insertion-Update}$(e^\star)$}\label{app:alg:insertion update}
\begin{algorithmic}[1]
    \State $\tilde \chi^{\textsc{prev}}(e^\star) \leftarrow \perp$
    \State $\textsc{Insert}(e^\star, i_{e^\star}, c_{e^\star})$
    \State $X \leftarrow \{\textsc{Reset-Color}(e^\star)\}$
    \State \textsc{Propagate-Changes}$(X, i_{e^\star})$
\end{algorithmic}
\end{algorithm}

\begin{algorithm}[H]
\caption{\textsc{Deletion-Update}$(e^\star)$}\label{app:alg:deletion update}
\begin{algorithmic}[1]
    \State $\tilde \chi^{\textsc{prev}}(e^\star) \leftarrow \tilde \chi(e^\star)$
    \State $\textsc{Delete}(e^\star)$
    \State $X \leftarrow \varnothing$
    \If{$\tilde \chi^{\textsc{prev}}(e^\star) \neq \perp$}
        \State $X \leftarrow X \cup \{e^\star\}$
    \EndIf
    \State \textsc{Propagate-Changes}$(X, i_{e^\star})$
\end{algorithmic}
\end{algorithm}

\begin{algorithm}[H]
\caption{\textsc{Propagate-Changes}$(X_{i'}, i')$}\label{app:alg:update implementation}
\begin{algorithmic}[1]
    \For{$i = i' + 1,\dots,T:$}
        \State $X'_i \leftarrow \bigcup_{e \in X_{i-1}} \Psi_{e,i}(\tilde \chi(e)) \cup \Psi_{e,i}(\tilde \chi^{\textsc{prev}}(e))$\;
        \State $\tilde \chi^{\textsc{prev}}(e) \leftarrow \tilde \chi(e)$ for all $e \in X_i'$\;
        \State $X_{i} \leftarrow X_{i-1}$\;
        \For{$e \in X'_i$}
            \State $X_{i} \leftarrow X_{i} \cup \textsc{Reset-Color}(e)$\;
        \EndFor        
    \EndFor
\end{algorithmic}
\end{algorithm}

\noindent
Here we let $\Psi_{e,i}(c)$ denote the set $\Psi_{u,i}(c) \cup \Psi_{v,i}(c)$, where $e = (u,v)$.

\subsubsection{Correctness}

Suppose that we call \textsc{Insertion-Update}$(e^\star)$ (resp. \textsc{Deletion-Update}$(e^\star)$) to update the data structure after the insertion (resp. deletion) of the edge $e^\star$ at time $t$. The following lemmas describe the behavior of our algorithm. We denote by $\ell_e$ and $P_i(u)$ the color indices and palettes maintained by our algorithm and by $\ell^{(t)}_e$ and $P^{(t)}_i(u)$ the color indices and palettes defined by Algorithm \ref{app:alg:nibble} on input $G^{(t)}$. Our objective is to show that $\ell_e = \ell^{(t)}_e$ for all edges $e$ \emph{after} the update is complete. Recall that we assume $\ell_e = \ell^{(t-1)}_e$ for all edges $e$ \emph{before} we the perform the update.

\begin{lem}\label{app:lem: correctness 1} Given any $X \subseteq S^{(t)}_{< i}$, we have that
    \[\Gamma_i^{(t)}( X ) \cup \Lambda_i^{(t)}( X) \subseteq \bigcup_{e \in X} \Psi_{e,i}(\tilde \chi^{(t)}(e)) \cup \Psi_{e,i}(\tilde \chi^{(t-1)}(e)).\]
\end{lem}

\begin{proof}
    Let $e \in \Gamma_i^{(t)}(X)$. Then there exists some $f \in X$ such that $\tilde \chi^{(t)}(f) = \tilde \chi^{(t-1)}(e)$, so $\tilde \chi^{(t)}(f) \in c_e$ and hence $e \in \Psi_{f,i}(\tilde \chi^{(t)}(f))$. Now let $e \in \Lambda_i^{(t)}(X)$. Then there exists some $f \in X$ such that $\tilde \chi^{(t-1)}(f) = \tilde \chi^{(t)}(e)$, so $\tilde \chi^{(t-1)}(f) \in c_e$ and hence $e \in \Psi_{f,i}(\tilde \chi^{(t)}(f))$.
\end{proof}

\begin{lem}\label{app:lem: correctness 2}
    For all $i \in [T]$, we have that if $\ell_e = \ell^{(t)}_e$ for all $e \in S^{(t)}_{< i}$, then running $\textsc{Recolor-Edge}$ on any $e \in S^{(t)}_{i}$ will set $\ell_e = \ell^{(t)}_e$.
\end{lem}

\begin{proof}
    Since $\ell_e = \ell^{(t)}_e$ for all $e \in S^{(t)}_{< i}$, it follows that $\tilde \chi(e) = \tilde \chi^{(t)}(e)$ for all $e \in S^{(t)}_{< i}$, so $P_i(u) = P_i^{(t)}(u)$ for all $u \in V$. Given any edge $e \in S^{(t)}_i$, we then have that $P_i(e) = P^{(t)}_i(e)$, so when we run $\textsc{Recolor-Edge}(e)$ we set the color index $\ell_e$ to the smallest $\ell$ such that $c_e(\ell) \in P^{(t)}_i(e)$ (or $0$ if no such $\ell$ exists), which is the exact definition of $\ell^{(t)}_e$.
\end{proof}

\begin{lem}\label{app:lem: correctness 3}
    For all $i \geq i_{e}$, we have that $X_i = A^{(t)}_{\leq i}$. Furthermore, after our update procedure terminates, we have that $\ell_e = \ell^{(t)}_e$ for all $e \in E^{(t)}$.
\end{lem}

\begin{proof}
    We prove this by induction. For all $i \geq i_{e^\star} + 1$, we show that the following are true at the start of the $i$\textup{th} iteration of the \textbf{\textup{for}} loop in Algorithm \ref{app:alg:update implementation} (where we start from iteration $i_{e^\star}+1$):
    \begin{enumerate}
        \item $X_{i-1} = A^{(t)}_{< i}$, and
        \item $\ell_e = \ell^{(t)}_e$ for all $e \in S^{(t)}_{< i}$.
    \end{enumerate}
    \textbf{Base case.} We begin by showing that these conditions all hold for $i = i_{e^\star} + 1$. Let $i^\star = i_{e^\star}$. We first show that $X_{i^\star} = A^{(t)}_{\leq i^\star}$. In the event that the $t^{th}$ update is an insertion, we have that $X_{i^\star}$ is empty if $\tilde \chi^{(t)}(e^\star) = \perp$ (note that \textsc{Reset-Color}$(e^\star)$ returns $e^\star$ if and only if this is not the case), in which case the set $A^{(t)}_{\leq i^\star}$ is also empty, and $X_{i^\star} = \{e^\star\}$ if $\tilde \chi^{(t)}(e^\star) \neq \perp$, in which case $A^{(t)}_{\leq i^\star} =\{e^\star\}$. Similarly, in the event that the $t^{th}$ update is an deletion, we have that $X_{i^\star}$ is empty if $\tilde \chi^{(t-1)}(e^\star) = \perp$, in which case $A^{(t)}_{\leq i^\star}$ is also empty, and $X_{i^\star} = \{e^\star\}$ if $\tilde \chi^{(t-1)}(e) \neq \perp$, in which case $A^{(t)}_{\leq i_e} =\{e^\star\}$. To see that $\ell_e = \ell^{(t)}_e$ for all $e \in S^{(t)}_{\leq i^\star}$ at the start of iteration $i^\star + 1$, note that $\ell_e = \ell^{(t)}_e$ for all $e \in S^{(t)}_{< i^\star}$ when we call \textsc{Reset-Color}$(e^\star)$ (since there are no dirty edges in the first $i^\star - 1$ rounds). Hence, by Lemma~\ref{app:lem: correctness 2}, we have that $\ell_{e^\star} = \ell^{(t)}_{e^\star}$ at the start of iteration $i^\star + 1$. As this is the only edge in round $i^\star$ that can change its tentative color, the claim follows.

    \medskip
    \noindent \textbf{Inductive step.} Suppose that the inductive hypothesis holds for $i_e + 1 \leq i \leq T$. We now show that it also holds for $i+1$. By Lemma~\ref{app:lem: correctness 1} and Corollary~\ref{app:lem:fact 4}, we can see that
    \[X_i' \supseteq \Gamma_i^{(t)}( X_{i-1} ) \cup \Lambda_i^{(t)}( X_{i-1}) =  \Gamma_i^{(t)}( A^{(t)}_{<i} ) \cup \Lambda_i^{(t)}( A^{(t)}_{<i}) = A_{i}^{(t)}. \]
    Since the algorithm scans through all of the edges in $X_i'$ and calls $\textsc{Recolor-Edge}$ on each of them, and we have that $\ell_e = \ell^{(t)}_e$ for all $e \in S^{(t)}_{< i}$, it follows by Lemma~\ref{app:lem: correctness 2} that we have $\ell_e = \ell_e^{(t)}$ for all $e \in S^{(t)}_{\leq i}$ at the end of the iteration
    (note that $X_i'$ contains all of the edges at round $i$ that change their color indices during this update).
    The algorithm places all of the edges $e \in X_i'$ that change their color index $\ell_e$ after calling \textsc{Recolor-Edge}$(e)$ into $X_i$ along with the edges in $A_{<i}^{(t)}$. Since $X_i' \subseteq A_{i}^{(t)}$, it follows that $X_i = A_{\leq i}^{(t)}$ at the end of the iteration.
\end{proof}

\subsubsection{Update Time Analysis}

\begin{lem}\label{app:lem: update time 1}
    For all $i > i_{e^\star}$, $\mathbb E[|X'_i|] \leq (32/\epsilon) \log(1/\epsilon) \cdot |X_{i-1}|$.
\end{lem}

\begin{proof}
    For all $i > i_e$,
    \[X'_i = \bigcup_{e \in X_{i-1}} \Psi_{e,i}(\tilde \chi(e)) \cup \Psi_{e,i}(\tilde \chi^{\textsc{prev}}(e)),\]
    which implies that
    \[|X'_i| \leq \sum_{e \in X_{i-1}} \left|   \Psi_{e,i}(\tilde \chi(e)) \right| + \left| \Psi_{e,i}(\tilde \chi^{\textsc{prev}}(e)) \right|.\]
    Let $u \in V$, $i \in [T]$, and $c \in [(1 + \epsilon)\Delta]$. Then for all $e \in N_i(u)$ we have that
    \[\Pr[e \in \Psi_{u,i}(c)] = \Pr[c \in c_e] \leq \frac{K}{(1 + \epsilon)\Delta}\]
    since $c_e$ is a sequence of $K$ colors sampled independently and uniformly at random from $[(1 + \epsilon)\Delta]$. By linearity of expectation we get that $\mathbb E[|\Psi_{u,i}(c)|] \leq |N_i(u)|\cdot K/((1 + \epsilon)\Delta)$. Taking expectations on both sides and noting that $\mathbb E[|N_i(u)|] \leq \epsilon(1 - \epsilon)^{i-1}\Delta$ (see Lemma \ref{lem:degree concentration}), we get that
    \[\mathbb E[|\Psi_{u,i}(c)|] \leq \mathbb E[|N_i(u)|]\cdot \frac{K}{(1 + \epsilon)\Delta} \leq \epsilon K = \frac{8}{\epsilon}\log\frac{1}{\epsilon}.\]
    It follows that
    \[\mathbb E[|X'_i|] \leq \sum_{e \in X_{i-1}} \mathbb E[\left| \Psi_{e,i}(\tilde \chi(e)) \right|] + \mathbb E[\left| \Psi_{e,i}(\tilde \chi^{\textsc{prev}}(e)) \right|] \leq |X_{i-1}| \cdot \frac{32}{\epsilon}\log\frac{1}{\epsilon}.\]
\end{proof}

\begin{lem}
    Algorithms \ref{app:alg:insertion update} and \ref{app:alg:deletion update} run in $O(\log^4(1 /\epsilon)/\epsilon^5 \cdot \mathbb E[|A^{(t)}|])$ expected time.
\end{lem}

\begin{proof}
    Algorithms \ref{app:alg:insertion update} and \ref{app:alg:deletion update} run in $O(KT)$ time excluding the call to \textsc{Propagate-Changes}$(X, i_{e^\star})$, since the calls to $\textsc{Insert}$ and $\textsc{Delete}$ take $O(K)$ time, the call to $\textsc{Reset-Color}(e^\star)$ takes $O(KT)$ time, and the rest of the operations run in $O(1)$ time. The $i^{th}$ iteration of Algorithm~\ref{app:alg:update implementation} (where we start from iteration $i_{e^\star}+1$) takes time $O(|X_{i-1}| + KT \cdot |X'_i|)$. Since we can return the sets $\Psi_{u,i}(c)$ in time proportional to their size, Line 2 runs in $O(|X_i'| + |X_{i-1}|)$ time. Lines 3-4 also run in time $O(|X_i'| + |X_{i-1}|)$. Lines 5-6 take $O(KT \cdot |X'_i|)$. It follows that the running time of these algorithms is
    \[O(KT) + \sum_{i = i^\star + 1}^T O(|X_{i-1}| + KT \cdot |X'_i|).\]
    Taking expectations, applying Lemma~\ref{app:lem: update time 1}, and then taking expectations again, it follows that the expected running time of these algorithms is
    \[O(KT) + \sum_{i = i^\star + 1}^T O \left(\frac{1}{\epsilon^4}\log^3 {\frac{1}{\epsilon}} \cdot \mathbb E[|X_{i-1}|] \right) = O \left(\frac{1}{\epsilon^4}\log^3 {\frac{1}{\epsilon}} \right) \cdot \sum_{i = i^\star + 1}^T \mathbb E[|X_{i-1}|].\]
    Applying Lemma~\ref{app:lem: correctness 3}, we can upper bound this by
    \[O\left(\frac{1}{\epsilon^4}\log^3 {\frac{1}{\epsilon}} \right) \cdot T \cdot \mathbb E[|A^{(t)}_{\leq T}|] \leq O\left(\frac{1}{\epsilon^5}\log^4 {\frac{1}{\epsilon}} \cdot \mathbb E[|A^{(t)}|] \right).\]
\end{proof}

\subsection{Proof of Lemma~\ref{app:thm:fast greedy 3}}\label{sec:proving greedy}

Let $0 < \delta \leq 1$ be a constant. Then we have the following lemma.

\begin{lem}\label{app:thm:fast greedy}
    There exists a dynamic data structure that, given a dynamic graph $G = (V,E)$ that undergoes edge insertions and deletions, can explicitly maintain a $(2 + \delta)\Delta(G)$-edge coloring of $G$, has an expected update time of $O(1/\delta)$, and only changes the colors of $O(1)$ many edges per update. 
\end{lem}

\begin{proof}
    Our data structure maintains the graph $G$ in a manner that allows edges to be inserted and deleted in $O(1)$ time (see Section~\ref{app:sec:key data structure}), and the coloring $\chi : E \longrightarrow [(2 + \delta)\Delta(G)]$ using a hashmap, allowing us to retrieve and set a color $\chi(e)$ in $O(1)$ time for any $e \in E$. Similarly, it maintains a hashmap $\psi : V \times [(2 + \delta)\Delta(G)] \longrightarrow E$ that maps nodes $u$ and a color $c$ to the edge $e$ incident on $u$ with $\chi(e) = c$, if such an edge exists. This algorithm maintains the invariant that each edge $e=(u,v)$ receives a color from $[(2 + \delta)\max\{\deg(u), \deg(v)\}]$.

    \medskip
    \noindent \textbf{Inserting an edge.} When we insert an edge $e=(u,v)$ into the graph $G$, we sample a color $c$ independently and u.a.r. from $[(2 + \delta)\max\{\deg(u), \deg(v)\}]$. Let $d^\star := \max\{\deg(u), \deg(v)\}$. The probability that this color is available at $e$, i.e. is contained in the palette $P(e) = P(u) \cap P(v)$ where $P(u) = [(2 + \delta)d^\star] \setminus \chi(N(u))$, is at least
    \[\Pr[c \in P(e)] = \frac{|P(e)|}{(2 + \delta)d^\star} = \frac{|P(u)| + |P(v)| - |P(u) \cup P(v)|}{(2 + \delta)d^\star}\]
    \[\geq \frac{(2 + \delta)d^\star - \deg(u) + (2 + \delta)d^\star - \deg(v) - (2 + \delta)d^\star}{(2 + \delta)d^\star} \geq \frac{\delta}{(2 + \delta)} \geq \frac{\delta}{3}.\]
    Using the hashmaps $\psi_u$ and $\psi_v$ we can check whether $c \in P(e)$ in $O(1)$ time. If $c \notin P(e)$, we sample another color in the same way, and repeat until we find a color $c \in P(e)$. Since these events that the sampled colors are in $P(e)$ are independent and all succeed with probability at least $\delta/3$, it follows that the expected number of colors that we have to sample before finding one in $P(e)$ is $3/\delta$. Hence, it takes us $O(1/\delta)$ time to find a color $c \in P(e)$ in expectation. Once we have found such a color, we set $\chi(e) \leftarrow c$, $\psi_u(c) \leftarrow e$, and $\psi_v(c) \leftarrow e$. Since the degrees of nodes cannot decrease during an edge insertion, the invariant is still satisfied.

    \medskip
    \noindent \textbf{Deleting an edge.}
    When we delete an edge $e = (u,v)$ from the graph $G$, we first set $\psi_u(\chi(e)) \leftarrow \textsc{null}$ and $\psi_v(\chi(e)) \leftarrow \textsc{null}$, and then set $\chi(e) \leftarrow \textsc{null}$. However, the degrees of $u$ and $v$ decrease by one, so there might be edges that no longer satisfy the invariant. In particular, any edge $f$ incident on $u$ that no longer satisfies the invariant will have $\chi(f) \in [(2 + \delta)(\deg(u) + 1)] \setminus [(2 + \delta)\deg(u)]$. Since there are at most $3$ such colors, we can use the hashmap $\psi_u$ to identify the edges incident on $u$ that receive one of these colors in $O(1)$ time. We then uncolor these edges, while ensuring to appropriately update the hashmaps, and recolor them using the same strategy as outlined above for an edge insertion. We do the same thing at node $v$. In total, we recolor at most $6$ edges in $G$ to ensure that we still satisfy the invariant, taking $O(1/\delta)$ time in expectation.
\end{proof}

\section{Implementing our Static Algorithm (Full Version)}\label{sec:app:datastructstatic}

We now show how to implement our static algorithm directly in order to obtain better performance than what follows immediately from our dynamic algorithm. We first show how to get linear time in expectation, and then expand on this to get linear time with high probability. The main result in this appendix is \Cref{thm:linear static whp}, which is restated below.

\begin{thm}
    There exists an algorithm that, given a graph $G$ with maximum degree $\Delta$ such that $\Delta \geq (100 \log n/\epsilon^4)^{(30/\epsilon) \log (1/\epsilon)}$, returns a $(1 + 61\epsilon)\Delta$-edge coloring of $G$ in $O(m \log(1/\epsilon)/\epsilon^2)$ time with probability at least $1 - O(1/n^6)$.
\end{thm}

\subsection{Linear Time in Expectation}

It follows immediately from Corollary~\ref{thm:main 2} that our static algorithm can be implemented to run in $O(m\log^4(1/\epsilon)/\epsilon^9)$ expected time. We now show how to implement our static algorithm more directly using the data structure from \Cref{sec:app:datastructs} in order to get the following theorem.

\begin{thm}
    There exists an algorithm that, given a graph $G$ with maximum degree $\Delta$ such that $\Delta \geq (100 \log n/\epsilon^4)^{(30/\epsilon) \log (1/\epsilon)}$, returns a $(1 + 61\epsilon)\Delta$-edge coloring of $G$ in $O(m \log(1/\epsilon)/\epsilon^2)$ expected time.
\end{thm}

\begin{proof}
    We begin by using our data structure as a black box in order to show how our static algorithm can be implemented to run in $O(KTm)$ expected time. We then show how a white-box application can improve this to $O(Km)$ expected time.

    Given the graph $G=(V,E)$, we begin by splitting the edges $E$ into $\mathcal E_1,\dots, \mathcal E_\eta$ as defined in \Cref{app:alg:split}. We do this by creating lists to store each of the $\mathcal E_j$ and scan through all of the edges in $e \in E$, assigning each to one of the $\mathcal E_j$ independently and u.a.r., taking $O(m)$ time. 
    If we can now show how to implement \Cref{app:alg:nibble} so that it runs in expected time $g(m')$ on an input graph with $m'$ edges, it follows that we can compute the failed edges and the tentative coloring in expected time $\sum_{j \in [\eta]}g(|\mathcal E_j|)$. By then noticing that Lemma~\ref{app:thm:fast greedy} immediately implies an expected linear time greedy algorithm, we get that the static algorithm runs in expected time $O(m) + \sum_{j \in [\eta]}g(|\mathcal E_j|)$. If the function $g$ is linear and $g(x) = \Omega(x)$ (as it will be in the following cases), this then gives an expected running time of $O(g(m))$. This algorithm produces a $(1 + 61\epsilon)\Delta$-edge coloring with probability at least $1 - O(1/n^6)$. If it uses more than $(1 + 61\epsilon)\Delta$ many colors, we can afford to keep rerunning it using fresh randomness until we use at most $(1 + 61\epsilon)\Delta$ many colors, without increasing the asymptotic expected running time. (see \Cref{sec:app:datastructs} for a precise description of how to implement the resampling efficiently).

    \medskip
    \noindent \textbf{Black-box application of our data structure.} Now we show how to implement \Cref{app:alg:nibble} so that it runs in $O(KTm)$ expected time given a graph $G = (V,E)$ with $m$ edges. We begin by splitting the set $E$ into $S_1,\dots,S_{T+1}$ by sampling a round $i_e$ for each edge $e \in E$ and then placing $e$ into a list containing the edges that are colored at round $i_e$, again taking $O(m)$ time in total. We then proceed to initialize our data structure, which takes $O(1)$ time. For $i = 1 \dots T$, we then insert all of the edges in $S_i$ into the data structure (in any order) so that they are assigned a color index of $0$ and hence left uncolored, and then call $\textsc{Reset-Color}$ on each edge in $S_i$ (again, in any order). By \Cref{app:lem: correctness 2}, we have (by induction) that this correctly computes the color indices of all the edges in $S_i$ as defined by the color sequences and rounds that we generate for the edges. Hence, we also compute the correct tentative colors and set of failed edges. Since $\textsc{Reset-Color}$ runs in time $O(KT)$, and it takes $O(K)$ time to insert an edge into the data structure, it follows that the algorithm runs in time $O(KTm)$ since must do this for each edge exactly once.

    \medskip
    \noindent \textbf{White-box application of our data structure.} In order to improve this to $O(Km)$, we notice that in this setting we can implement the query \textsc{Edge-Palette-Query} to run in in $O(1)$ expected time, improving the runtime of \textsc{Reset-Color} to $O(K)$, which gives the result.
    We do this by noticing that we can maintain a hashmap $\phi' :  V \times \mathcal C \longrightarrow 2^{N(u)}$ where $\phi'_{u}(c) = \{e \in N(u) \, | \, \tilde \chi(e) = c\}$ and each $\phi'_{u}(c)$ is implemented in the same way as the $\phi_{u,i}(c)$. Since $\phi'_u(c) = \bigcup_{i =1}^T \phi_{u,i}(c)$, we can maintain the map $\phi'$ by appropriately updating the set $\phi'_{u}(c)$ every time we update one the sets $\phi_{u,i}(c)$, incurring only $O(1)$ overhead. Since this $O(1)$ overhead does not change the asymptotic behavior of our data structure, this does not change the asymptotic running time of any of the updates or queries.
    When our algorithm now makes a call to \textsc{Edge-Palette-Query} after calling \textsc{Reset-Color} on an edge $e=(u,v)$ appearing in round $i$, we note that no edges in $S_{> i}$ have been inserted into the graph yet. Hence, we have that a color $c$ is contained in $P_i(e)$ iff $|\phi'_{u}(c)| - |\phi_{u,i}(c)| = 0$ and $|\phi'_{v}(c)| - |\phi_{v,i}(c)| = 0$, which we can check in $O(1)$ time.
\end{proof}

\subsection{Linear Time with High Probability}

In order to obtain an algorithm that runs in $O_\epsilon(m)$ time with high probability, we need to obtain concentration. Our first obstacle is our use of hashmaps. We need to argue that we can implement these hashmaps so that not only can we handle insert, delete, and query operations in $O(1)$ expected time, but also so that these operations all take $O(1)$ time with high probability. For this, we use the following lemma which follows from \cite{dietzfelbinger1990new}.

\begin{lem}[Theorem 5.5, \cite{dietzfelbinger1990new}]\label{lem:good hashmap}
    There exists a dynamic dictionary that, given a parameter $k$, can handle $k$ insertion, deletion, and query operations, uses $O(k)$ space, and takes $O(1)$ worst-case time per operation with probability at least $1 - O(1/k^7)$.
\end{lem}

\noindent The second obstacle comes from our $O(\Delta)$ greedy algorithm, where we can make an unbounded number of queries to a hashmap. By applying the following concentration inequality for sums of geometric random variables given in \cite{Brown11}, we show that we only perform $O(m)$ many queries to the hashmap with high probability.

\begin{lem}\label{lem:geometric concentration}
    Let $X_1,\dots,X_n$ be $n$ independent geometric random variables with success probability $p$, and let $X = \sum_i X_i$. Then, for $0 < \epsilon \leq 1$, we have that
    \[\Pr[X > (1 + \epsilon)\mathbb E[X]] \leq \exp{ \left(-\frac{\epsilon^2}{4} n \right)}. \]
\end{lem}

\noindent By combining these tools, we can analyze our static algorithm more carefully, and get the following result.

\begin{thm}\label{thm:linear static whp}
    There exists an algorithm that, given a graph $G$ with maximum degree $\Delta$ such that $\Delta \geq (100 \log n/\epsilon^4)^{(30/\epsilon) \log (1/\epsilon)}$, returns a $(1 + 61\epsilon)\Delta$-edge coloring of $G$ in $O(m \log(1/\epsilon)/\epsilon^2)$ time with probability at least $1 - O(1/n^6)$.
\end{thm}

\begin{proof}
    We begin by showing how we can use the dynamic dictionary from \Cref{lem:good hashmap} to implement the hashmaps in \Cref{app:alg:nibble} so that each insertion, deletion, and query made by the algorithm which we described in the preceding section takes $O(1)$ time with probability at least $1 - O(1/m^7)$, where $m$ is the number of edges in the input graph $G=(V,E)$. This will immediately imply that we can implement \Cref{app:alg:nibble} to run in $O(Km)$ time with probability at least $1 - O(1/m^7)$. We first note that our data structure uses the maps $\phi$, $\phi'$, $\textsc{Round}$, $\textsc{Color-Index}$, $\textsc{Color-Sequence}$, and $\textsc{Failed}$. Our implementation of \Cref{app:alg:nibble} does not need to use the map $\Psi$, so we can ignore all operations performed on this map in this context. As for the sets $\phi_{u,i}(c)$, note that they can only increase in size as we run the algorithm. Since the algorithm never accesses the elements in $\phi_{u,i}(c)$ after its size exceeds 2, we do not need to store the set once it becomes large enough, and instead we just store its size after this point. Hence, we do not need to implement them as hashmaps and can implement all operations on the set $\phi_{u,i}(c)$ in $O(1)$ worst-case time. The same applies to the sets $\phi'_{u}(c)$. Since we never need to access the elements in $\phi'_{u}(c)$, it is sufficient to simply keep track of its size, which can easily be achieved with a counter that can be updated in $O(1)$ time.
    
    For each edge $e \in E$, we make one call to $\textsc{Insert}(e)$ and one call to $\textsc{Reset-Color}(e)$. It follows that, throughout the entire run of \Cref{app:alg:nibble}, we perform $O(Km)$ many operations on the map $\phi'$ and $O(m)$ many operations on each of the other maps. Hence, we can implement each of these maps using the dynamic dictionary in \Cref{lem:good hashmap}, initializing each one with a parameter of size $O(Km)$. It follows from a union bound that all the operations performed on all the hashmaps throughout the run of \Cref{app:alg:nibble} takes $O(1)$ time with probability at least $1 - O(1/(Km)^7) \geq 1 - O(1/n^7)$. By the arguments in the preceding section, it follows that the total time taken to subsample the graph and handle each call to \Cref{app:alg:nibble} for each of the subsampled graphs is $O(Km)$ with probability at least $1 - O(\eta/n^7) \geq 1 - O(1/n^6)$. 

    In order to implement our greedy algorithm so that it runs in $O(m)$ time with high probability, we again implement the hashmap $\psi$ used by the algorithm using the dynamic dictionary in \Cref{lem:good hashmap}, passing a sufficiently large parameter of order $O(m)$. However, it is not immediately clear that the algorithm performs at most $O(m)$ operations on the map $\psi$. We can observe that the insertion of the edge $e$ into the greedy algorithm leads to $O(X_e)$ many operations on the map $\psi$, where $X_e$ is a random variable denoting the number of colors sampled by the algorithm while trying to sample a color from the palette of edge $e$. Clearly, the value of $X_e$ is $O(1)$ in expectation. It follows that, in expectation, we perform $\mathbb E[O(\sum_e X_e)] \leq O(m)$ many operations on the map $\psi$. In order to establish concentration, we can observe that $\{X_e\}_e$ is a collection of independent geometric random variables. Thus, we can apply the concentration inequality from \Cref{lem:geometric concentration} to get that our algorithm performs at most $O(m)$ many operations on $\psi$ with probability at least $1 - e^{-\Theta(m)}$. Conditioned on this event, all of the operations performed on our hashmap run in $O(1)$ worst-case time with probability at least $1 - O(1/m^7)$. Hence, the greedy algorithm runs in $O(m)$ time with probability at least $1 - O(1/m^7)$.

    Putting everything together, we get that our static algorithm runs in $O(Km)$ time with probability at least $1 - O(1/n^6)$. Since the algorithm returns a $(1 + 61 \epsilon)\Delta$-edge coloring with probability at least $1 - O(1/n^6)$, the lemma follows.
\end{proof}

\section{Concentration of Measure}

In this appendix, we state the probabilistic tools that we use to establish concentration of measure throughout our paper. This appendix is essentially a subset of Appendix E in \cite{BhattacharyaGW21}.

\subsection{Concentrartion Bounds}

We now introduce some standard concentration bounds for independent random variables. The proofs of all of these bounds can be found in \cite{DubhashiP09-book}.

\begin{prop}[Chrenoff Bounds]\label{prop:chernoff}
    Let $X$ be the sum of $n$ mutually independent indicator random variables $X_1,\dots,X_n$. Then, for any $\mu_L \leq \mathbb E[X] \leq \mu_H$, for all $\epsilon > 0$, we have that
    \[\Pr[X > (1 + \epsilon) \mu_H] \leq \exp{ \left(-\frac{\epsilon^2}{3} \mu_H \right)},\]
    \[\Pr[X < (1 - \epsilon) \mu_L] \leq \exp{ \left(-\frac{\epsilon^2}{2} \mu_L \right)}.\]
\end{prop}

\begin{prop}[Hoeffding Bounds]\label{prop:hoeffding}
    Let $X$ be the sum of $n$ mutually independent indicator random variables $X_1,\dots,X_n$. Then, for all $t > 0$, we have that
    \[\Pr[X >  \mathbb E[X] + t] \leq e^{-2t^2/n},\]
    \[\Pr[X <  \mathbb E[X] - t] \leq e^{-2t^2/n}.\]
\end{prop}

\begin{definition}[Lipschitz Functions]\label{def:lippy}
    Consider $n$ sets $A_1,\dots,A_n$ and a real valued function $f : A_1,\dots,A_n \longrightarrow \mathbb R$. The function $f$ satisfies the Lipschitz property with constants $d_1,\dots,d_n$ if and only if $|f(x) - f(y)| \leq d_i$ whenever $x$ and $y$ differ only in the $i$\textup{th} coordinate, for all $i \in [n]$.
\end{definition}

\begin{prop}[Method of Bounded Differences]\label{lem:bounded diff}
    If $f$ satisfies the Lipschitz property with constants $d_1,\dots,d_n$, and $X_1,\dots,X_n$ are independent random variables, then, for all $t > 0$, we have that
    \[\Pr[f < \mathbb E[f] + t] \leq \exp{\left(\frac{2t^2}{d}\right)},\]
    \[\Pr[f > \mathbb E[f] - t] \leq \exp{\left(\frac{2t^2}{d}\right)},\]
where $d = \sum_{i}d_i^2$.
\end{prop}

\subsection{Negatively Associated Random Variables}

We will sometimes need to get concentration around the sums for \emph{negatively associated random variables}. We will need the following tools.

\begin{definition}[Negatively Associated Random Variables, \cite{DubhashiP09-book, joag1983negative}]\label{def:NA}
    We say that the random variables $X_1,\dots,X_n$ are \emph{negatively associated} (NA), if any two monotone increasing functions $f$ and $g$ defined on disjoint subsets of the variables in $\{X_i\}_i$ are negatively correlated. That is,
    \[\mathbb E[f \cdot g] \leq \mathbb E[f] \cdot \mathbb E[g].\]
\end{definition}

\noindent Independent random variables are trivially NA.

\begin{prop}[0-1 Principle, \cite{dubhashi1996balls}]\label{prop:0-1 NA}
    Let $X_1,\dots,X_n \in \{0,1\}$ be binary random variables such that $\sum_{i}X_i \leq 1$. Then $X_1,\dots,X_n$ are NA.
\end{prop}

\begin{definition}[Permutation Distribution]\label{def:perm dist}
    Let $x_1,\dots,x_n$ be $n$ values and let $X_1,\dots,X_n$ be random variables taking on all permutations of $(x_1,\dots,x_n)$ with equal probability. Then we call the collection of random varaibles $X_1,\dots,X_n$ a permutation distribution.
\end{definition}

\begin{prop}[\cite{joag1983negative}]\label{prop:perm dist NA}
    Collections of random variables that form permutation distributions are NA.
\end{prop}

\begin{prop}[NA Closure Properties, \cite{DubhashiP09-book}]\label{prop:closure NA} \quad
    \begin{itemize}
        \item \textup{\textbf{Closure under products.}} If $X_1,\dots,X_n$ and $Y_1,\dots,Y_m$ are two independent families of random variables that are separately NA, then $X_1,\dots,X_n,Y_1,\dots,Y_m$ is also NA.
        \item \textup{\textbf{Disjoint Monotone Aggregation.}} If $X_1,\dots,X_n$ are NA, and $f_1,\dots,f_k$ are monotone (either all increasing or all decreasing) functions defined on disjoint subsets of the random variables, then $f_1(X_1,\dots,X_n),\dots,f_k(X_1,\dots,X_n)$ are NA.
    \end{itemize}
\end{prop}

\noindent The Chernoff-Hoeffding bounds from Propositions~\ref{prop:chernoff} and \ref{prop:hoeffding} extend to the case where the random variables are NA \cite{DubhashiP09-book}.

\end{document}